\documentclass[a4paper,12pt]{article}
\usepackage{geometry}
\usepackage{amsmath}
\usepackage{amsthm}
\numberwithin{equation}{section}
\numberwithin{figure}{section}
\usepackage{authblk}
\usepackage{hyperref}
\usepackage{bm}
\usepackage{cite}
\usepackage{mathrsfs}
\usepackage{amssymb}
\usepackage{amsfonts}
\usepackage{graphicx}
\usepackage{subcaption}
\usepackage{tikz}
\usepackage[title]{appendix}
\usepackage{mathtools}

\usetikzlibrary{intersections}
\allowdisplaybreaks

\theoremstyle{plain}
\newtheorem{prop}{Proposition}[section]

\newtheorem{theo}{Theorem}[section]
\newtheorem{lem}{Lemma}[section]

\newtheorem{assu}{Assumption}[section]
\newtheorem{rhp}{RH problem}[section]

\newtheorem{rhp-Dbar}{RH-$\bar{\partial}$ problem}[section]

\def\rd{{\rm d}}
\def\re{{\rm e}}
\def\ri{{\rm i}}

\def\rim{{\rm {Im}}\,}
\def\rre{{\rm {Re}}\,}

\allowdisplaybreaks
\newtheorem{remark}{Remark}[section]

\title{{\bf Long-time asymptotics of the Newell equation on the line}
}

\author[$^{\dagger}$]{Deng-Shan Wang}
\author[$^{\dagger}$]{Yingmin Yang\thanks{Corresponding author: yangym@mail.bnu.edu.cn}}

\affil[$^{\dagger}$]{\small School of Mathematical Sciences,
	Beijing Normal University, Beijing 100875, China}

\date{}

\begin{document}
\maketitle
\begin{abstract}
In 1978, A. C. Newell [SIAM J. Appl. Math. 35(4) (1978) 650-664] proposed an exactly solvable model called Newell equation, which simulates the investigation of significant interaction mechanism between long and short waves. Nearly fifty years have passed, yet the long-time asymptotics of the Newell equation remains an open problem to date, with no results reported. In this work, the long-time asymptotic behaviors of the solutions to this model under Schwartz class initial conditions are studied by using the Riemann-Hilbert formulation. Through direct and inverse scattering analysis, the corresponding Riemann-Hilbert problem is formulated, and its relationship with the solution to the initial-value problem of the Newell equation is established. The existence and uniqueness of the solution to the Riemann-Hilbert problem is proved by vanishing lemma. Subsequently, the asymptotic expressions of the solution to the initial-value problem in the dispersive wave region are obtained by using the Deift-Zhou nonlinear steepest descent method. This work extends Newell's original results, providing a rigorous proof for the findings presented in Section 4 of his paper, along with explicit expressions. Furthermore, the comparison between direct numerical simulations and the theoretical results obtained in this paper demonstrates the reliability of the asymptotic expressions.
\end{abstract}

	\section{Introduction}
\ \ \ \
    Fifty years ago, Benney \cite{Benney-1976,Benney-1977} developed a general theory for long wave-short wave interactions, discussed its application to water waves, and proposed a mechanism for significant interaction between long and short waves. In 1978, Newell \cite{Newell-1978} proposed a long wave-short wave interaction model that is now called Newell equation:
	\begin{equation}\label{Newell_1978}
		\left\lbrace \begin{aligned}
			&\ri\frac{\partial B}{\partial T}+\frac{\partial^2B}{\partial X^2}+\left(A^2-2\sigma|B|^2+\ri\frac{\partial A}{\partial X} \right)B=0,\\
			&\frac{\partial A}{\partial T}-2\sigma\frac{\partial \left|B \right|^2 }{\partial X}=0,
		\end{aligned}\right.
	\end{equation}
where $\sigma=\pm 1$, $A(X,T)$ and $B(X,T)$ represent the amplitudes of the long and short waves, respectively. The Newell equation (\ref{Newell_1978}) is a completely solvable model with third-order Yajima-Oikawa type spectral problem (Lax pair) \cite{YO-1975,YO-1976}. Based on this third-order Lax pair, Newell \cite{Newell-1978} proposed the exact soliton solutions of equation (\ref{Newell_1978}) by inverse scattering transform. Then Chowdhury and Chanda \cite{Chowdhury-1986} further considered the Painlevé integrabilty and B{\"a}cklund transformation of the Newell equation by Weiss-Tabor-Carnevale Painlevé test \cite{WTC-1983}. Ling and Liu \cite{Ling-2011} proposed Darboux transformation to construct the multi-soliton solution of this equation. Moreover, Liu \cite{Liu-1994} shew that the Newell equation is related with the un-reduced Yajima-Oikawa equation by a Muira transformation. Both Newell equation and Yajima-Oikawa equation are long wave-short wave resonance model \cite{MaYC-1978,Wright-2006}, and there are various integrable extensions of this kind of model \cite{Chen-Feng-2018,Li_Geng_2020}. Recently, the present authors \cite{WangDS-2025} investigated the long-time asymptotics of the Yajima-Oikawa equation on line by Deift-Zhou nonlinear steepest descent
method \cite{DZ_1993}. However, the long-time asymptotics of the Newell equation (\ref{Newell_1978}) remains an open problem to date. This work is dedicated to solving this nearly 50-year-old problem.
\par
In order to move the critical line $x=0$ to the line $x=t$, taking the Galilean transformation 
$A(X,T)=-\sqrt{2}r(x,t), B(X,T)=q(x,t)$
with $t=T$ and $x=T+\frac{\sqrt{2}}{2}X$, the Newell equation (\ref{Newell_1978}) is converted into 
\begin{equation}\label{Newell}
	\left\lbrace 
	\begin{aligned}
		&\ri\frac{\partial q}{\partial t}+\ri\frac{\partial q}{\partial x}+\frac{1}{2}\frac{\partial^2q}{\partial x^2}+\left(2r^2-2\sigma|q|^2-\ri\frac{\partial r}{\partial x} \right)q=0,\\
		&\frac{\partial r}{\partial t}+\frac{\partial r}{\partial x}+\sigma\frac{\partial \left|q \right|^2 }{\partial x}=0,
	\end{aligned}
	\right.
\end{equation}
which has three-order Lax pair of the form
   \begin{equation}\label{lax_phi}
   	\begin{aligned}
   		&\Psi_x(x,t,k)=U(x,t,k)\Psi(x,t,k), \\ 
   		&\Psi_t(x,t,k)=V(x,t,k)\Psi(x,t,k),
   	\end{aligned}
   \end{equation}
   where 
   \begin{equation*}
   	\begin{aligned}
   		U(x,t,k)=\mathcal{U}(k)+U_1(x,t,k), \quad V(x,t,k)=\mathcal{V}(k)+V_1(x,t,k),
   	\end{aligned}
   \end{equation*}
   and $\mathcal{U}(k)=\text{diag}\{u_1(k),u_2(k),u_3(k)\}$, $\mathcal{V}(k)=\text{diag}\{v_1(k),v_2(k),v_3(k)\}$ with $ u_1(k)=3\ri k,\, u_2(k)=\ri k,\, u_3(k)=-\ri k,\, 
   v_1(k)=\frac{2\ri k^2}{3}-2\ri k,\,  v_2(k)=-\frac{4\ri k^2}{3},\,  v_3(k)=\frac{2\ri k^2}{3}+2\ri k$, and
   \begin{equation} \label{U1-V1}
   \begin{aligned}
   	&U_1(x,t,k)=\begin{pmatrix}
   		0 & \sigma\bar{q} & 2\ri r\\
   		2q & 0 & 2q\\
   		2\ri r & \sigma\bar{q} & 0\\
   	\end{pmatrix},\\
    &V_1(x,t,k)=\begin{pmatrix}
   	\ri \sigma|q|^2 & \sigma\bar{q}\left(k+r-1\right)-\frac{\ri \sigma }{2}\bar{q}_x &  -2\ri r-\ri \sigma|q|^2\\
   	2q\left(k+r-1\right)+\ri q_x & -2\ri\sigma|q|^2 & 2q\left(-k+r-1\right)+\ri q_x\\
   	-2\ri r-\ri \sigma |q|^2 & \sigma\bar{q}\left(-k+r-1\right)-\frac{\ri \sigma }{2} \bar{q}_x & \ri\sigma |q|^2\\
   	\end{pmatrix}. 
    \end{aligned}
   \end{equation}
\par
This work is devoted to studying the long-time asymptotic behaviors of initial-value problem
   \begin{equation}\label{Newell_initial}
   	\left\lbrace 
   	\begin{aligned}
   		&\ri\frac{\partial q}{\partial t}+\ri\frac{\partial q}{\partial x}+\frac{1}{2}\frac{\partial^2q}{\partial x^2}+\left(2r^2-2 \sigma|q|^2-\ri\frac{\partial r}{\partial x} \right)q=0,\\
   		&\frac{\partial r}{\partial t}+\frac{\partial r}{\partial x}+\sigma\frac{\partial \left|q \right|^2 }{\partial x}=0,\\
   		&q(x,0)=q_0(x)\in\mathcal{S}(\mathbb{R}),\quad r(x,0)=r_0(x)\in\mathcal{S}(\mathbb{R}),
   	\end{aligned}
   	\right.
   \end{equation}
where $\mathcal{S}(\mathbb{R})$ is Schwartz space on $\mathbb{R}$. The Riemann-Hilbert (RH) problem associated with the problem (\ref{Newell_initial}) is constructed by direct and inverse scattering analysis. To analyze the long-time asymptotics, the Deift-Zhou nonlinear steepest descent method is then employed to deform the RH problem into a solvable form. 
\par
This paper is organized as follows.  Section \ref{Section-2} lists the main results obtained in this work. The original RH problem is formulated in Section \ref{Section-3} based on the direct and inverse scattering transform. Section \ref{sec_LT} studies the long-time asymptotics of the Newell equation and proves that main theorem of this work.

	\section{Main results} \label{Section-2}\ \ \ \
	 The results of this paper are primarily composed of two reflection coefficients, whose definitions are given first here. Provide the unique solutions $X_{\pm}(x,k)$ and $X_\pm^A(x,k)$ defined by the following Volterra integral equations:
	\begin{align*}
		&X_\pm(x,k)=I-\int_{x}^{\pm\infty}\re^{(x-y)\widehat{\mathcal{U}}(k)}(U_1X)(y,k)\rd y,\\
		&X^A_\pm(x,k)=I+\int_{x}^{\pm\infty}\re^{-(x-y)\widehat{\mathcal{U}}(k)}(U_1^TX^A)(y,k)\rd y,
	\end{align*}
	where  $\re^{\widehat{\mathcal{U}}}A=\re^{\mathcal{U}}A\,\re^{-\mathcal{U}}$. The scattering matrices $s(k)= (s_{ij}(k))_{3×3}$ and $s^A(k)= (s^A_{ij}(k))_{3×3}=\left((-1)^{i+j}\right.$ $\left.m_{ij}(s)\right)_{3\times3}$ are defined through direct scattering analysis
	\begin{align}
		&s(k)=I-\int_\mathbb{R}\re^{-x\widehat{\mathcal{U}}(k)}(U_1X)(x,k)\rd x,\label{sk}\\
		&s^A(k)=I+\int_\mathbb{R}\re^{x\widehat{\mathcal{U}}(k)}(U_1^TX^A)(x,k)\rd x,\label{sak}
	\end{align}
and two reflection coefficients:
\begin{equation}\label{r1_r2}
		r_1(k):=\frac{s_{12}(k)}{s_{11}(k)},\quad
			r_2(k):=\frac{s^A_{31}(k)}{s^A_{33}(k)}, \qquad k\in \mathbb{R}\setminus\{0\}.                                        
	\end{equation}
	
	 \begin{assu}\label{assu_solitonless}
		Assume that $s_{11}(k)$ is nonzero for $k\in\mathbb{C}_+$.
	\end{assu}

    \begin{assu}\label{assu_LT}
    	Assume that the reflection coefficients $r_1(k)$ and $r_2(k)$ as defined in equation \eqref{r1_r2} satisfy the relations:
    	\begin{equation*}
    		\begin{aligned}
    			&{1-2|r_1(-k)|^2+|r_2(k)|^2}>0,
    		\end{aligned}\qquad k\in \mathbb{R}.
    	\end{equation*} 
    \end{assu}

   \begin{rhp}\label{rhp_M}
   	Find a $3\times3$ matrix-valued function $M(x,t,k)$ with the following properties:
   	\begin{enumerate}
   		\item The function $M(x,t,k)$ is analytic for $k\in\mathbb{C}\setminus\mathbb{R}$.
   		\item As $k$ approaches $\mathbb{R}$ from the left and right, the boundary values $M_+(x,t,k)$ and $M_-(x,t,k)$ of $M(x,t,k)$ exist and satisfy the following relationship:
   		\begin{equation*}
   			M_+(x,t,k)=M_-(x,t,k)v(x,t,k), \quad k\in \mathbb{R},
   		\end{equation*}
   		where
   		\begin{equation}\label{jump_0}
   			v(x,t,k)=
   			\begin{pmatrix}
   				1 &  -r_1(k)\re^{t\Phi_{12}(x,t,k)} & r_2(k)\re^{t\Phi_{13}(x,t,k)}\\
   				2\sigma r_1^*(k)\re^{-t\Phi_{12}(x,t,k)} & 1-2\sigma \left|r_1(k) \right|^2 &-2\sigma \alpha(k)\re^{t\Phi_{23}(x,t,k)}\\
   				r_2^*(k)\re^{-t\Phi_{13}(x,t,k)} & \alpha^*(k)\re^{-t\Phi_{23}(x,t,k)} & 1-2\sigma \left|r_1(-k) \right|^2+ \left|r_2(k) \right|^2\\
   			\end{pmatrix},
   		\end{equation}
   		and $\alpha(k)=r_1^*(-k)-r_1^*(k)r_2(k)$, $\Phi_{ij}(x,t,k)=(u_i(k)-u_j(k))x/t+(v_i(k)-v_j(k))$, $1\leq i<j\leq3$.
   		
   		\item \label{rhp_M_asy} For $k\in\mathbb{C}\setminus\mathbb{R}$, $M(x,t,k)=I+\mathcal{O}(k^{-1})$ as $k\rightarrow\infty$; $M(x,t,k)=\mathcal{O}(1)$ as  $k\rightarrow0$.

   		\item  $M$ satisfies the symmetries for $k \in \mathbb{C} \setminus \mathbb{R}$:
   		\begin{equation}\label{sym_M}
   		\begin{aligned}
   			&M^{-1}(x,t, k) = \mathcal{A}^{-1} M^\dagger(x,t, \bar k) \mathcal{A},\\
   			&M(x,t, k)=\mathcal{B} M(x,t,-k) \mathcal{B},
   		\end{aligned}\end{equation}
   		where the superscript $\dagger$ denotes the conjugate transpose, and
   		\begin{equation}\label{sym_AB}
   			\mathcal{A}=\begin{pmatrix}
   				1 & 0 & 0\\
   				0 & -\frac{\sigma}{2} & 0\\
   				0 & 0 & 1\\
   			\end{pmatrix},\quad \mathcal{B}=\begin{pmatrix}
   				0 & 0 & 1\\
   				0 & 1 & 0\\
   				1 & 0 & 0\\
   			\end{pmatrix}.
   		\end{equation}
   	\end{enumerate}
   \end{rhp}

 \begin{lem}\label{lem_vanishing}
    Suppose that the reflection coefficients $r_j(k)$, $j=1,2,$ defined in~\eqref{r1_r2} satisfy Assumptions \ref{assu_solitonless} and \ref{assu_LT}. Then the corresponding RH problem \ref{rhp_M} for $M(x,t,k)$ admits a unique solution.
\end{lem}

This lemma can be proved by the strategy of Zhou's vanishing lemma \cite{Zhou_SIAM_1989_vanishing}.

    \begin{theo}\label{theo_reconstruction formula}
   	Let $q(x,t)$ and $r(x,t)$ be a solution to the initial value problem \eqref{Newell_initial}, $r_1(k),\, r_2(k)$ are defined in equation \eqref{r1_r2} and satisfy Assumptions \ref{assu_solitonless} and \ref{assu_LT}.  Define an existence time $T\in \left( 0,\infty\right] $. The RH problem \ref{rhp_M} admits a unique solution $M(x,t,k)$ for each point $(x,t)\in \mathbb{R}\times \left[0,T\right)$, and moreover, the solution to the initial value problem \eqref{Newell_initial} of Newell equation can be expressed by
   	\begin{equation}\label{reconstruct}
   		\left\lbrace 
   		\begin{aligned}
   			&q(x,t)=\ri\lim\limits_{k\to\infty}k\left(M(x,t,k)-I \right)_{21},\\
   			&r(x,t)=2\lim\limits_{k\to\infty}k\left(M(x,t,k)-I \right)_{31}.
   		\end{aligned}\right.
   	\end{equation}
   \end{theo}

      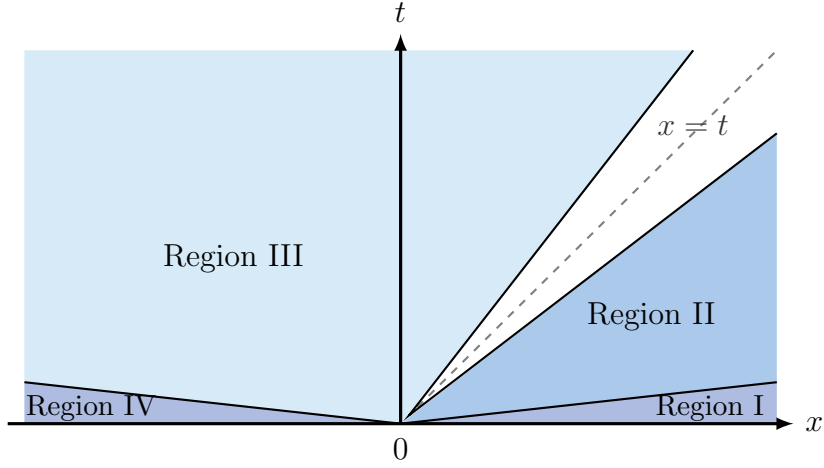
\begin{figure}[htbp]
   	\centering
   	\begin{tikzpicture}[scale=1.1]
   		\draw[thick,gray,dashed] (0.2,0.2) -- (4.5,4.5);
   		
   		\definecolor{mycolor1}{HTML}{AEBCE1} 
   		\definecolor{mycolor2}{HTML}{D6EAF8} 
   		\definecolor{mycolor3}{HTML}{ABC9EB} 
   		\definecolor{mycolor4}{HTML}{98A3CA} 
   		
   		\fill[mycolor1] (4.5,0.5) -- (0,0) -- (4.5,0) -- cycle;
   		\fill[mycolor3] (4.5,3.5) -- (0,0) -- (4.5,0.5) -- cycle;
   		\fill[mycolor2] (3.5,4.5) -- (0,0) -- (-4.5,0.5) -- (-4.5,4.5) -- cycle;
   		\fill[mycolor1] (-4.5,0.5) -- (0,0) -- (-4.5,0)-- cycle;

   		\node at (0,-0.3) {0};
   		\node at (3,1.3) {Region II};
   		\node at (-2,2) {Region III};
   		\node[above] at (-3.7,-0.1) {\small Region IV};
   		\node[above] at (3.7,-0.1) {\small Region I};
   		
   		\draw[thick] (3.5,4.5) -- (0.1,0.1) -- (4.5,3.5);
   		\draw[thick] (-4.5,0.5) -- (0,0) -- (4.5,0.5);
   		
   		\node[black!80] at (3.5,3.6) {$x=t$};
   		
   		\draw[very thick,-latex] (-4.7,0) -- (4.7,0) node[right] {$x$};
   		\draw[very thick,-latex] (0,0) -- (0,4.7) node[above] {$t$};
   		
  	\end{tikzpicture}
   	\caption{Division of asymptotic regions in the upper $(x,t)$-half plane. 
   		Region I: $\tau \in \mathcal{I}_1 = [0, \tau_{\max}]$; 
   		Region II: $\zeta \in \mathcal{I}_2 \subset (1,\infty)$; 
   		Region III: $\zeta \in \mathcal{I}_3 \subset (-\infty,1)$; 
   		Region IV: $\tau \in \mathcal{I}_4 = [-\tau_{\max}, 0]$, 
   		with $\tau = 1/\zeta = t/x$, $\tau_{\max} \in (0,1)$ and $\mathcal{I}_2, \mathcal{I}_3$ are fixed compact sets.}
   	\label{fig_asy}
   \end{figure}

   \begin{theo}\label{theo_region}
   	Assume that $r_0, q_0 \in \mathcal{S}(\mathbb{R})$ and that the conditions of Theorem \ref{theo_reconstruction formula} are satisfied, then the solution $r(x,t)$ and $q(x,t)$ of the initial value problem \eqref{Newell_initial} obey the following asymptotic formulas in the four regions  shown in Figure \ref{fig_asy}(the parameter ranges of the regions are given in the caption of figure):
   		\begin{align}
   				&\textbf{Region {\rm{I}}:}\left\lbrace 
   				\begin{aligned}
   					r(x,t)&=\mathcal{O}\left(\frac{1}{x^{N}}+\frac{C_N(\tau)\ln x}{x}\right),\\
   					q(x,t)&=\frac{\sqrt{2\pi}\left(2\sqrt{x\tau}\right)^{-2\ri\nu}\text{exp}\left(\ri\nu \ln(2k_0)+\frac{3\pi \ri}{4}+2\ri x\tau k_0^2-\frac{\pi\nu}{2}+s_1\right)}{2\sqrt{x\tau}r_1(-k_0)\Gamma(-\ri\nu)}\\
                    &\quad+\mathcal{O}\left(\frac{1}{x^{N}}+\frac{C_N(\tau)\ln x}{x}\right).
   				\end{aligned} \right.\label{theo_region1}\\
   				&\textbf{Region {\rm{II}}:}\left\lbrace 
   				\begin{aligned}
   					r(x,t)
   					&=\mathcal{O}\left(\frac{\ln t}{t}\right),\\
   					q(x,t)&=\frac{\sqrt{2\pi}\left(2\sqrt{t}\right)^{-2\ri\nu}\text{exp}\left(\ri\nu \ln(2k_0)+\frac{3\pi \ri}{4}+2\ri tk_0^2-\frac{\pi\nu}{2}+s_1\right)}{2\sqrt{t}r_1(-k_0)\Gamma(-\ri\nu)}+\mathcal{O}\left(\frac{\ln t}{t}\right).
   				\end{aligned}\right.\label{theo_region2}\\
   				&\textbf{Region {\rm{III}}:}\left\lbrace 
   				\begin{aligned}
   					r(x,t)&
   					=\mathcal{O}\left(\frac{\ln t}{t}\right),\\
   					q(x,t)
   					&=\frac{\sqrt{2\pi}\left(2\sqrt{t}\right)^{-2\ri\nu}\text{exp}\left(\ri\nu \ln(-2k_0)+\frac{3\pi \ri}{4}+2\ri tk_0^2-\frac{\pi\nu}{2}-\pi \nu_2+ s_2\right)}{2\sqrt{t} \alpha^*(k_0)\Gamma(-\ri\nu)}+\mathcal{O}\left(\frac{\ln t}{t}\right).
   				\end{aligned}\right.\label{theo_region3}\\
   				&\textbf{Region {\rm{IV}}:}\left\lbrace 
   				\begin{aligned}
   					r(x,t)&=\mathcal{O}\left(\frac{1}{|x|^{N}}+\frac{C_N(\tau)\ln |x|}{|x|}\right),\\
   					q(x,t)
   					&=\frac{\sqrt{2\pi}\left(2\sqrt{x\tau}\right)^{-2\ri\nu}\text{exp}\left(\ri\nu \ln(-2k_0)+\frac{3\pi \ri}{4}+2\ri x\tau k_0^2-\frac{\pi\nu}{2}-\pi \nu_2+ s_2\right)}{2\sqrt{x\tau} \alpha^*(k_0)\Gamma(-\ri\nu)}\\
                    &\quad+\mathcal{O}\left(\frac{1}{|x|^{N}}+\frac{C_N(\tau)\ln |x|}{|x|}\right).
   				\end{aligned}
   				\right.\label{theo_region4}
   			\end{align}
   		Here,  $\Gamma$ denotes the Gamma function, $C_N(\tau)$ is a smooth function that vanishes to all orders as $\tau\to0$ for each $N$, $k_0=\frac{x-t}{2t}=\frac{1-\tau}{2\tau}$ and $\alpha(k_0)=r_1^*(-k_0)-r_1^*(k_0)r_2(k_0)$,
   	\begin{align*}
   		&\nu=-\frac{1}{2\pi}\ln(1-2\sigma \left|r_1(-k_0) \right|^2 ), \quad &&\nu_2=-\frac{1}{2\pi} \ln \left(1-2\sigma|r_1(k_0)|^2+|r_2(-k_0)|^2\right),
   	\end{align*}
   		\begin{align}
   			s_1=&\frac{1}{2\pi\ri}\int_{k_0}^{\infty}\ln\left({(s+k_0)}{(s-k_0)^2}\right){\rm{d}}\ln(1-2\sigma \left|r_1(-s)\right|^2 ),\label{s1}\\
   			s_2=&\frac{1}{2\pi\ri}\int_{-k_0}^{\infty}\ln\left|{(s+k_0)^2}{(s-k_0)}\right|{\rm{d}}\ln(1-2\sigma \left|r_1(-s)\right|^2 ) \nonumber\\
            &+\frac{1}{2\pi\ri}\int_{-\infty}^{k_0}\ln|(s-k_0)(s+k_0)|{\rm{d}}\ln \left(1-2\sigma|r_1(s)|^2+|r_2(-s)|^2\right)\nonumber\\
   			&+\frac{1}{2\pi\ri}\int_{k_0}^{-k_0}\ln\left|(s-k_0)^2(s+k_0) \right| {\rm{d}}\ln \left(1-2\sigma|r_1(-s)|^2+|r_2(s)|^2\right).\nonumber
   		\end{align}
  
   \end{theo}
   
  The long-time asymptotic behavior of the solution in the four regions is derived and proved in Section \ref{sec_LT}. Additionally, a comparative validation between the asymptotic and numerical solutions is presented. Given the following two sets of initial conditions:
   \begin{equation}\label{initial_condition}
   	\left\lbrace
   	\begin{aligned}
   		&r_0(x)=0.3\cos(x)\,\re^{-\frac{x^2}{2}},\\
        &q_0(x)=0.2\sin(x)\,\re^{-\frac{x^2}{2}}.
   	\end{aligned}\right.\qquad\qquad\qquad \\
   	\left\lbrace
   	\begin{aligned}
       	&r_0(x)=0.1\,\re^{-\frac{x^2}{2}},\\
   		&q_0(x)=0.2\,\re^{-\frac{x^2}{2}}.
   	\end{aligned}\right.
   \end{equation}
   Figures \ref{figrhp-direct} and \ref{figrhp-direct2} present a comparison between the direct numerical simulation results and the theoretical results from Theorem \ref{theo_region} under the two sets of initial conditions \eqref{initial_condition} for the two cases $\sigma=1$ and $\sigma=-1$. It can be seen from both figures that the long-time asymptotic results given in Theorem \ref{theo_region} agree well with the numerical results within an acceptable range of error.
   \begin{figure}[ht]
   	\centering
   	\begin{subfigure}[b]{0.49\textwidth}
   		\includegraphics[width=\textwidth]{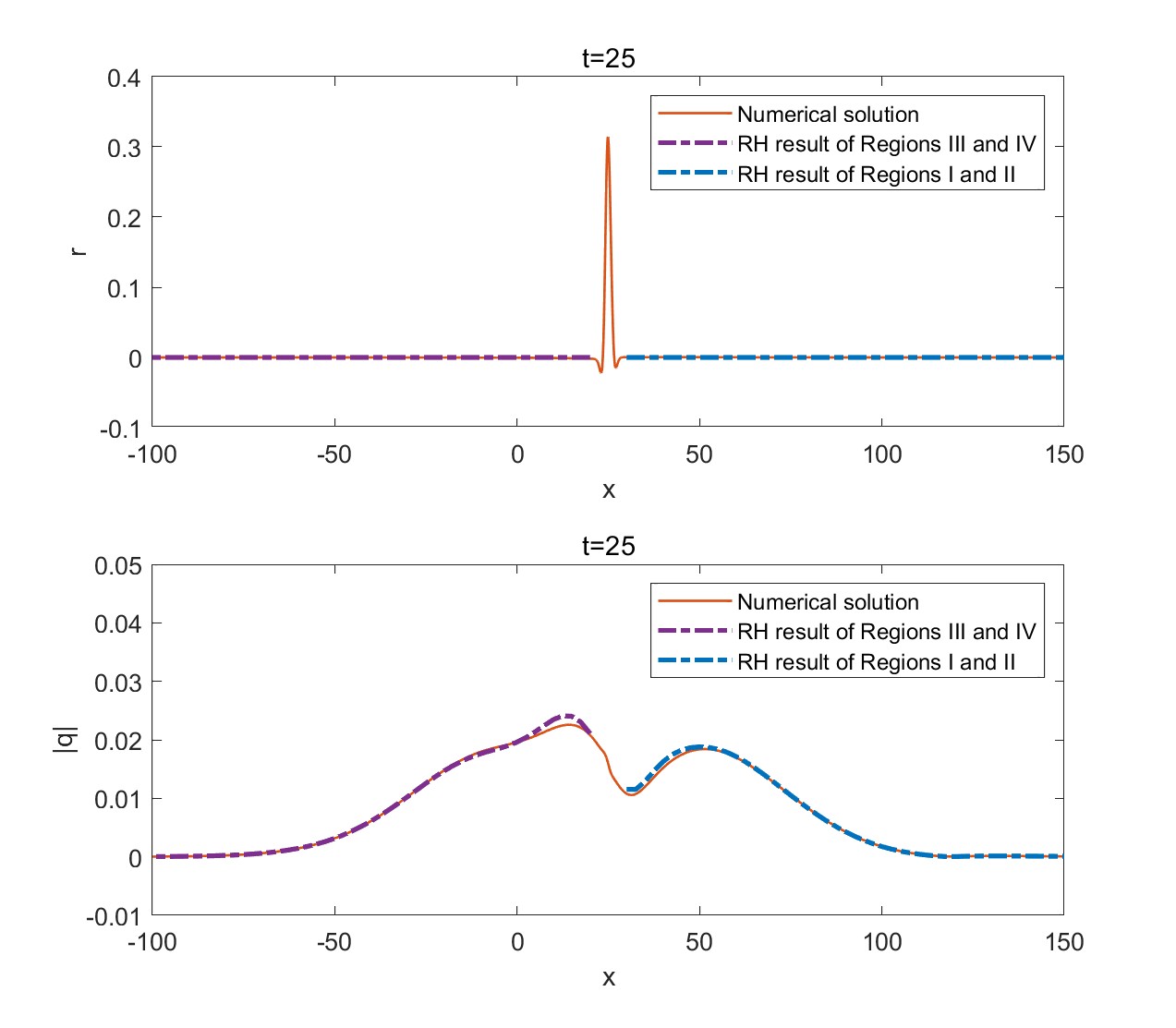}
   		\caption{The initial value condition (\ref{initial_condition} left) for $t=25$ and $\sigma=1$.}
   		\label{figrhp-1}
   	\end{subfigure}
   	\begin{subfigure}[b]{0.49\textwidth}
   		\includegraphics[width=\textwidth]{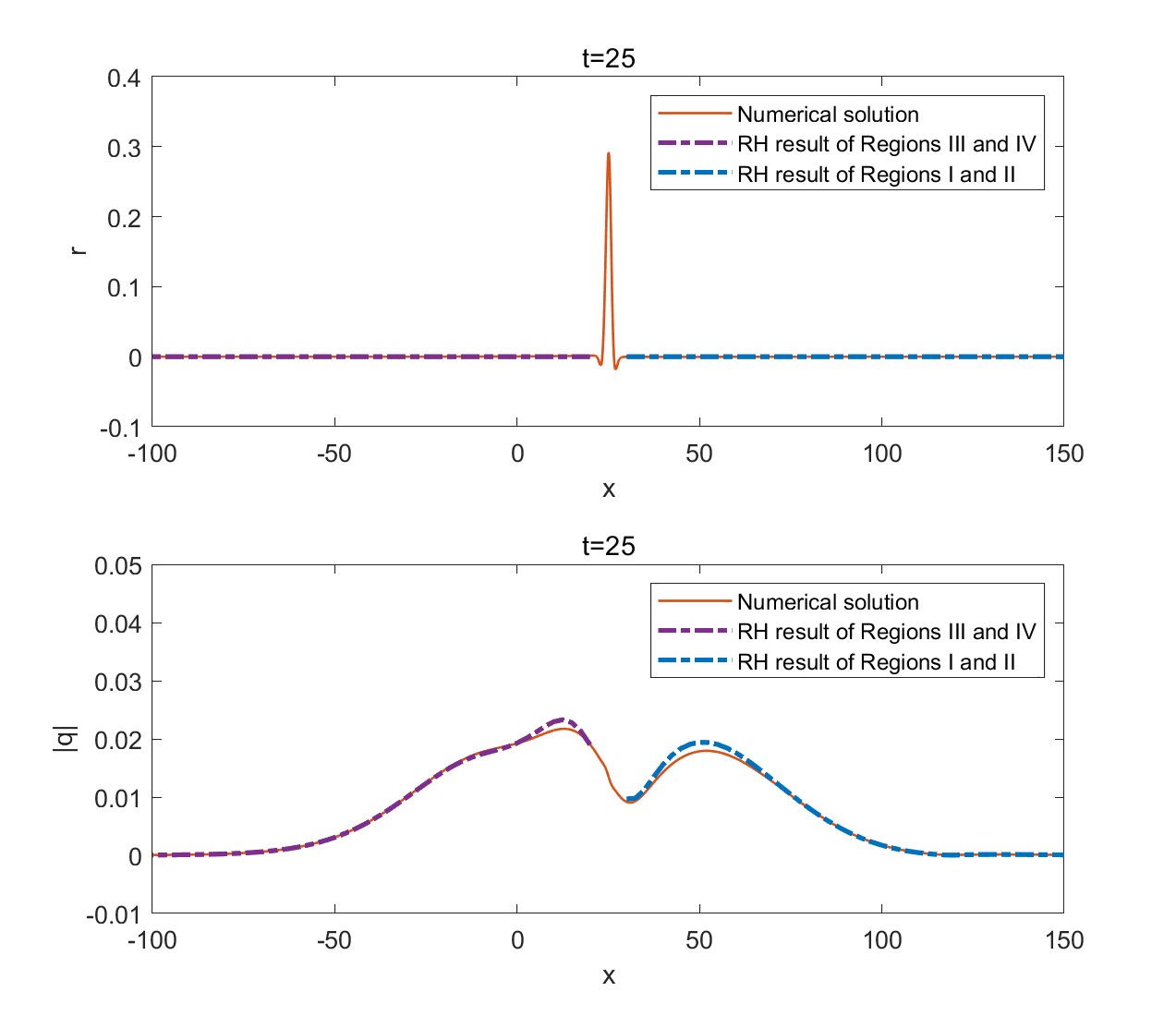}
   		\caption{The initial value condition (\ref{initial_condition} left) for $t=25$ and $\sigma=-1$.}
   		\label{figrhp-2}
   	\end{subfigure}
   	\caption{Comparisons of theoretical result given by Theorem \ref{theo_region} and the full numerical simulations of the Newell equation \eqref{Newell} under the initial value condition (\ref{initial_condition} left).}
   	\label{figrhp-direct}
   \end{figure}
   \begin{figure}[ht]
   	\centering
   	\begin{subfigure}[b]{0.49\textwidth}
   		\includegraphics[width=\textwidth]{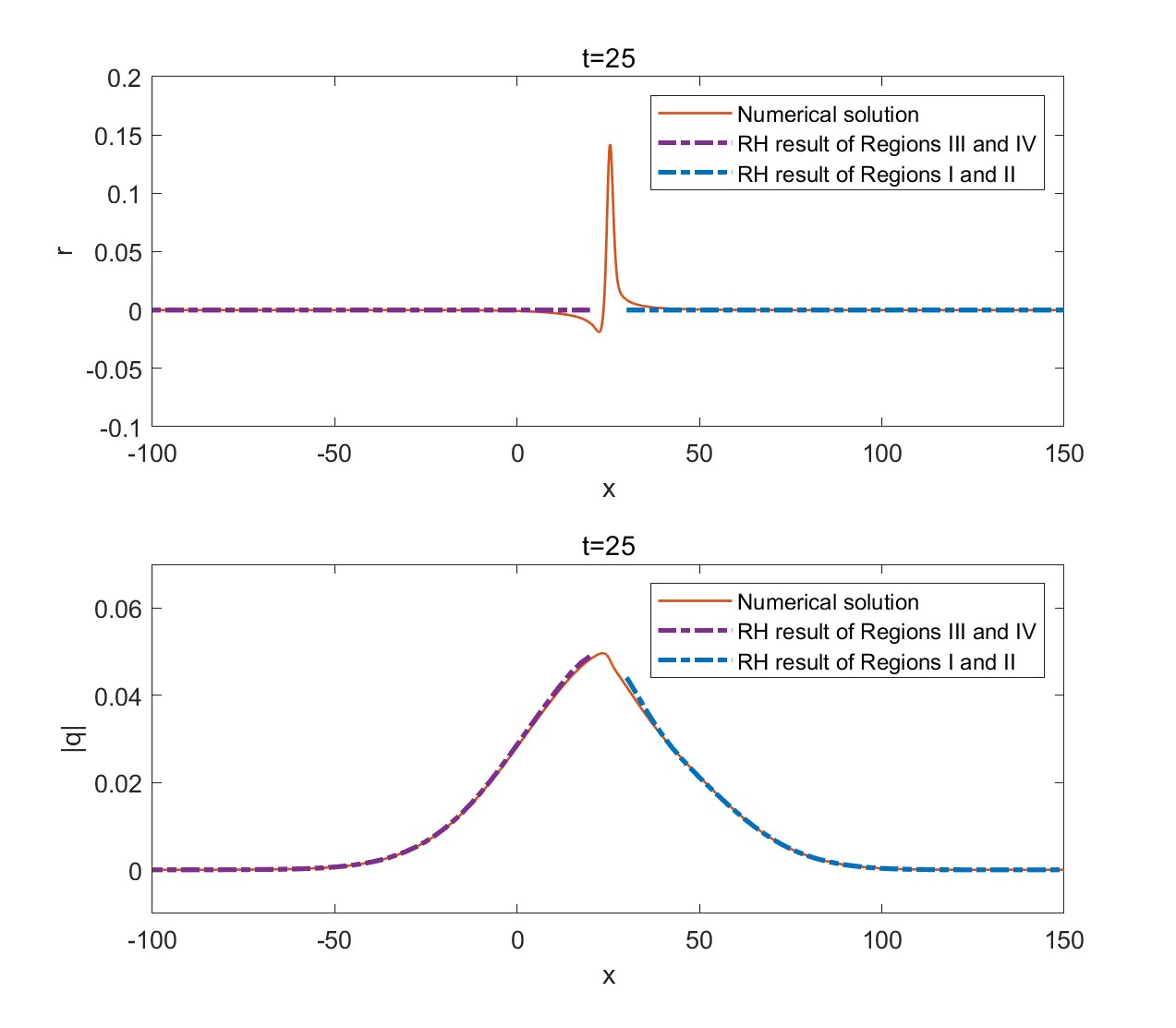}
   		\caption{The initial value condition (\ref{initial_condition} right) for $t=25$ and $\sigma=1$.}
   		\label{figrhp-3}
   	\end{subfigure}
   	\begin{subfigure}[b]{0.49\textwidth}
   		\includegraphics[width=\textwidth]{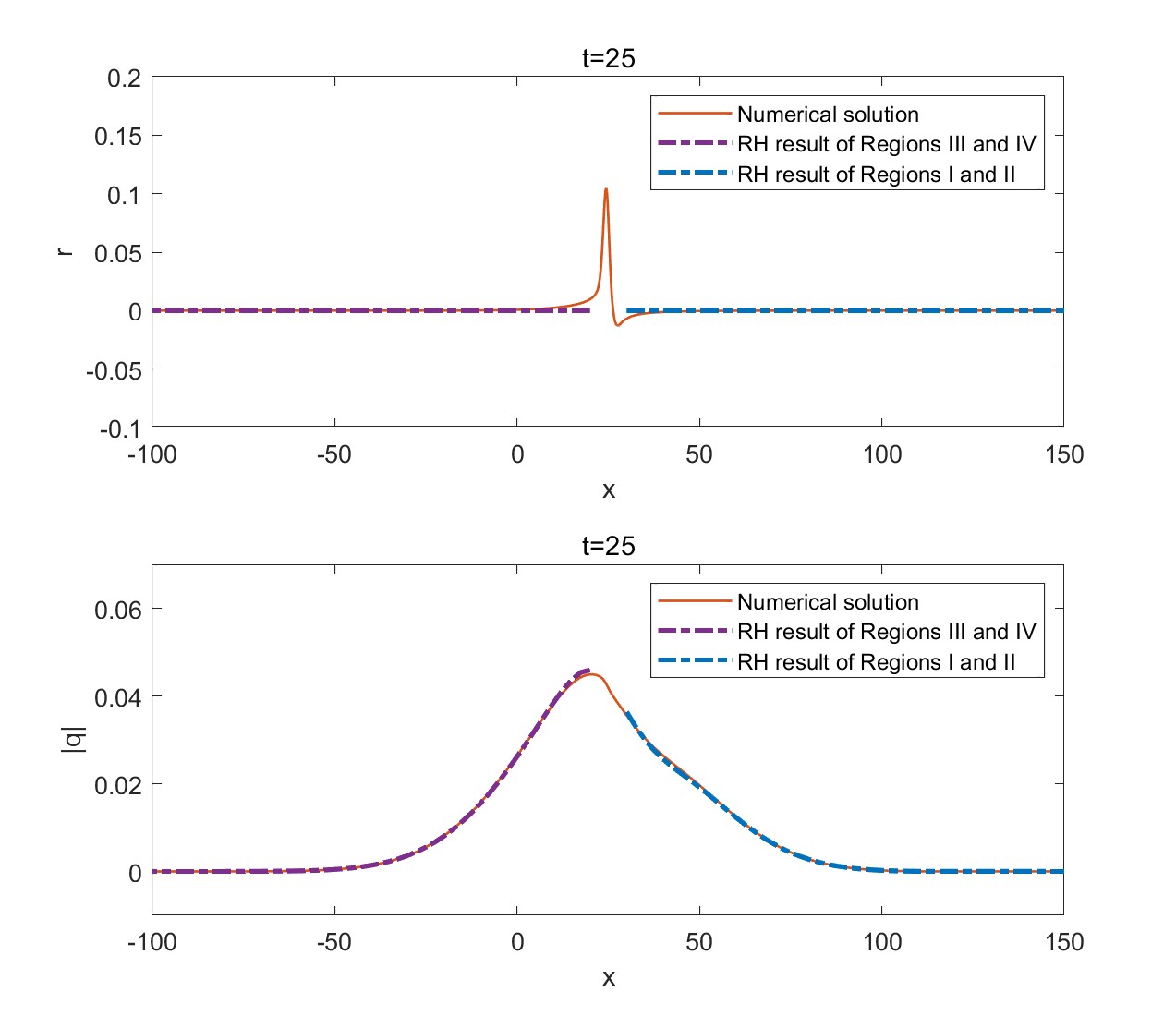}
   		\caption{The initial value condition (\ref{initial_condition} right) for $t=25$ and $\sigma=-1$.}
   		\label{figrhp-4}
   	\end{subfigure}
   	\caption{Comparisons of theoretical result given by Theorem \ref{theo_region} and the full numerical simulations of the Newell equation \eqref{Newell} under the initial value condition (\ref{initial_condition} right).}
   	\label{figrhp-direct2}
   \end{figure}

   \section{The Riemann-Hilbert problem} \label{Section-3} \ \ \ \

   Follwoing the precedure in \cite{Charlier-2022,Charlier-2023}, after some calculations, it can be shown that the matrix-valued functions $U_1(x,t,k)$ and $V_1(x,t,k)$ in (\ref{U1-V1}) satisfy the following symmetries
	\begin{equation*}\label{symmetry}
		\begin{aligned}
			&U_1(x,t,k)=-\mathcal{A}^{-1} U_1^\dagger(x,t,\bar k)\mathcal{A},\quad &&U_1(x,t,k)=\mathcal{B} U_1(x,t,-k)\mathcal{B},\\
			&V_1(x,t,k)=-\mathcal{A}^{-1} V_1^\dagger(x,t,\bar k)\mathcal{A},\quad &&V_1(x,t,k)=\mathcal{B} V_1(x,t,-k)\mathcal{B},
		\end{aligned}
	\end{equation*}
    where $\mathcal{A}$ and $\mathcal{B}$ are given by equation \eqref{sym_AB}. Construct the Lax pair in the following matrix form:
	\begin{equation}\label{lax_check}
		\left\lbrace
		\begin{aligned}
			\check{X}_x=\check{L}\check{X},\\
			\check{X}_t=\check{Z}\check{X},
		\end{aligned} \right.
	\end{equation}
	where $\check{X}$ is a $3\times 3$ matrix-valued function composed of three linearly independent solutions of the Lax pair \eqref{lax_phi}.
	Introduce $\check{X}=X e^{\mathcal{U}x+\mathcal{V}t}$, and substitute it into \eqref{lax_check}. Then we obtain
\begin{equation}\label{lax_X}
		\left\lbrace
		\begin{aligned}
			X_x-\left[ \mathcal{U},X\right]=U_1X,\\
			X_t-\left[ \mathcal{V},X\right]=V_1X.\\
		\end{aligned} \right.
	\end{equation}
	Fix $t=0$ and denote $U(x, k):=U(x,0, k)$ for notational simplicity. Next, consider the case of the $x$-part of the Lax pair \eqref{lax_X} for $t=0$:
	\begin{equation}\label{lax_xpart}
		X_x-\left[ \mathcal{U},X\right]=U_1X.
	\end{equation}
	Define $X_+(x,k)$ and $X_-(x,k)$ as the matrix-valued solutions of the linear Volterra integral equations:
	\begin{align}
		&X_+(x, k)=I-\int_{x}^{\infty}\re^{(x-y)\widehat{\mathcal{U}}(k)}(U_1X_+)(y,k)\rd y,\label{volX+}\\
		&X_-(x, k)=I+\int_{-\infty}^{x}\re^{(x-y)\widehat{\mathcal{U}}(k)}(U_1X_-)(y,k)\rd y.\label{volX-}
	\end{align}

    \begin{prop}\label{prop_Xpm}
		Suppose $r_0,\, q_0\in\mathcal{S}(\mathbb{R})$, the equations \eqref{volX+} and \eqref{volX-} uniquely define the two solutions $X_\pm(x,k)$ in equation \eqref{lax_X}, which satisfy the following properties:
		\begin{enumerate}
			\item The vector functions $X_{+,1}(x,k)$ and $X_{-,3}(x,k)$  are well-defined and analytic for $x\in\mathbb{R}$ and $\rim k>0$, and the vector functions $X_{+,3}(x,k)$ and $X_{-,1}(x,k)$  are well-defined and analytic for $x\in\mathbb{R}$ and $\rim k<0$. Additionally, $X_{\pm,2}(x,k)$ are not analytic in any domain.
            \item For each $n \geq 1$ and $\epsilon > 0$, there are bounded smooth positive functions $f_\pm(x)$ of $x \in \mathbb{R}$ with rapid decay as $x \to \pm\infty$, respectively, such that the following estimates hold for $x \in \mathbb{R}$ and $j = 0, 1, \ldots, n$:
            \begin{equation*}
                \begin{aligned}
                    &\left| \frac{\partial^j}{\partial k^j} \bigl( X_\pm(x, k) - I \bigr) \right| \leq f_\pm(x), \quad k\in (\overline{\mathbb{C}_\pm},\mathbb{R},\overline{\mathbb{C}_\mp}), \ |k| > \epsilon.
                \end{aligned}
            \end{equation*}
			\item For $k\in\mathbb{R}\setminus\{0\}$, $\det X_{\pm}=1$. The functions $X_{\pm}(x,k)$ satisfy the symmetries
			\begin{equation*}\label{sym_X_A}
				X^H_{\pm}\left(  x,\bar k \right)=\mathcal{A}X^{-1}_{\pm}(x,k)\mathcal{A}^{-1} ,
			\end{equation*}
			\begin{equation*}\label{sym_X_B}
				X_{\pm}(x,k)=\mathcal{B} X_{\pm}(x,-k)\mathcal{B}.
			\end{equation*}
		
		\end{enumerate}
	\end{prop}

	Equation \eqref{lax_xpart} has the following formal power series solutions
	\begin{equation}\label{X_pm_formal}
		X_{\pm,formal}(x,k)=X_{\pm0}(x)+\frac{X_{\pm1}(x)}{k}+\frac{X_{\pm2}(x)}{k^2}+\cdots,\quad k\to\infty,
	\end{equation}
	where the coefficients $\left\lbrace X_{\pm j}\right\rbrace_1^\infty $ satisfy $\lim\limits_{x\to\pm\infty}X_{\pm j}(x)=0_{3\times3}$ for $j\geq 1$. After calculation, the first few coefficients are as
    follows:
	\begin{align*}
		&X_{\pm0}=I,\\
		& X_{\pm 1}=\begin{pmatrix}
	 		\ri \int_{\pm\infty}^{x}(r_0(x')-\sigma|q_0(x')|^2)\rd x' & \frac{\ri\sigma}{2}\bar{q}_0(x) & -\frac{r_0(x)}{2}\\
	 		-\ri q_0(x) & 0 & \ri q_0(x)\\
	 		\frac{r_0(x)}{2} & -\frac{\ri\sigma}{2}\bar{q}_0(x) & -\ri \int_{\pm\infty}^{x}(r_0(x')-\sigma|q_0(x')|^2)\rd x'
	 	\end{pmatrix}.
	\end{align*}
	The remaining series of coefficients can all be calculated, but they are not given here. We will next describe the behavior of $X_\pm$ when $k$ is sufficiently large.
	
	\begin{prop}\label{prop_Xpm_k_infty}
		Suppose $q_0,\, r_0\in\mathcal{S}(\mathbb{R})$. As $k\to\infty$, $X_{\pm}$ coincide to all orders with $X_{\pm,formal}$, respectively. For any integer $j\geq 0$, $p\geq0$,
		 \begin{equation*}
		 	\left|\frac{\partial^j}{\partial k^j}\left(X_\pm-X_{\pm}^{[p]}\right)  \right| \leq \frac{f_\pm(x)}{\left|k \right|^{p+1} },\quad x\in\mathbb{R},\,\,k\in (\overline{\mathbb{C}_\pm},\mathbb{R},\overline{\mathbb{C}_\mp}),\,\,\left|k \right| \geq 2,
		 \end{equation*}
        where $f_\pm(x)$ are bounded smooth positive functions of $x\in\mathbb{R}$ with rapid decay as $x\to \pm \infty$, respectively, and
		\begin{equation}\label{Xp_infty}
			X^{[p]}_{\pm}(x,k)=I+\frac{X_{\pm1}(x)}{k}+\frac{X_{\pm2}(x)}{k^2}+\cdots+\frac{X_{\pm p}(x)}{k^p}.
		\end{equation}
	\end{prop}
	\begin{proof}
		The proof follows by considering the equation satisfied by the quotient $\left(X_{\pm}^{[p]}\right)^{-1}X_\pm$.
	\end{proof}

    The spectral matrix function $s(k)$ is defined by equation \eqref{sk}, and its properties are given by the following proposition.
	\begin{prop}\label{prop_s}
		Given that $q_0, r_0\in \mathcal{S}(\mathbb{R})$, the matrix-valued function $s( k )$ defined by equation \eqref{sk} exhibits the following  properties:
		\begin{enumerate}
			\item The entries of $s(k)$ are defined and continuous for
			\begin{equation}\label{sk_k_define}
				k\in \begin{pmatrix}
					\overline{\mathbb{C}_+} & \mathbb{R} & \mathbb{R}\\
					\mathbb{R} & \mathbb{R} & \mathbb{R}\\
					\mathbb{R} & \mathbb{R} & \overline{\mathbb{C}_-}\\
				\end{pmatrix}\setminus\{0\}.
			\end{equation}
			The function $s_{11}(k)$ is analytic on $\overline{\mathbb{C}_+}$, and the function $s_{33}(k)$ is analytic on $\overline{\mathbb{C}_-}$. The remaining elements of $s(k)$ are not analytic in any region.
            \item \label{prop_sub_sk_ktoinfty}$s(k)$ approaches the identity matrix as $k\to\infty$, and the following equation holds for $j=1,2,\cdots,N$:
             \begin{equation*}
                \left|\partial_k^j\left(s(k)-I-\sum_{j=1}^N\frac{s_j}{k^j}\right)\right|=\mathcal{O}\left(k^{-N-1}\right),\quad k\to\infty,\,\, k\,\,\rm{in} \,\,\eqref{sk_k_define},
            \end{equation*}
            where $\{s_j\}_{j=1}^\infty$ are diagonal matrices.
            
			\item  For $k\in\mathbb{R}\setminus\{0\}$, $\det s(k)=1$, and the function $s( k )$ satisfies the symmetries
			\begin{equation}\label{sym_s}
				\begin{aligned}
				&s^{H}\left(  \bar k \right) =\mathcal{A}s^{-1} (k)\mathcal{A}^{-1},\\
				& s(-k)=\mathcal{B}s(k)\mathcal{B}.
 			\end{aligned}
 			\end{equation}
			\item If $q_0$, $r_0$ have compact support, then $s(k)$ is defined and analytic for $k\in\mathbb{C}\setminus\{0\}$, $\det s(k)=1$, and
            \begin{equation}\label{X_s}
		X_+(x, k )=X_-(x, k )\re^{x\widehat{\mathcal{U}}( k )}s( k ),\quad  k \in \mathbb{C}\setminus \left\lbrace 0 \right\rbrace.
	\end{equation}
		\end{enumerate}
	\end{prop}

	Denote $X^A:=\left( X^{-1}\right) ^T$, where the superscript $A$ is defined as the cofactor matrix of $X$, and $X^A$ satisfies the following Lax equation
	\begin{equation}\label{lax_XA}
		\left(X^A \right)_x+\left[ \mathcal{U},X^A\right]=-U_1^TX^A.  
	\end{equation}
    Moreover, $X^A_\pm$ are solutions of the following Volterra integral equations, respectively:
\begin{align}
	&X_+^A(x, k)=I+\int_{x}^{\infty}\re^{-(x-y)\widehat{\mathcal{U}}(k)}(U_1^TX^A_+)(y,k)\rd y,\label{volXA+}\\
	&X_-^A(x, k)=I-\int_{-\infty}^{x}\re^{-(x-y)\widehat{\mathcal{U}}(k)}(U_1^TX_-^A)(y,k)\rd y.\label{volXA-}
\end{align}

	\begin{prop}\label{prop_XApm}
		Suppose $q_0,\, r_0\in\mathcal{S}(\mathbb{R})$, the equations \eqref{volXA+} and \eqref{volXA-} uniquely define the two solutions $X^A_\pm(x,k)$ in equation \eqref{lax_XA}, which satisfy the following properties:
		\begin{enumerate}
        \item The vector functions $X^A_{+,1}(x,k)$ and $X^A_{-,3}(x,k)$  are well-defined and analytic for $x\in\mathbb{R}$ and $\rim k<0$, and the vector functions $X^A_{+,3}(x,k)$ and $X^A_{-,1}(x,k)$  are well-defined and analytic for $x\in\mathbb{R}$ and $\rim k>0$. Additionally, $X^A_{\pm,2}(x,k)$ are not analytic in any domain.

            \item For each $n \geq 1$ and $\epsilon > 0$, there are bounded smooth positive functions $f_\pm(x)$ of $x \in \mathbb{R}$ with rapid decay as $x \to \pm\infty$, respectively, such that the following estimates hold for $x \in \mathbb{R}$ and $j = 0, 1, \ldots, n$:
            \begin{equation*}
                \begin{aligned}
                    &\left| \frac{\partial^j}{\partial k^j} \bigl( X^A_\pm(x, k) - I \bigr) \right| \leq f_\pm(x), \quad k\in (\overline{\mathbb{C}_\mp},\mathbb{R},\overline{\mathbb{C}_\pm}), \ |k| > \epsilon.
                \end{aligned}
            \end{equation*}

			\item For $k\in\mathbb{R}\setminus\{0\}$, $\det X_\pm ^A=1$. The functions $X^A_\pm$ satisfy the following symmetries for $x\in\mathbb{R}$ and $k\in (\overline{\mathbb{C}_\mp},\mathbb{R},\overline{\mathbb{C}_\pm})\setminus\{0\}$:
			\begin{align*}
				&\left( X^A_{\pm}\right)^\dagger\left(  x,\bar k \right)=\mathcal{A}^{-1}X_\pm^T(x,k)\mathcal{A},\\
				&X^A_{\pm}(x,k)=\mathcal{B}X_\pm^A(x,-k)\mathcal{B}.
			\end{align*}

		\end{enumerate}
	\end{prop}
	
	Similar to the discussion for $X_\pm$, equation \eqref{lax_xpart} has the following formal power series solutions
	\begin{equation*}
		X^A_{\pm,formal}(x,k)= X_{\pm0}^A(x)+\frac{ X_{\pm1}^A(x)}{k}+\frac{ X_{\pm2}^A(x)}{k^2}+\cdots,\quad k\to\infty,
	\end{equation*}
	where $X^A_{\pm,formal}(x,k)$ be the cofactor matrices of the functions in \eqref{X_pm_formal}.
	
	\begin{prop}\label{prop_XApm_k_infty}
		Suppose $q_0,\, r_0\in\mathcal{S}(\mathbb{R})$. As $k$ approaches infinity, $X^A_{\pm}$ coincide to all orders with $X^A_{\pm,formal}$, respectively. For any integer $j\geq 0$, $p\geq0$,
		\begin{equation*}
			\left|\frac{\partial^j}{\partial k^j}\left(X^A_\pm-\left( X^A_{\pm}\right) ^{[p]}\right)  \right| \leq \frac{f_\pm(x)}{\left|k \right|^{p+1} },\quad x\in\mathbb{R},\,\,\left|k \right| \geq 2,
		\end{equation*}
        where $f_\pm(x)$ are bounded smooth positive functions of $x\in\mathbb{R}$ with rapid decay as $x\to \pm \infty$, respectively, and
		\begin{equation*}
			\left( X^A_{\pm}\right) ^{[p]}(x,k)=I+\frac{ X_{\pm1}^A(x)}{k}+\frac{ X_{\pm2}^A(x)}{k^2}+\cdots+\frac{ X_{\pm p}^A(x)}{k^p}.
		\end{equation*}

	\end{prop}

	\begin{prop}\label{prop_sA}
		Given that $q_0, r_0\in \mathcal{S}(\mathbb{R})$, the matrix-valued function $s^A( k )$ defined by equation \eqref{sak} exhibits the following  properties:
		\begin{enumerate}
			\item The entries of $s^A(k)$ are defined and continuous for
			\begin{equation}\label{sak_define}
				k\in \begin{pmatrix}
					\overline{\mathbb{C}_-} & \mathbb{R} & \mathbb{R}\\
					\mathbb{R} & \mathbb{R} & \mathbb{R}\\
					\mathbb{R} & \mathbb{R} & \overline{\mathbb{C}_+}\\
				\end{pmatrix}\setminus\{0\}.
			\end{equation}
			The function $s^A_{11}(k)$ is analytic on the upper half complex plane $({\rm{Im}}\, k <0)$, and the function $s^A_{33}(k)$ is analytic on the lower half complex plane $({\rm{Im}}\, k >0)$. The remaining elements of $s^A(k)$ are not analytic in any region.

            \item $s^A(k)$ approaches the identity matrix as $k\to\infty$, and the following equation holds for $j=1,2,\cdots,N$:
            \begin{equation*}
                \left|\partial_k^j\left(s^A(k)-I-\sum_{j=1}^N\frac{s^A_j}{k^j}\right)\right|=\mathcal{O}\left(k^{-N-1}\right),\quad k\to\infty,\,\, k\,\,\rm{in} \,\,\eqref{sak_define},
            \end{equation*}
             where $\{s^A_j\}_{j=1}^\infty$ are diagonal matrices.

			\item  For $k\in\mathbb{R}\setminus\{0\}$, $\det s(k)=1$, and the function $s^A( k )$ satisfies the symmetries
			\begin{equation}\label{sym_sA}
				\begin{aligned}
					&\left(s^A\right)^\dagger\left(  \bar k \right)=\mathcal{A}^{-1}s^{T} (k)\mathcal{A},\\
					& s^A(-k)=\mathcal{B}s^A(k)\mathcal{B}.
				\end{aligned}
			\end{equation}

			\item If $q_0$, $r_0$ have compact support, then $s^A(k)$ is defined and analytic for $k\in\mathbb{C}\setminus\{0\}$, $\det s^A(k)=1$, and 
        \begin{equation*}
		X^A_+(x, k )=X^A_-(x, k )\re^{-x\widehat{\mathcal{U}}( k )}s^A( k ),\quad  k \in \mathbb{C}\setminus \left\lbrace 0 \right\rbrace.
	\end{equation*}
		\end{enumerate}
	\end{prop}
	
   For $k \in \mathbb{C}_\pm$, we write $M(x,k) = M_\pm(x,k)$, where $M_+(x,k)$ and $M_-(x,k)$ are matrix-valued solutions of equation \eqref{lax_X}, defined by the following Fredholm integral equations:
	\begin{equation}\label{Mn_fredholm}
			\left( M_\pm\right)_{ij}(x,k) =\delta_{ij}+\int_{\gamma _{ij}^\pm}\left( \re^{(x-y)\widehat{\mathcal{U}}(k)}\left(U_1M_\pm \right) \right)_{ij}\rd y,
		\qquad i,j=1,2,3, 
	\end{equation}
	where
	\begin{equation*}\delta_{ij}=\left\lbrace\begin{aligned}
			&1,\quad i=j,\\
			&0,\quad i\neq j,\\
		\end{aligned} \right. \qquad
		\gamma _{ij}^+=\left\lbrace
		\begin{aligned}
			&\left(-\infty,x \right),\quad i< j, \\
			&\left(+\infty,x \right),\quad i\geq j,  
		\end{aligned} \right. \qquad
		\gamma _{ij}^-=\left\lbrace
		\begin{aligned}
			&\left(-\infty,x \right),\quad i> j, \\
			&\left(+\infty,x \right),\quad i\leq j. 
		\end{aligned} \right. \qquad
	\end{equation*}
The contours $\gamma _{ij}^\pm$ are selected such that, for $k\in \mathbb{C}_\pm$ and $y\in \gamma _{ij}^\pm$, the exponents in the off-diagonal entries of the matrix $\re^{(x-y)\widehat{\mathcal{U}}(k)}\left(U_1M\right) $ appearing in the integral equation \eqref{Mn_fredholm} are bounded. The definition of $M_\pm$ via \eqref{Mn_fredholm} can be extended continuously to the boundary of $\mathbb{C}_\pm$. The following proposition demonstrates that every entry of $M_\pm$ is well-defined for $k\in \mathbb{C}_\pm\setminus \mathcal{Q}$, where
\[
\mathcal{Q}=\mathcal{Z}\cup \{0\},
\]
and $\mathcal{Z}$ denotes the set of zeros of the Fredholm determinants associated with \eqref{Mn_fredholm}.
	
	\begin{prop}\label{prop_Mn_basic}
	Suppose $q_0,\, r_0\in\mathcal{S}(\mathbb{R})$, the equation \eqref{Mn_fredholm} uniquely define the two solutions $M_\pm(x,k)$ in equation \eqref{lax_X}, which satisfy the following properties for $ x \in \mathbb{R}$:
	\begin{enumerate}
        \item The function $M_\pm(x,k)$ is defined for $x\in \mathbb{R}$ and $k\in \overline{\mathbb{C}_\pm}\setminus \mathcal{Q}$. For each fixed $k\in \overline{\mathbb{C}_\pm}\setminus \mathcal{Q}$, $M_\pm(\cdot, k)$ is smooth and satisfies equation \eqref{lax_X}. For each fixed $x\in\mathbb{R}$, the matrix-valued function $M_\pm(x, \cdot)$ is continuous on $\overline{\mathbb{C}_\pm}\setminus \mathcal{Q}$ and analytic on $\mathbb{C}_\pm \setminus \mathcal{Q}$.
		\item For each $\epsilon>0$, there exists $C:=C(\epsilon)$ such that
		\begin{equation*}
			\left|M_\pm(x,k) \right| \leq C,\quad  k\in \overline{\mathbb{C}_\pm},\,\, {\rm{dist}}(k,\mathcal{Q})\geq \epsilon. 
		\end{equation*}
		\item For each $j\in \mathbb{N}_+$, the partial derivative $\frac{\partial^j M_\pm}{\partial k^j}(x, \cdot)$ has a continuous extension to $ \overline{\mathbb{C}_\pm} \setminus \mathcal{Q}$.
		\item For $k \in \overline{\mathbb{C}_\pm} \setminus \mathcal{Q}$ and $x\in\mathbb{R}$, $\det M_\pm(x, k) = 1$ .
		\item The sectionally analytic function $ M(x, k) = M_\pm(x, k) $ for $k \in \mathbb{C}_\pm$ satisfies the symmetries
        \begin{equation*}\label{M-sym}
            \begin{aligned}
			&M^{-1}(x, k) = \mathcal{A}^{-1}M^\dagger(x, \bar k) \mathcal{A},\\
            &M(x, k)=\mathcal{B} M(x,-k) \mathcal{B},
		\end{aligned}  \qquad k \in \mathbb{C} \setminus \mathcal{Q}.
        \end{equation*}
		\end{enumerate}
	\end{prop}
	
	\begin{lem}\label{M_kinfty}
		Suppose $q_0,\, r_0\in\mathcal{S}(\mathbb{R})$ and $q_0,\, r_0\not\equiv0$. For any integer $p>0$, the function $X_+^{[p]}$ defined in \eqref{Xp_infty} such that
		\begin{equation*}
			\left| M(x,k)-X_+^{[p]}(x,k)\right| \leq \frac{C}{\left| k\right| ^{p+1}},\qquad x\in\mathbb{R},\,\,k\in \mathbb{C}\setminus\mathbb{R},\,\,\left| k\right| \geq 2.
		\end{equation*}
	\end{lem}

	\begin{lem}\label{lem_M_ST}
		Suppose $q_0,\, r_0\in\mathcal{S}(\mathbb{R})$ have compact support. Then
		\begin{equation*}\label{def_ST}
			\begin{aligned}
				M_\pm(x,k)=&X_-(x,k)\re^{x\widehat{\mathcal{U}}(k)}S_\pm(k)=X_+(x,k)\re^{x\widehat{\mathcal{U}}(k)}T_\pm(k),\\
			\end{aligned}
			\qquad  x\in\mathbb{R},\,\,k\in \overline{\mathbb{C}_\pm}\setminus\mathcal{Q},
		\end{equation*}
	where $S_\pm(k)$ and $T_\pm(k)$  are expressed by
\begin{equation*}\begin{aligned}
		S_+(k) &=\begin{pmatrix}
			s_{11}	&	0	&0\\
			s_{21}	&\frac{m_{33}(s)}{s_{11}}	&0\\
			s_{31}	&\frac{m_{23}(s)}{s_{11}}&\frac{1}{m_{33}(s)}
		\end{pmatrix},&&
		S_-(k) =\begin{pmatrix}
			\frac{1}{m_{11}(s)}	&\frac{m_{21}(s)}{s_{33}}		&s_{13}\\
			0		&\frac{m_{11}(s)}{s_{33}}	&s_{23}\\
			0	&0&s_{33}
		\end{pmatrix},\\
		T_+(k) &=\begin{pmatrix}
			1	&	-\frac{s_{12}}{s_{11}}	&\frac{m_{31}(s)}{m_{33}(s)}\\
			0	&1	&-\frac{m_{32}(s)}{m_{33}(s)}\\
			0	&0&1
		\end{pmatrix},&&
		T_-(k) =\begin{pmatrix}
			1	&0		&0\\
			-\frac{m_{12}(s)}{m_{11}(s)}		&1	&0\\
			\frac{m_{13}(s)}{m_{11}(s)}	&-\frac{s_{32}}{s_{33}}&1
		\end{pmatrix}.
\end{aligned}\end{equation*}

	\end{lem}
	
	Let $\eta \in C_c^\infty(\mathbb{R})$ be a smooth cutoff function satisfying $\eta(x)=1$ for $|x|\le 1$ and $\eta(x)=0$ for $|x|\ge 2$. For $j\ge 1$, set $\eta_j(x):=\eta(x/j)$. Given any $f\in \mathcal{S}(\mathbb{R})$, the family $\{\eta_j f\}_{j=1}^\infty$ consists of smooth functions with compact support. Moreover, this family converges to $f$ in the $\mathcal{S}(\mathbb{R})$ topology as $j\to\infty$.

	\begin{lem}\label{lem_Cauchy sequence}
		 Let $q_0,\, r_0\in\mathcal{S}(\mathbb{R})$. Let ${s(k), M_\pm(x,k)}$ and ${s^{(i)}(k), M_\pm^{(i)}(x,k)}$ be the spectral functions and eigenfunctions associated with $(q_0, r_0)$ and
	$
			\left(q_0^{(i)}(x), r_0^{(i)}(x)\right) := (\eta_i q_0, \eta_i r_0) \in \mathcal{S}(\mathbb{R}) \times \mathcal{S}(\mathbb{R}), 
	$
		respectively. Then
		
		\begin{equation*}
			\begin{aligned}
			&\lim_{i \to \infty} s^{(i)}(k) = s(k), \quad && k\,\, {\rm{in}}\,\, \eqref{sk_k_define},\\
			&\lim_{i \to \infty} \left(s^A\right)^{(i)}(k) = s^A(k), \quad && k\,\, {\rm{in}}\,\, \eqref{sak_define},\\
			&\lim_{i \to \infty} X_\pm^{(i)}(x,k) = X_\pm(x,k), \quad && x\in\mathbb{R},\,\,k\in (\overline{\mathbb{C}_\pm},\mathbb{R},\overline{\mathbb{C}_\mp})\setminus\{0\},\\
			&\lim_{i \to \infty} M_\pm^{(i)}(x,k) = M_\pm(x,k), \quad &&x \in \mathbb{R}, \,\, k \in \overline{\mathbb{C}_\pm} \setminus \mathcal{Q}.
			\end{aligned}
		\end{equation*}
	
	\end{lem}

	\begin{lem}\label{lem_M_jump}
		Suppose $q_0,\, r_0\in\mathcal{S}(\mathbb{R})$. For each $x\in\mathbb{R}$, $M(x,k)$ satisfies the jump condition
		\begin{equation*}
			M_+(x,k)=M_-(x,k)v(x,0,k),\qquad k\in \mathbb{R}\setminus\mathcal{Q},
		\end{equation*}
	where $v$ is the jump matrix defined in equation \eqref{jump_0}.
	\end{lem}
	
	\begin{proof}
		If $q_0,\, r_0\in\mathcal{S}(\mathbb{R})$ are compactly supported, then we have
		
		\begin{equation*}
			v(x,0,k)=\re^{x\widehat{\mathcal{U}}(k)}\left(S_-^{-1}(k)S_+(k) \right) :=\re^{x\widehat{\mathcal{U}}(k)}J(k),
		\end{equation*}
		where
		\begin{equation*}
			\begin{aligned}
				J(k)&=\begin{pmatrix}
					1 & 0 & 0\\
					\frac{m_{12}(s)}{m_{11}(s)} & 1 & 0\\
					\frac{s_{31}}{s_{33}} & \frac{s_{32}}{s_{33}} & 1\\
				\end{pmatrix}\begin{pmatrix}
				1 & -\frac{s_{12}}{s_{11}} & \frac{m_{31}(s)}{m_{33}(s)}\\
				0 & 1 & -\frac{m_{32}(s)}{m_{33}(s)}\\
				0 & 0 & 1\\
				\end{pmatrix}=\begin{pmatrix}
				1 & -\frac{s_{12}}{s_{11}} & \frac{m_{31}(s)}{m_{33}(s)}\\
				\frac{m_{12}(s)}{m_{11}(s)} & 1-\frac{s_{12}m_{12}(s)}{s_{11}m_{11}(s)} & -\frac{s_{23}}{m_{11}(s)m_{33}(s)}\\
				\frac{s_{31}}{s_{33}} & \frac{m_{23}(s)}{s_{11}s_{33}} & \frac{1}{s_{33}m_{33}(s)}\\
				\end{pmatrix}.
			\end{aligned}
		\end{equation*}
		According to the symmetry conditions satisfied by $s(k)$ and $s^A(k)$, denoted as equations \eqref{sym_s} and \eqref{sym_sA}, we can derive
	\begin{equation*}
			\begin{aligned}
				J(k)&=\begin{pmatrix}
   				1 &  -r_1(k) & r_2(k)\\
   				2\sigma r_1^*(k) & 1-2\sigma \left|r_1(k) \right|^2 &-2\sigma \alpha(k)\\
   				r_2^*(k) & \alpha^*(k) & 1-2\sigma \left|r_1(-k) \right|^2+ \left|r_2(k) \right|^2\\
   			\end{pmatrix},
			\end{aligned}
		\end{equation*}
        where $\alpha(k)=r_1^*(-k)+ r_1^*(k)r_2(k)$, $r_1(k)$ and $r_2(k)$ are defined by equation \eqref{r1_r2}.
        
       Suppose $q_0, r_0 \in \mathcal{S}(\mathbb{R})$ are not compactly supported. We then consider a sequence $\{(q_0^{(i)}, r_0^{(i)})\}_{i=1}^\infty$ of smooth, compactly supported functions converging to $(q_0, r_0)$. By virtue of Lemma \ref{lem_Cauchy sequence}, we may pass to the limit $i \to \infty$ in the identities satisfied by $(q_0^{(i)}, r_0^{(i)})$.
		
	\end{proof}
	
	\begin{lem}\label{lem_M_pm}
		Suppose $q_0,\, r_0\in\mathcal{S}(\mathbb{R})$. The functions $M_\pm$ can be expressed in terms of the entries of $X_\pm$, $X_\pm^A$, $s$ and $s^A$ as follows:
		\begin{equation*}
			M_+=\begin{pmatrix}
				X_{+,11} & \frac{X^A_{+,23}X^A_{-,31}-X^A_{+,33}X^A_{-,21}}{s_{11}} & \frac{X_{-,13}}{s^A_{33}}\\
				X_{+,21} & \frac{X^A_{+,33}X^A_{-,11}-X^A_{+,13}X^A_{-,31}}{s_{11}} & \frac{X_{-,23}}{s^A_{33}}\\
				X_{+,31} & \frac{X^A_{+,13}X^A_{-,21}-X^A_{+,23}X^A_{-,11}}{s_{11}} & \frac{X_{-,33}}{s^A_{33}}\\
			\end{pmatrix},\quad 	M_-=\begin{pmatrix}
			\frac{X_{-,11}}{s^A_{11}} & \frac{X^A_{+,31}X^A_{-,23}-X^A_{+,21}X^A_{-,33}}{s_{33}} & X_{+,13}\\
			\frac{X_{-,21}}{s^A_{11}} & \frac{X^A_{+,11}X^A_{-,33}-X^A_{+,31}X^A_{-,13}}{s_{33}} & X_{+,23}\\
			\frac{X_{-,31}}{s^A_{11}} & \frac{X^A_{+,21}X^A_{-,13}-X^A_{+,11}X^A_{-,23}}{s_{33}} & X_{+,33}\\
			\end{pmatrix},
		\end{equation*}
		for  all $x\in\mathbb{R}$ and $k\in\overline{ \mathbb{C}_\pm}\setminus\mathcal{Q}$.
		
	\end{lem}

	\begin{lem}\label{lem_M_sing}
		Suppose $q_0,\, r_0\in\mathcal{S}(\mathbb{R})$ such that Assumptions \ref{assu_solitonless} and \ref{assu_LT} hold, then Proposition \ref{prop_Mn_basic} and Lemmas \ref{lem_M_ST}, \ref{lem_Cauchy sequence}, \ref{lem_M_jump} remain valid after replacing $\mathcal{Q}$ with $\{0\}$.
	\end{lem}
	
   We next examine how $M$ evolves over time. By substituting $U_1(x, k)$ in the integral equation \eqref{Mn_fredholm} with its time-dependent counterpart $U_1(x, t, k)$, we obtain two time-dependent eigenfunctions denoted $M_\pm(x, t, k)$. We then define $M(x, t, k) = M_\pm(x, t, k)$ for $k \in \mathbb{C}_\pm$, which yields a piecewise analytic function $M(x, t, k)$.
   \begin{lem}\label{lem_Mxtk_smooth}
   	The function $M_\pm(x,t,k)$ is a smooth function of $(x,t)\in\mathbb{R}\times \left[ 0,T\right) $ for $k\in\overline{\mathbb{C}_\pm}\setminus\{0\}$ satisfying equations \eqref{lax_X}.
   \end{lem}
   
   \begin{lem}\label{lem_Mxtk_jump}
   	For each $(x,t)\in\mathbb{R}\times \left[ 0,T\right) $, the function $M(x,t,k)$ is analytic in $k\in\mathbb{C}\setminus \mathbb{R} $ and admits continuous limits onto $\mathbb{R}\setminus\{0\}$ from both half-planes. Furthermore, when $k$ tends to $\mathbb{R}$ from above and below, $M(x,t,k)$ obeys the jump relation \eqref{jump_0}.
    
   \end{lem}
	
	Thus we have completed the construction of RH problem \ref{rhp_M}.

	\section{Long-time asymptotics}\label{sec_LT}\ \ \ \
   The goal of this section is to present the asymptotic behavior of the solution to the initial value problem \eqref{Newell_initial} of the Newell equation as the time variable $t\to\infty$, under the Assumption \ref{assu_solitonless} of no solitons, i.e., $s_{11}(k)\neq0$ for $k\in\mathbb{C}_+$.

    \subsection{Long-time asymptotics in Region {\rm{II}}}\label{subsec_regionII}
    \ \ \ \
       Consider performing the Deift-Zhou nonlinear steepest descent analysis on the initial RH problem \ref{rhp_M}. First, for the stationary phase functions $\Phi_{ij}(\zeta,k)$, $1\leq i<j\leq3$ appearing in the jump matrix $v(x,t,k)$ in equation \eqref{jump_0}:
        \begin{align*}\label{theta}
            &\Phi_{12}(\zeta,k):=2\ri k\left(k+\zeta-1\right),\quad\Phi_{13}(\zeta,k):=4\ri k\left(\zeta-1\right),\quad \Phi_{23}(\zeta,k):=-2\ri k\left(k+1-\zeta\right),
        \end{align*}
      it is computed that $\Phi_{23}(\zeta,k)$ has a single zero $k_0=(\zeta-1)/2$ and $\Phi_{12}(\zeta,k)$ has a single zero $-k_0=(1-\zeta)/2$ for $\zeta\in\mathcal{I}_2$ and $t>0$. In Figure \ref{fig_theta_sign1}, the  signature tables of the real parts of the three phase functions are given for $(x,t)$ in Region \rm{II}.
    \begin{figure}[htbp]
    \centering
    \begin{subfigure}[t]{0.28\textwidth}
        \centering
     	\begin{tikzpicture} [scale=0.7]
        	\definecolor{mycolor}{HTML}{f5d5d8}
        	
         \fill[mycolor] (-3,0) -- (-1.5,0) -- (-1.5,3) -- (-3,3) -- cycle;
        \fill[mycolor] (-1.5,0) -- (3,0) -- (3,-3) -- (-1.5,-3) -- cycle;
	 		
	  \draw [very thick,black](-3,0) -- (3,0);
        \draw [very thick,black](-1.5,-3) -- (-1.5,3);

        \fill (-1.5,0) circle (1.5pt);
        \fill (0,0) circle (1.5pt);
	  \fill (1.5,0) circle (1.5pt);

        \node[below] at (-2,0) {$-k_0$};
	  \node[below] at (0,0) {$\text{0}$};
	  \node[below] at (1.5,0) {$k_0$};

      \node at (-2.4,2) {$D_2^{(12)}$};
      \node at (1,2) {$D_1^{(12)}$};
      \node at (1,-2) {$D_4^{(12)}$};
      \node at (-2.4,-2) {$D_3^{(12)}$};

        \node[right] at (3,0) {$\mathbb{R}$};
     	\end{tikzpicture}
     		\caption{The signature of $\Phi_{12}$.}
    \end{subfigure}
    \quad
    \begin{subfigure}[t]{0.28\textwidth}
        \centering
     	\begin{tikzpicture} [scale=0.7]
        
        	\definecolor{mycolor}{HTML}{f5d5d8}

         \fill[mycolor] (-3,0) -- (3,0) -- (3,-3) -- (-3,-3) -- cycle;
	 		
	  \draw [very thick,black](-3,0) -- (3,0);

        \fill (-1.5,0) circle (1.5pt);
        \fill (0,0) circle (1.5pt);
	  \fill (1.5,0) circle (1.5pt);

        \node[above] at (-1.5,0) {$-k_0$};
	  \node[above] at (0,0) {$\text{0}$};
	  \node[above] at (1.5,0) {$k_0$};
      
      \node at (0,2) {$D_1^{(13)}$};
      \node at (0,-2) {$D_2^{(13)}$};

        \node[right] at (3,0) {$\mathbb{R}$};
     	\end{tikzpicture}
     		\caption{The signature of $\Phi_{13}$.\label{subfig_theta13_1}}
    \end{subfigure}
    \quad
    \begin{subfigure}[t]{0.28\textwidth}
        \centering
     	\begin{tikzpicture} [scale=0.7]
        \definecolor{mycolor}{HTML}{f5d5d8}
        	
         \fill[mycolor] (3,0) -- (1.5,0) -- (1.5,3) -- (3,3) -- cycle;
        \fill[mycolor] (1.5,0) -- (-3,0) -- (-3,-3) -- (1.5,-3) -- cycle;
	 		
	  \draw [very thick,black](-3,0) -- (3,0);
        \draw [very thick,black](1.5,-3) -- (1.5,3);

        \fill (-1.5,0) circle (1.5pt);
        \fill (0,0) circle (1.5pt);
	  \fill (1.5,0) circle (1.5pt);

        \node[below] at (-1.5,0) {$-k_0$};
	  \node[below] at (0,0) {$\text{0}$};
	  \node[below] at (1.8,0) {$k_0$};

      \node at (2.4,2) {$D_1^{(23)}$};
      \node at (-1,2) {$D_2^{(23)}$};
      \node at (-1,-2) {$D_3^{(23)}$};
      \node at (2.4,-2) {$D_4^{(23)}$};

        \node[right] at (3,0) {$\mathbb{R}$};
     	\end{tikzpicture}
     		\caption{The signature of $\Phi_{23}$.}
             \end{subfigure}
             \caption{ Open sets in the complex $k$-plane for Region \rm{II}: $\rre \Phi_{ij}>0$ (shaded) and $\rre \Phi_{ij}<0$ (white).}
             \label{fig_theta_sign1}
\end{figure}

The next step is to decompose the jump on $\mathbb{R}$, with the aim of obtaining the new jump that approaches the identity matrix everywhere except near the two critical points $\pm k_0$. Prior to this, it is necessary to perform analytical approximation for the functions $r_1(k)$, $r_2^*(k)$, $\hat{r}_1( k)=\frac{r_1(k)}{1-2\sigma|r_1(k)|^2}$, and $\hat{\alpha}(k)=\frac{\alpha(k)}{1-2\sigma|r_1(k)|^2}$.

\begin{lem}\label{lem_r_decomposition}
For each $\zeta\in\mathcal{I}_2$ and $t>0$, there are several decompositions:
    \begin{align*}
        &r_1(k)=r_{1,a}(x,t,k)+r_{1,r}(x,t,k),\quad && k\in \left[-k_0,\infty\right),\\
        &\hat{r}_1(k)=\hat{r}_{1,a}(x,t,k)+\hat{r}_{1,r}(x,t,k),\quad && k\in \left(-\infty,-k_0\right],\\
        &\hat{\alpha}(k)=\hat{\alpha}_{a}(x,t,k)+\hat{\alpha}_{r}(x,t,k),\quad && k\in \left(-\infty,-k_0\right],\\
        &r_2(k)=r_{2,a}(x,t,k)+r_{2,r}(x,t,k),\quad && k\in \mathbb{R},
    \end{align*}
    where the above functions satisfy the following properties:
    \begin{itemize}
     \item $r_{1,a}(x,t,k)$ is defined and continuous for $k\in \overline{D_1^{(12)}}$, and is analytic for $k\in D_1^{(12)}$. $\hat{r}_{1,a}(x,t,k)$ is defined and continuous for $k\in \overline{D_3^{(12)}}$, and is analytic for $k\in D_3^{(12)}$. $\hat{\alpha}_{a}(x,t,k)$ is defined and continuous for $k\in \overline{D_2^{(12)}}$, and is analytic for $k\in D_2^{(12)}$. $r_{2,a}(x,t,k)$ is defined and continuous for $k\in \overline{\mathbb{C}_+}$, and is analytic for $k\in \mathbb{C}_+$. 
    \item The functions $r_{1,a}$, $\hat{r}_{1,a}$, $\hat{\alpha}_{a}$ and $r_{2,a}$ satisfy
        \begin{align*}
            &\left|r_{1,a}(x,t,k)-\sum_{j=0}^{N}\frac{r^{(j)}_{1,a}(-k_0)}{j!}(k+k_0)^j\right|\leq C \left|k+k_0\right|^{N+1}\re^{\frac{t}{4}\left|\rre\Phi_{12}(\zeta,k)\right|},\quad k\in \overline{D_1^{(12)}},\\
            &\left|\hat{r}_{1,a}(x,t,k)-\sum_{j=0}^{N}\frac{\hat{r}^{(j)}_{1,a}(-k_0)}{j!}(k+k_0)^j \right|\leq C \left|k+k_0\right|^{N+1}\re^{\frac{t}{4}\left|\rre\Phi_{12}(\zeta,k)\right|},\quad  k\in \overline{D_3^{(12)}},\\
            &\left|\hat{\alpha}_{a}(x,t,k)-\sum_{j=0}^{N}\frac{\hat{\alpha}^{(j)}_{a}(-k_0)}{j!}(k+k_0)^j \right|\leq C \left|k+k_0\right|^{N+1}\re^{\frac{t}{4}\left|\rre\Phi_{23}(\zeta,k)\right|},\quad  k\in \overline{D_2^{(12)}},\\
            &\left|r_{2,a}(x,t,k)-\sum_{j=0}^{N}\frac{r^{(j)}_{2,a}(0)}{j!}k^j\right|\leq C \left|k\right|^{N+1}\re^{\frac{t}{4}\left|\rre\Phi_{13}(\zeta,k)\right|},\quad k\in \overline{\mathbb{C}_+},
            \end{align*}
            and
            \begin{align*}
            & \left|r_{1,a}(x,t,k)\right|\leq  \frac{C}{1+\left|k\right|}\re^{\frac{t}{4}\left|\rre\Phi_{12}(\zeta,k)\right|}, \quad k\in \overline{D_1^{(12)}},\\
            & \left|\hat{r}_{1,a}(x,t,k)\right|\leq  \frac{C}{1+\left|k\right|}\re^{\frac{t}{4}\left|\rre\Phi_{12}(\zeta,k)\right|},\quad  k\in \overline{D_3^{(12)}},\\
            & \left|\hat{\alpha}_{a}(x,t,k)\right|\leq  \frac{C}{1+\left|k\right|}\re^{\frac{t}{4}\left|\rre\Phi_{23}(\zeta,k)\right|},\quad  k\in \overline{D_2^{(12)}},\\
            & \left|r_{2,a}(x,t,k)\right|\leq  \frac{C}{1+\left|k\right|}\re^{\frac{t}{4}\left|\rre\Phi_{13}(\zeta,k)\right|},\quad k\in \overline{\mathbb{C}_+},
            \end{align*}

        where the constant $C$ is independent of $\zeta,t,k$.
    \item For each $1\leq p\leq \infty$,
    \begin{align*}
        &\left\|(1+|\cdot|) r_{1,r}(x,t,\cdot)\right\|_{L^p\left(-k_0,\infty\right)}=\left\|\frac{r_{1,r}(x,t,\cdot)}{\cdot+k_0}\right\|_{L^p\left(-k_0,\infty\right)}=\mathcal{O}\left(t^{-N}\right),\\
        &\left\|(1+|\cdot|) \hat{r}_{1,r}(x,t,\cdot)\right\|_{L^p\left(-\infty,-k_0\right)}=\left\|\frac{\hat{r}_{1,r}(x,t,\cdot)}{\cdot+k_0}\right\|_{L^p\left(-\infty,-k_0\right)}=\mathcal{O}\left(t^{-N}\right),\\
        &\left\|(1+|\cdot|) \hat{\alpha}_{r}(x,t,\cdot)\right\|_{L^p\left(-\infty,-k_0\right)}=\left\|\frac{\hat{\alpha}_{r}(x,t,\cdot)}{\cdot+k_0}\right\|_{L^p\left(-\infty,-k_0\right)}=\mathcal{O}\left(t^{-N}\right),\\
        &\left\|(1+|\cdot|) {r}_{2,r}(x,t,\cdot)\right\|_{L^p(\mathbb{R})}=\mathcal{O}\left(t^{-N}\right),
    \end{align*}
    uniformly for $\zeta\in\mathcal{I}_2$ as $t\to\infty$.
        \end{itemize}
\end{lem}

\begin{proof}
    The proof can be found in References \cite{DZ_1993} and \cite{Lenells_2017}.
\end{proof}

Note that the reflection coefficient $\hat{r}_1(k)$ is introduced in Lemma \ref{lem_r_decomposition} because we will subsequently construct a scalar RH problem concerning the function $\delta(k)$. By applying a matrix transformation composed of $\delta$, the initial jump matrix $v(x,t,k)$ in \eqref{jump_0} can be decomposed. This scalar RH problem is as follows:
     \begin{equation}\label{delta}
    	\left\{
    	\begin{aligned}
    		&\begin{aligned}\delta_{+}(k)
    			=&\delta_{-}(k)(1-2\sigma\left|r_1(k)\right|^2),\quad  &&k\in \Sigma^{(1)}_1,\\
    			=&\delta_{-}(k), &&k\in\mathbb{C}\setminus\Sigma^{(1)}_1,
    		\end{aligned}\\
            &\delta(k)\rightarrow1,\qquad\qquad\qquad\qquad\quad\quad\! k\rightarrow\infty.
    	\end{aligned}
    	\right.
    \end{equation}
    The function $\delta(k)$ defined in equation \eqref{delta} can be expressed as the following integral equation by Plemelj formulas:
    \begin{equation}\label{vol_delta}
        \delta(\zeta,k)={\rm{exp}}\left\{\frac{1}{2 \pi\ri}\int_{-\infty}^{-k_0}\frac{\ln \left(1-2\sigma\left|r_1(s)\right|^2\right)}{s-k}{\rm{d}}s\right\},\quad k\in \mathbb{C}\backslash \Sigma^{(1)}_1.
    \end{equation}
    \begin{figure}[htbp]
    	\centering
    	\begin{tikzpicture}[scale=1]
        
        \draw[very thick, black!20!blue] (-4,0) -- (4,0);
    		
    		\draw[very thick, black!20!blue, -latex] (-4,0) -- (-2.5,0);
            \draw[very thick, black!20!blue, -latex] (-4,0) -- (0.2,0);
            \draw[very thick, black!20!blue, -latex] (-4,0) -- (3,0);

    		\fill (-1.5,0) circle (1.5pt);
            \fill (0,0) circle (1.5pt);
	      \fill (1.5,0) circle (1.5pt);

           \node[below] at (-1.5,0) {$-k_0$};
	     \node[below] at (0,-0.2) {$\text{0}$};
	     \node[below] at (1.5,0) {$k_0$};
         
           \node[red!70!black,above] at (-3,0.2) {$\Sigma^{(1)}_1$};
	     \node[red!70!black,above] at (0,0.2) {$\Sigma^{(1)}_2$};
          \node[red!70!black,above] at (3,0.2) {$\Sigma^{(1)}_3$};

          \node[right] at (4,0) {$\rre k$};
    		
    	\end{tikzpicture}
    	\caption{The jump contour $\Sigma^{(1)}$ in the complex $k$-plane.}
    	\label{fig_Gamma1}
    \end{figure}
    
Let $\ln_0(k)=\ln|k|+\ri \arg_0k$ and  $\ln_\pi(k)=\ln|k|+\ri \arg_\pi k$ with $\arg_0 k\in (0,2\pi)$ and $\arg_\pi k\in (-\pi,\pi)$. 
    \begin{lem}\label{lem_delta}
    	The function $\delta(\zeta,k)$ satisfies the following properties:
    	\begin{enumerate}
     		\item  The functions 
            $\delta^{\pm 1}(\zeta,k)$ are analytic on $\mathbb{C}\setminus\left(-\infty,-k_0\right] $, and $\delta$ can be written as
			\begin{equation*}
				\delta(\zeta,k)={\rm{e}}^{\ri\nu\ln_\pi(k+k_0)}{\rm{e}}^{-\chi(\zeta,k)},\quad 1<\zeta<\infty,
			\end{equation*}
			where 
			\begin{equation}\label{nu_chi}
				\begin{aligned}
					\nu=-\frac{1}{2\pi}\ln(1-2\sigma\left|r_1(-k_0) \right|^2 ),\quad \chi(\zeta,k)=\frac{1}{2\pi\ri}\int_{-\infty}^{-k_0}\ln_\pi(k-s){\rm{d}}\ln(1-2\sigma\left|r_1(s)\right|^2 ).
				\end{aligned}
			\end{equation}
    	\item For each $\zeta\in\mathcal{I}_2$, $\delta(k)$ is bounded in $\mathbb{C}\setminus\left(-\infty,-k_0\right] $ and satisfies $\delta^{-1}(k)= \overline{\delta(\overline{k})}$.
        \item As $k\to -k_0$ along the non-tangential direction of $(-\infty,-k_0)$, we have the following equation holds, for $C$ independent of $\zeta$:
        \begin{align*}
            &\left|\chi(\zeta,k)-\chi({\zeta,-k_0})\right|\leq C|k+k_0|\left(1+|\ln|k+k_0||\right).
             \end{align*}
            
    	\end{enumerate}
    \end{lem}

    \begin{proof}
        The lemma is derived from equation \eqref{vol_delta}, and can be proved by direct estimation.
    \end{proof}

   Next, we introduce the matrix transformation:
    \begin{equation}\label{trans_Delta}
        M^{(1)}(x,t,k)=M(x,t,k)\Delta(k), \quad \Delta(k)=\begin{pmatrix}
        	\delta(k) & 0 & 0\\
        	0 & \frac{1}{{\delta(k)}{\delta(-k)}} & 0\\
        	0 & 0 & {\delta(-k)}\\
        \end{pmatrix}.
    \end{equation}
    Based on the properties satisfied by the $\delta(\zeta,k)$ function in Lemma \ref{lem_delta}, it can be deduced that the functions  $\Delta^{\pm 1}$ are uniformly bounded with respect to $\zeta\in\mathcal{I}_2$ and $k\in \mathbb{C}\setminus\left(\Gamma^{(1)}_1\cup \Gamma^{(1)}_3\right)$, and the following holds:
    \begin{equation*}\label{Delta_bound}
    	\begin{aligned}
    	 &\Delta(\zeta,k)=I+\mathcal{O}\left(k^{-1}\right),\quad k\to\infty,\\
    	 &\Delta^\dagger(\zeta,\bar k)=\mathcal{A}\Delta^{-1}(\zeta,k)\mathcal{A}^{-1},\quad \Delta(\zeta,k)=\mathcal{B}\Delta(\zeta,-k)\mathcal{B}.
    	\end{aligned}
    \end{equation*}

    Under the transformation \eqref{trans_Delta}, the new jump matrices $V_j^{(1)}(x,t,k):=V^{(1)}_{j,l}(x,t,k)V^{(1)}_{j,r}(x,t,k)$  for $k\in \Sigma_j^{(1)}$, $j=1,2,3,$ are obtained, which can be written in the following form:
        \begin{align*}
           & V^{(1)}_{1,l}
            =\begin{pmatrix}
					1 & -\hat{r}_1(k)\delta_{21-}^{-1}(k)\re^{t\Phi_{12}} & 0\\
					0 & 1 & 0\\
					-r_2(-k)\frac{\delta_{21-}(k)}{\delta_{23-}(k)}\re^{-t\Phi_{13}} & \hat{\alpha}^*(k)\delta_{23-}^{-1}(k)\re^{-t\Phi_{23}} & 1\\
				\end{pmatrix},\nonumber\\
                &V^{(1)}_{1,r}= \begin{pmatrix}
					1 & 0 & -r_2^*(-k)\frac{\delta_{23+}(k)}{\delta_{21+}(k)}\re^{t\Phi_{13}}\\
					2\sigma\hat{r}_1^*(k)\delta_{21+}(k)\re^{-t\Phi_{12}} & 1 & -2\sigma\hat{\alpha}(k)\delta_{23+}(k)\re^{t\Phi_{23}}\\
					0 & 0 & 1\\
				\end{pmatrix}\\,
                &V^{(1)}_{2,l}
                = \begin{pmatrix}
					1 & 0 & 0\\
					2\sigma r_1^*(k)\delta_{21}(k)\re^{-t\Phi_{12}} & 1 & 0\\
					r_2^*(k)\frac{\delta_{21}(k)}{\delta_{23}(k)}\re^{-t\Phi_{13}} & r_1(-k)\delta_{23}^{-1}(k)\re^{-t\Phi_{23}} & 1\\
				\end{pmatrix},\\
                &V^{(1)}_{2,r}= \begin{pmatrix}
					1 & -r_1(k)\delta_{21}^{-1}(k)\re^{t\Phi_{12}} & r_2(k)\frac{\delta_{23}(k)}{\delta_{21}(k)}\re^{t\Phi_{13}}\\
					0 & 1 & -2\sigma r_1^*(-k)\delta_{23}(k)\re^{t\Phi_{23}}\\
					0 & 0 & 1\\
				\end{pmatrix},\\ 
               & V^{(1)}_{3,l}
                = \begin{pmatrix}
					1 & 0 & 0\\
					2 \sigma r_1^*(k)\delta_{21-}(k)\re^{-t\Phi_{12}} & 1 & -2 \sigma\hat{r}_1^*(-k)\delta_{23-}(k) \re^{t\Phi_{23}}\\
					r_2^*(k)\frac{\delta_{21-}(k)}{\delta_{23-}(k)}\re^{-t\Phi_{13}} & 0 & 1\\
				\end{pmatrix}, \\
                &V^{(1)}_{3,r}= \begin{pmatrix}
					1 & -r_1(k)\delta_{21+}^{-1}(k)\re^{t\Phi_{12}} & r_2(k)\frac{\delta_{23+}(k)}{\delta_{21+}(k)}\re^{t\Phi_{13}}\\
					0 & 1 & 0\\
					0 & \hat{r}_1(-k)\delta_{23+}^{-1}(k)\re^{-t\Phi_{23} } & 1\\
				\end{pmatrix},\label{v_1_3}
            \end{align*}
            where $\delta_{21}(k)=\delta(k)^2\delta(-k)$ and $\delta_{23}(k)=\delta(k)\delta(-k)^2$.

Next, we define the following matrix function $T(\zeta,k)=T_j(\zeta,k)$ for $k\in D_j$, $j=1,2,\cdots,10$, where the open regions $D_j$ are shown in Figure \ref{fig_Gamma2}:
\begin{align*}
	&T_1(\zeta,k)=\begin{pmatrix}
		1 & \hat{\alpha}^*_{a}(-k)\delta_{21}^{-1}(k)\re^{t\Phi_{12}} & -r_{2,a}(k)\frac{\delta_{23}(k)}{\delta_{21}(k)}\re^{t\Phi_{13}}\\
		0 & 1 & 0\\
		0 & -\hat{r}_{1,a}(-k)\delta_{23}^{-1}(k)\re^{-t\Phi_{23} } & 1\\
	\end{pmatrix},\\
    & T_2(\zeta,k)=\begin{pmatrix}
	1 & r_{1,a}(k)\delta_{21}^{-1}(k)\re^{t\Phi_{12}} & r_{2,a}^*(-k)\frac{\delta_{23}(k)}{\delta_{21}(k)}\re^{t\Phi_{13}}\\
	0 & 1 & 2\sigma r_{1,a}^*(-k)\delta_{23}(k)\re^{t\Phi_{23}}\\
	0 & 0 & 1\\
	\end{pmatrix},\nonumber\\
	&T_3(\zeta,k)=\begin{pmatrix}
		1 & 0 & 0\\
		2\sigma r_{1,a}^*(k)\delta_{21}(k)\re^{-t\Phi_{12}} & 1 & 0\\
		r_{2,a}^*(k)\frac{\delta_{21}(k)}{\delta_{23}(k)}\re^{-t\Phi_{13}} & r_{1,a}(-k)\delta_{23}^{-1}(k)\re^{-t\Phi_{23}} & 1\\
	\end{pmatrix},\\
    &T_4(\zeta,k)=\begin{pmatrix}
	1 & 0 & 0\\
	2\sigma  r_{1,a}^*(k)\delta_{21}(k)\re^{-t\Phi_{12}} & 1 & -2 \sigma \hat{r}_{1,a}^*(-k)\delta_{23}(k)\re^{t\Phi_{23}}\\
r_{2,a}^*(k)\frac{\delta_{21}(k)}{\delta_{23}(k)}\re^{-t\Phi_{13}} & 0 & 1\\
	\end{pmatrix},\nonumber\\
	&T_5(\zeta,k)=\begin{pmatrix}
		1 & 0 &  r_{2,a}^*(-k) \frac{\delta_{23}(k)}{\delta_{21}(k)}\re^{t\Phi_{13}}\\
		-2\sigma \hat{r}_{1,a}^*(k)\delta_{21}(k)\re^{-t\Phi_{12}} & 1 & 2\sigma \hat{r}_{1,a}^*(-k)\delta_{23}(k)\re^{t\Phi_{23}}\\
		0 & 0 & 1\\
	\end{pmatrix},\\
    & T_6(\zeta,k)=\begin{pmatrix}
	1 & -\hat{r}_{1,a}(k)\delta_{21}^{-1}(k)\re^{t\Phi_{12}}  & 0\\
	0 & 1 & 0\\
	-r_{2,a}(-k)\frac{\delta_{21}(k)}{\delta_{23}(k)}\re^{-t\Phi_{13}} &  \hat{\alpha}_a^*(k)\delta_{23}^{-1}(k)\re^{-t\Phi_{23}}  & 1\\
	\end{pmatrix},\nonumber\\
		&T_7(\zeta,k)=\begin{pmatrix}
		1 & r_{1,a}(k)\delta_{21}^{-1}(k)\re^{t\Phi_{12}} & -r_{2,a}(k)\frac{\delta_{23}(k)}{\delta_{21}(k)}\re^{t\Phi_{13}}\\
		0 & 1 & 0\\
		0 & 0 & 1\\
	\end{pmatrix},\\
    & T_8(\zeta,k)=\begin{pmatrix}
		1 & 0 & r_{2,a}^*(-k)\frac{\delta_{23}(k)}{\delta_{21}(k)}\re^{t\Phi_{13}}\\
		0 & 1 & 2\sigma  r_{1,a}^*(-k)\delta_{23}(k)\re^{t\Phi_{23}}\\
		0 & 0 & 1\\
	\end{pmatrix},\\
	&T_9(\zeta,k)=\begin{pmatrix}
		1 & 0 & 0\\
		0 & 1 & 0\\
		-r_{2,a}(-k)\frac{\delta_{21}(k)}{\delta_{23}(k)}\re^{-t\Phi_{13}} &  r_{1,a}(-k)\delta_{23}^{-1}(k)\re^{-t\Phi_{23}} & 1\\
	\end{pmatrix},\\
    & T_{10}(\zeta,k)=\begin{pmatrix}
		1 & 0 & 0\\
		2\sigma  r_{1,a}^*(k)\delta_{21}(k)\re^{-t\Phi_{12}} & 1 & 0\\
		r_{2,a}^*(k)\frac{\delta_{21}(k)}{\delta_{23}(k)}\re^{-t\Phi_{13}} & 0 & 1\\
	\end{pmatrix}.\nonumber
\end{align*}

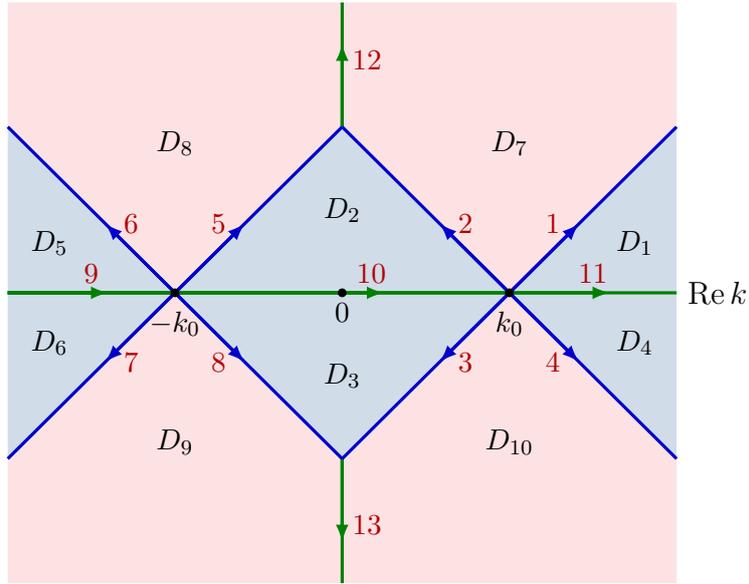
\begin{figure}[htbp]
	\centering
	\begin{tikzpicture}[scale=1.1]
		
		\definecolor{mycolor1}{HTML}{D0DCE8}
		\definecolor{mycolor2}{HTML}{FDE2E4}
		
		\fill[fill=mycolor1] (0,2) -- (2,0) -- (0,-2) -- (-2,0) -- cycle;
		\fill[fill=mycolor1] (2,0) -- (4,-2) -- (4,2) -- cycle;
		\fill[fill=mycolor1] (-2,0) -- (-4,-2) -- (-4,2) -- cycle;
		
		\fill[fill=mycolor2] (0,3.5) -- (0,2) --  (2,0)  -- (4,2) -- (4,3.5) -- cycle;
		\fill[fill=mycolor2] (0,-3.5) -- (0,-2) --  (2,0)  -- (4,-2) -- (4,-3.5) -- cycle;
		\fill[fill=mycolor2] (0,3.5) -- (0,2) --  (-2,0)  -- (-4,2) -- (-4,3.5) -- cycle;
		\fill[fill=mycolor2] (0,-3.5) -- (0,-2) --  (-2,0)  -- (-4,-2) -- (-4,-3.5) -- cycle;

		\draw[very thick, black!50!green] (-4,0) -- (4,0);
		\draw[very thick, black!50!green] (0,2) -- (0,3.5);
		\draw[very thick, black!50!green] (0,-2) -- (0,-3.5);
		
		
		\draw[very thick, black!50!green, -latex] (-4,0) -- (-2.8,0);
		\draw[very thick, black!50!green, -latex] (-4,0) -- (0.5,0);
		\draw[very thick, black!50!green, -latex] (-4,0) -- (3.2,0);
		
		\draw[very thick, black!50!green, -latex] (0,-2) -- (0,-3);
		\draw[very thick, black!50!green, -latex] (0,2) -- (0,3);
		
		\draw[very thick, black!20!blue] (4,2) -- (0,-2);
		\draw[very thick, black!20!blue] (4,-2) -- (0,2);
		
		\draw[very thick, black!20!blue] (-4,2) -- (0,-2);
		\draw[very thick, black!20!blue] (-4,-2) -- (0,2);
		
		\draw[very thick, black!20!blue,-latex] (2,0) --(2+1/1.2,1/1.2);
		\draw[very thick, black!20!blue,-latex] (2,0) -- (-1/1.2+2,-1/1.2);
		\draw[very thick, black!20!blue,-latex] (2,0) -- (-1/1.2+2,1/1.2);
		\draw[very thick, black!20!blue,-latex] (2,0) -- (2+1/1.2,-1/1.2);
		
		\draw[very thick, black!20!blue,-latex] (-2,0) --(-2+1/1.2,1/1.2);
		\draw[very thick, black!20!blue,-latex] (-2,0) -- (-1/1.2-2,-1/1.2);
		\draw[very thick, black!20!blue,-latex] (-2,0) -- (-1/1.2-2,1/1.2);
		\draw[very thick, black!20!blue,-latex] (-2,0) -- (-2+1/1.2,-1/1.2);
		
		\fill (-2,0) circle (1.5pt);
		\fill (0,0) circle (1.5pt);
		\fill (2,0) circle (1.5pt);
		
		{\small 	\node[below] at (-2,-0.1) {$-k_0$};
			\node[below] at (0,0) {$\text{0}$};
			\node[below] at (2,-0.1) {$k_0$};
			
			\node[red!70!black,left] at (1.9+1/1.2,1/1.2) {$1$};
			\node[red!70!black,right] at (2.1-1/1.2,1/1.2) {$2$};
			\node[red!70!black,right] at (2.1-1/1.2,-1/1.2) {$3$};
			\node[red!70!black,left] at (1.9+1/1.2,-1/1.2) {$4$};
			
			\node[red!70!black,left] at (-2.1+1/1.2,1/1.2) {$5$};
			\node[red!70!black,right] at (-1.9-1/1.2,1/1.2) {$6$};
			\node[red!70!black,right] at (-1.9-1/1.2,-1/1.2) {$7$};
			\node[red!70!black,left] at (-2.1+1/1.2,-1/1.2) {$8$};
			
			\node[red!70!black,above] at (-3,0) {$9$};
			\node[red!70!black,above] at (0.35,0) {$10$};
			\node[red!70!black,above] at (3,0) {$11$};
			
			\node[red!70!black,right] at (0,2.8) {$12$};
			\node[red!70!black,right] at (0,-2.8) {$13$};
			}
		
		{\small 	\node at (3.5,0.6) {$D_1$};
			\node at (0,1) {$D_2$};
			\node at (3.5,-0.6) {$D_4$};
			\node at (0,-1) {$D_3$};
			
			\node at (-3.5,0.6) {$D_5$};
			\node at (-3.5,-0.6) {$D_6$};
			
			\node at (2,1.8) {$D_7$};
			\node at (-2,1.8) {$D_{8}$};
			\node at (-2,-1.8) {$D_{9}$};
			\node at (2,-1.8) {$D_{10}$};
		}

		\node[right] at (4,0) {$\rre k$};
		
	\end{tikzpicture}
	\caption{The jump contour $\Sigma^{(2)}$ and regions  $D_j$, $j=1,2,\cdots,10,$ in the complex $k$-plane.}
	\label{fig_Gamma2}
\end{figure}

        \begin{lem}\label{lem_T_bound}
        	$T(\zeta,k)$ is uniformly bounded for $\zeta\in\mathcal{I}_2$, $k\in \mathbb{C}\setminus\Sigma^{(2)}$, and $t>0$. Moreover,
        	\begin{equation*}\label{T_k_infty}
        		T(k)=I+\mathcal{O}(k^{-1}),\quad k\to\infty.
        	\end{equation*}
        \end{lem}
        \begin{proof}
        	We take the case in region $D_1$ as an example for the proof; the cases in the remaining regions can be proved similarly. From Figure \ref{fig_theta_sign1}, we can obtain:
        	\begin{equation*}
        		t\, \rre \Phi_{12}(\zeta,k)<0,\quad t\,\rre \Phi_{13}(\zeta,k)<0,\quad -t\,\rre \Phi_{23}(\zeta,k)<0,\quad \zeta\in \mathcal{I}_2,\,\, k\in D_1,
        	\end{equation*}
        	and according to Lemmas \ref{lem_r_decomposition}, \ref{lem_delta}, the following holds
        	\begin{align*}
        		\left| T_1(\zeta,k)-I\right| =&\left| \hat{\alpha}^*_{a}(-k)\delta_{21}^{-1}(k)\re^{t\Phi_{12}}\right| +\left| \hat{r}_{1,a}(-k)\delta_{23}^{-1}(k)\re^{-t\Phi_{23} }\right|  +\left|  r_{2,a}(k)\frac{\delta_{23}(k)}{\delta_{21}(k)}\re^{t\Phi_{13}}\right|\\
        		\leq & \frac{C}{1+|k|}\left(\re^{-\frac{3}{4}t|\rre\Phi_{12}|} +\re^{-\frac{3}{4}t|\rre\Phi_{23}|}+\re^{-\frac{3}{4}t|\rre\Phi_{13}|}\right) .\qedhere
        	\end{align*}
        
        \end{proof}

   Construct the matrix transformation 
   \begin{equation}\label{trans_T}
   	M^{(2)}(x,t,k)=M^{(1)}(x,t,k)T(x,t,k),
   \end{equation}
    to obtain the eigenfunction $M^{(2)}(x,t,k)$ that satisfies the new jump property:
    $$M_+^{(2)}(x,t,k)=M_-^{(2)}(x,t,k)V^{(2)}(x,t,k),\quad k\in \Sigma^{(2)},$$ 
    where $\Sigma^{(2)}$ as given in Figure \ref{fig_Gamma2}, and $V^{(2)}_j(x,t,k)$ for $k\in \Sigma^{(2)}_j$, $j=1,2,\cdots,8$, are as follows
     \begin{align}
         &V^{(2)}_1(x,t,k)=\begin{pmatrix}
            1 & 0 & 0\\
            0 & 1 & 0\\
            0 & \hat{r}_{1,a}(-k)\delta_{23}^{-1}(k)\re^{-t\Phi_{23}} & 1\\
        \end{pmatrix},\, &&V^{(2)}_2(x,t,k)=\begin{pmatrix}
            1 & 0 & 0\\
            0 & 1 & 2\sigma  r_{1,a}^*(-k)\delta_{23}(k)\re^{t\Phi_{23}}\\
            0 & 0 & 1\\
        \end{pmatrix},\nonumber\\
        &V^{(2)}_3(x,t,k)=\begin{pmatrix}
            1 & 0 & 0\\
            0 & 1 & 0\\
            0 & -r_{1,a}(-k)\delta_{23}^{-1}(k)\re^{-t\Phi_{23}}  & 1\\
        \end{pmatrix},\, &&V^{(2)}_4(x,t,k)=\begin{pmatrix}
            1 & 0 & 0\\
            0 & 1 & -2\sigma  \hat{r}_{1,a}^*(-k)\delta_{23}(k)\re^{t\Phi_{23}}\\
            0 & 0 & 1\\
        \end{pmatrix},\label{v_2_1-8}\\
        &V^{(2)}_5(x,t,k)=\begin{pmatrix}
            1 & -r_{1,a}(k)\delta_{12}^{-1}(k)\re^{t\Phi_{12}} & 0\\
            0 & 1 & 0\\
            0 & 0 & 1\\
        \end{pmatrix},\,
        &&V^{(2)}_6(x,t,k)=\begin{pmatrix}
            1 & 0 & 0\\
            -2\sigma  \hat{r}_{1,a}^*(k)\delta_{12}(k)\re^{-t\Phi_{12}} & 1 & 0\\
            0 & 0 & 1\\
        \end{pmatrix},\nonumber\\
        &V^{(2)}_7(x,t,k)=\begin{pmatrix}
            1 & \hat{r}_{1,a}(k)\delta_{12}^{-1}(k)\re^{t\Phi_{12}} & 0\\
            0 & 1 & 0\\
            0 & 0  & 1\\
        \end{pmatrix},\, &&V^{(2)}_8(x,t,k)=\begin{pmatrix}
            1 & 0 & 0\\
            2\sigma r_{1,a}^*(k)\delta_{12}(k)\re^{-t\Phi_{12}} & 1 & 0\\
            0 & 0 & 1\\
        \end{pmatrix}.\nonumber
     \end{align}
     
     In fact, the jump conditions generated by the above transformation \eqref{trans_T} exist not only for $k\in \Sigma^{(2)}_j$, $j=1,2,\cdots,8$, but the eigenfunction $M^{(2)}(x,t,k)$ is also non-analytic as $k\in \Sigma^{(2)}_j$, $j=9,10,\cdots,13$ (corresponding to the green jump line portions in Figure \ref{fig_Gamma2}). In the two remarks below, we will explain that the jump matrices corresponding to $k\in \Sigma^{(2)}_j$, $j=9,10,\cdots,13$, will decay to the identity matrix faster than those corresponding to $k\in \Sigma^{(2)}_j$, $j=1,2,\cdots,8$, as $t\to\infty$, which is the reason why we use green lines in Figure \ref{fig_Gamma2}.

     \begin{remark}\label{remark_v2_9-11}
     	In the case of $k\in\Sigma^{(2)}_j$, $j=9,10,11$, we have
     	\begin{align*}
     		V^{(2)}_9(k)=T_6^{-1}(k)V^{(1)}_1(k)T_5(k),\quad
     		V^{(2)}_{10}(k)=T_3^{-1}(k)V^{(1)}_2(k)T_2(k),\quad
     		V^{(2)}_{11}(k)=T_4^{-1}(k)V^{(1)}_3(k)T_1(k).
     	\end{align*}
     	Through calculation, it can be shown that the non-zero elements of the jump matrices $V^{(2)}_j-I$, $j=9,10,11$, can be completely controlled by $r_{1,r}(k)$, $\hat r_{1,r}(k)$, $\hat \alpha_{r}(k)$, and $r_{2,r}(k)$; the explicit forms are too complicated to be presented here in detail.
       \end{remark}

     \begin{remark}\label{remark_v2_12-13}
         The explicit forms of the jump matrices $V^{(2)}_{12}(x,t,k)$ and $V^{(2)}_{13}(x,t,k)$ are given by
          \begin{align*}\label{v_1314}
         &V^{(2)}_{12}(x,t,k)=T^{-1}_7(x,t,k)T_8(x,t,k)=\begin{pmatrix}
            1 & -r_{1,a}(k)\delta_{21}^{-1}(k)\re^{t\Phi_{12}} & 0\\
            0 & 1 & 2\sigma  r_{1,a}^*(-k)\delta_{23}(k)\re^{t\Phi_{23}}\\
            0 & 0 & 1\\
        \end{pmatrix},\\
        &V^{(2)}_{13}(x,t,k)=T^{-1}_9(x,t,k)T_{10}(x,t,k)=\begin{pmatrix}
            1 & 0 & 0\\
            2\sigma r_{1,a}^*(k)\delta_{23}^{-1}(k)\re^{-t\Phi_{12}} & 1 & 0\\
            0 & -r_{1,a}(-k)\delta_{23}^{-1}(k)\re^{-t\Phi_{23}} & 1\\
        \end{pmatrix}.
     \end{align*}
         It is evident that $\rre \Phi_{12}(k)<0$, $\rre \Phi_{23}(k)<0$ for $k\in \Sigma^{(2)}_{12}$ and $\rre \Phi_{12}(k)>0$, $\rre \Phi_{23}(k)>0$ for $k\in \Gamma^{(2)}_{13}$, which means that the two jump matrices decay exponentially to the identity matrix.
        
     \end{remark}

    \begin{lem}\label{lem_v2_error}
   The jump matrix $V^{(2)}(x,t,k)$ converges to the identity matrix $I$ uniformly for $\zeta\in\mathcal{I}_2$ and $k\in \Sigma^{(2)}$ except near the two critical points $\pm k_0$ as $t\to\infty$. Additionally, the following estimates are established:
       \begin{align*}
            &\left\|(1+|\cdot|)\left(V^{(2)}_j(x,t,\cdot)-I\right)\right\|_{(L^1\cap L^\infty)\left(\Sigma^{(2)}_j\right)}\leq Ct^{-N}, \quad  &&j=9,10,11,\\
            &\left\|(1+|\cdot|)\left(V^{(2)}_j(x,t,\cdot)-I\right)\right\|_{(L^1\cap L^\infty)\left(\Sigma^{(2)}_j\right)}\leq C\re^{-ct}, \quad &&j=12,13.
        \end{align*}
    \end{lem}
    \begin{proof}
    	The proof is divided into three parts. First, when $k \in \Sigma^{(2)}_1$, there exists $\rre\Phi_{23}(k)\geq0$ such that $V^{(2)}_1(x,t,k)$ converges to $I$ as $t \to \infty$. However, when $k$ is near $k_0$, its $(2,1)$-element does not converge uniformly to $0$, because $\rre\Phi_{23}(k_0)=0$. The cases where $k \in \Sigma^{(2)}_j$, $j=2,3,\cdots,8$, are similar and can be proved analogously.
    	
    	According to Figure $\ref{fig_theta_sign1}$, we have $\rre\Phi_{12}(k)=\rre\Phi_{13}(k)=\rre\Phi_{23}(k)=0$ for $k\in \mathbb{R}$. Moreover, based on Lemma \ref{lem_r_decomposition} and Remark \ref{remark_v2_9-11}, it follows that $V^{(2)}_j(x,t,k)$ converge uniformly to the identity matrix for $k\in \Sigma^{(2)}_j$, $j=9,10,11$, except at the points $k_0$ and $-k_0$.
    	
    	Finally, we have $\rre\Phi_{12}(k)<0$, $\rre\Phi_{23}(k)<0$ for $k\in \Sigma^{(2)}_{12}$, and according to Lemmas \ref{lem_r_decomposition} and \ref{lem_delta},
   \begin{align*}
       &\left|\left(V^{(2)}_{12}-I\right)_{12}\right|=\left|-r_{1,a}(k)\delta_{21}^{-1}(k)\re^{t\Phi_{12}(\zeta,k)}\right|\leq \frac{C}{1+|k|}\re^{-\frac{3}{4}t|\rre\Phi_{12}(\zeta,k)|}\leq C \re^{-ct},\\
       &\left|\left(V^{(2)}_{12}-I\right)_{23}\right|=\left|2\sigma  r_{1,a}^*(-k)\delta_{23}(k)\re^{t\Phi_{23}(\zeta,k)}\right|\leq \frac{C}{1+|k|}\re^{-\frac{3}{4}t|\rre\Phi_{23}(\zeta,k)|}\leq C \re^{-ct}.
   \end{align*}
   Thus the proof of $j=12$ is completed, and the case for $\Sigma^{(2)}_{13}$ can be proved similarly.
    \end{proof}
    
    According to the conclusion of Lemma \ref{lem_v2_error}, we next only need to consider the behavior near the critical points $\pm k_0$. Since $\rre\Phi_{12}(-k_0)=\rre\Phi_{23}(k_0)=0$, it is necessary to introduce gauge transformations to obtain new variables, transforming the jumps near the points $\pm k_0$ into a model that can be solved using parabolic cylinder functions. The required scaling transformation is constructed as follows:
    \begin{equation}
    	z_{(\pm)}=2\sqrt{t}(k\mp k_0),\quad  k\in B_\epsilon(\pm k_0),\label{zpm}
    \end{equation}
    where $\epsilon=|k_0|/2$, $B_\epsilon(\pm k_0)=\{ k|\left|k\mp k_0\right|<\epsilon \}$, and we denot $B_\epsilon=B_\epsilon(k_0)\cup B_\epsilon(-k_0)$.
    Based on the aforementioned transformation, the functions can be written as
    \begin{equation*}
    t\Phi_{12}(\zeta,k)=t\Phi_{12}(\zeta,-k_0)+\frac{\ri z_{(-)}^2}{2}, \quad t\Phi_{23}(\zeta,k)=t\Phi_{23}(\zeta,k_0)-\frac{\ri z_{(+)}^2}{2}.
    \end{equation*}
    
    From equation \eqref{v_2_1-8}, it is evident that the off-diagonal entries of the jump matrices $V^{(2)}_j$, $j=1,2,\cdots,8$, always involve the functions $\delta(\pm k)$. Under the transformed spectral variables $z_{(\pm)}$, we also need to carry out a new decomposition of these functions. For $k\in B_\epsilon(k_0)\setminus\left[k_0,k_0+\epsilon \right) $,
    \begin{align*}
        \delta_{23}(k)={\delta(k)}{\delta(-k)^2}=\delta(k)(2\sqrt{t})^{-2\ri \nu}\re ^{2\ri\nu\ln_0(z_{(+)})}\re^{-2\chi(\zeta,-k)}:=\re ^{2\ri\nu\ln_0(z_{(+)})}d^{(k_0)}_{0}(\zeta)d^{(k_0)}_{1}(\zeta,k),
    \end{align*}
    where 
    \begin{align*}
        &d^{(k_0)}_{0}(\zeta)=(2\sqrt{t})^{-2\ri \nu}\re^{-2\chi(\zeta,-k_0)}\delta(k_0),\quad d^{(k_0)}_{1}(\zeta,k)=\re^{-2\chi(\zeta,-k)+2\chi(\zeta,-k_0)}\frac{\delta(k)}{\delta(k_0)}.
    \end{align*}
    Thus, the following eigenfunction can be constructed near $k_0$:
    \begin{equation*}\label{trans_M_k1}
        \tilde{M}^{(k_0)}(\zeta,k)=M^{(2)}(\zeta,k)Y_{(+)}(\zeta), \quad k\in B_\epsilon(k_0),
        \end{equation*}
        where
        \begin{equation*}
        Y_{(+)}(\zeta)=\begin{pmatrix}
        	1 & 0 & 0\\
        	0 & \left(d^{(k_0)}_{0}(\zeta)\right)^{1/2}\re^{\frac{t}{2}\Phi_{23}(\zeta,k_0)} & 0\\
        	0 & 0 & \left(d^{(k_0)}_{0}(\zeta)\right)^{-1/2}\re^{-\frac{t}{2}\Phi_{23}(\zeta,k_0)}
        \end{pmatrix}.
    \end{equation*}
    Then the jump matrices $\tilde{V}^{(k_0)}_j(\zeta,z_{(+)}(k))$ for $k\in X_j^{k_0}$, $j=1,2,3,4,$ of $\tilde{M}^{(k_0)}(\zeta,k)$ are
    \begin{align*}
         &\tilde{V}^{(k_0)}_1(\zeta,z_{(+)})=\begin{pmatrix}
            1 & 0 & 0\\
            0 & 1 & 0\\
            0 & \hat{r}_{1,a}(-k)\left(d^{(k_0)}_{1}\right)^{-1}\re ^{-2\ri\nu\ln_0(z_{(+)})}\re^{\frac{\ri z_{(+)}^2}{2}} & 1\\
        \end{pmatrix},\\
        & \tilde{V}^{(k_0)}_2(\zeta,z_{(+)})=\begin{pmatrix}
            1 & 0 & 0\\
            0 & 1 & 2\sigma r_{1,a}^*(-k)d^{(k_0)}_{1}\re ^{2\ri\nu\ln_0(z_{(+)})}\re^{-\frac{\ri z_{(+)}^2}{2}}\\
            0 & 0 & 1\\
        \end{pmatrix},\\
        &\tilde{V}^{(k_0)}_3(\zeta,z_{(+)})=\begin{pmatrix}
            1 & 0 & 0\\
            0 & 1 & 0\\
            0 & -r_{1,a}(-k)\left(d^{(k_0)}_{1}\right)^{-1}\re ^{-2\ri\nu\ln_0(z_{(+)})}\re^{\frac{\ri z_{(+)}^2}{2}}  & 1\\
        \end{pmatrix},\\
        &\tilde{V}^{(k_0)}_4(\zeta,z_{(+)})=\begin{pmatrix}
            1 & 0 & 0\\
            0 & 1 & -2\sigma \hat{r}^*_{1,a}(-k)d^{(k_0)}_{1}\re ^{2\ri\nu\ln_0(z_{(+)})} \re^{-\frac{\ri z_{(+)}^2}{2}}\\
            0 & 0 & 1\\
        \end{pmatrix},
    \end{align*}
    where $X^{k_0}_j=(k_0+X_j)\cap B_\epsilon(k_0)$, $X_j=\left\{z\in \mathbb{C}\,|\,z=s\re^{\frac{(2j-1)\pi\ri}{4}}, 0\leq s\leq\infty\right\}$, $j=1,2,3,4$, and $X^{k_0}=\bigcup_{j=1}^4X_j^{k_0}$,  $X=\bigcup_{j=1}^4X_j$.

    For the case where $k\in B_\epsilon(-k_0)\setminus\left(-k_0-\epsilon,-k_0\right]$, we define
    \begin{align*}
        \delta_{21}^{-1}(k)=\frac{1}{\delta(k)^2\delta(-k)}&=\frac{1}{\delta(-k)}(2\sqrt{t})^{2\ri \nu}\re ^{-2\ri\nu\ln_\pi(z_{(-)})}\re^{2\chi(\zeta,k)}:=\re ^{-2\ri\nu\ln_\pi(z_{(-)})}d_{0}^{(-k_0)}(\zeta)d_{1}^{(-k_0)}(\zeta,k),
    \end{align*}
    where
    \begin{align*}
        &d_{0}^{(-k_0)}(\zeta)=(2\sqrt{t})^{2\ri \nu}\re^{2\chi(\zeta,-k_0)}\frac{1}{\delta(k_0)},\quad d_{1}^{(-k_0)}(\zeta,k)=\re^{2\chi(\zeta,k)-2\chi(\zeta,-k_0)}\frac{\delta(k_0)}{\delta(-k)}.
    \end{align*}
    Define $\tilde{M}^{(-k_0)}$ for $k$ near $-k_0$ by
    \begin{equation*}\label{trans_M_-k1}
        \tilde{M}^{(-k_0)}(\zeta,k)=M^{(2)}(\zeta,k)Y_{(-)}(\zeta), \quad k\in B_\epsilon(-k_0),\end{equation*}
        where
        \begin{equation*}
       Y_{(-)}(\zeta)=\begin{pmatrix}
        	\left(d_{0}^{(-k_0)}(\zeta)\right)^{1/2}\re^{\frac{t}{2}\Phi_{12}(\zeta,-k_0)}  & 0 & 0\\
        	0 & \left(d_{0}^{(-k_0)}(\zeta)\right)^{-1/2}\re^{-\frac{t}{2}\Phi_{12}(\zeta,-k_0)} & 0\\
        	0 & 0 & 1
        \end{pmatrix},
    \end{equation*}
    And the jump matrices $\tilde{V}^{(-k_0)}_j(\zeta,z_{(-)}(k))$ for $k\in X_{j+4}^{-k_0}=(-k_0+X_j)\cap B_\epsilon(-k_0)$, $j=1,2,3,4,$ are
    \begin{align*}
         &\tilde{V}^{(-k_0)}_5(\zeta,z_{(-)})=\begin{pmatrix}
            1 & -r_{1,a}(k)d_{1}^{(-k_0)}\re ^{-2\ri\nu\ln_\pi(z_{(-)})}  \re^{\frac{\ri z_{(-)}^2}{2}} & 0\\
            0 & 1 & 0\\
            0 & 0 & 1\\
        \end{pmatrix},\\
        &\tilde{V}^{(-k_0)}_6(\zeta,z_{(-)})=\begin{pmatrix}
            1 & 0 & 0\\
            -2\sigma \hat{r}^*_{1,a}(k)\left(d_{1}^{(-k_0)}\right)^{-1} \re ^{2\ri\nu\ln_\pi(z_{(-)})} \re^{-\frac{\ri z_{(-)}^2}{2}} & 1 & 0\\
            0 & 0 & 1\\
        \end{pmatrix},\\
        &\tilde{V}^{(-k_0)}_7(\zeta,z_{(-)})=\begin{pmatrix}
            1 & \hat{r}_{1,a}(k)d_{1}^{(-k_0)}\re ^{-2\ri\nu\ln_\pi(z_{(-)})}  \re^{\frac{\ri z_{(-)}^2}{2}} & 0\\
            0 & 1 & 0\\
            0 & 0 & 1\\
        \end{pmatrix},\\
        &\tilde{V}^{(-k_0)}_8(\zeta,z_{(-)})=\begin{pmatrix}
            1 & 0 & 0\\
            2\sigma r^*_{1,a}(k)\left(d_{1}^{(-k_0)}\right)^{-1} \re ^{2\ri\nu\ln_\pi(z_{(-)})}  \re^{-\frac{\ri z_{(-)}^2}{2}} & 1 & 0\\
            0 & 0 & 1\\
        \end{pmatrix}.
    \end{align*}

        \begin{lem}\label{lem_Y12_bound}
    	For $\zeta\in\mathcal{I}_2$ and $t\geq2$, the functions $Y_{(\pm)}(\zeta)$ are uniformly bounded, i.e.,
    	\begin{equation*}
    		\sup_{\zeta\in\mathcal{I}_2}\sup_{t\geq2}\left|Y_{(\pm)}(\zeta)\right|<C,\quad 	\sup_{\zeta\in\mathcal{I}_2}\sup_{t\geq2}\left|Y_{(\pm)}^{-1}(\zeta)\right|<C.
    	\end{equation*}
    	The functions $d_0^{(\pm k_0)}(\zeta)$ and $d_1^{(\pm k_0)}(\zeta,k)$ satisfy:
    	\begin{equation}\label{Error_d1-1}
    		\begin{aligned}
    		&\left|d_0^{(k_0)}(\zeta)\right|=\re^{2\pi\nu},  \quad &&\left|d_1^{(k_0)}(\zeta,k)-1\right|\leq C|k-k_0|(1+|\ln|k-k_0||),\\
    		&\left|d_0^{(-k_0)}(\zeta)\right|=1, \quad &&\left|d_1^{(-k_0)}(\zeta,k)-1\right|\leq C|k+k_0|(1+|\ln|k+k_0||).
    		\end{aligned}
    	\end{equation}
    	
    \end{lem}

    \begin{proof}
    	Direct calculation shows that the following two equations hold:
    	\begin{equation*}
    		\begin{aligned}
    			&\rre\chi(\zeta,-k)|_{k=k_0}=-\frac{1}{2\pi}\int_{k_0}^{\infty}\pi \rd\ln (1-2|r_1(-s)|^2)=\frac{1}{2} \ln (1-2|r_1(-k_0)|^2)=-\pi \nu,\\
    			&\rre\chi(\zeta,k)|_{k=-k_0}=-\frac{1}{2\pi}\int_{k_0}^{\infty}0 \rd\ln (1-2|r_1(-s)|^2)=0,
    		\end{aligned}
    	\end{equation*}
    	and consequently $\left|d_0^{(k_0)}(\zeta)\right|=\re^{2\pi\nu}$ and $\left|d_0^{(-k_0)}(\zeta)\right|=1$. The two expressions on the right-hand side of equation \eqref{Error_d1-1} can also be directly obtained from the conclusions in Lemma \ref{lem_delta} that the function $\delta(\zeta,k)$ satisfies. 
    \end{proof}

   Denote $\mathrm{r}=r_1(-k_0)$ and $\hat{\mathrm{r}}=\mathrm{r}/(1-2|\mathrm{r}|^2)$. For any fixed $z_{(\pm)}$, $r_{1,a}(\mp k)\to \mathrm{r}$, $\hat{r}_{1,a}(\mp k)\to\hat{\mathrm{r}}$ and $d_1^{(\pm k_0)}\to 1$ as $t\to\infty$. Thus, the jump matrices $\tilde{V}^{(k_0)}_j(\zeta,z_{(+)}(k))$ and $\tilde{V}^{(-k_0)}_{j+4}(\zeta,z_{(-)}(k))$ tend to the matrices $V^{X}_{ k_0}(\mathrm{r},z_{(+)})$ and $V^{X}_{- k_0}(\mathrm{r},z_{(-)})$, respectively, where
   \begin{equation}\label{model_V_k0}
   \begin{aligned}
   	&V^{X}_{ k_0}(\mathrm{r},z_{(+)})=\begin{pmatrix}
   		1 & 0 & 0\\
   		0 & 1 & 0\\
   		0 & \hat{\mathrm{r}}z_{(+)} ^{-2\ri\nu(\mathrm{r})}\re^{\frac{\ri z_{(+)}^2}{2}} & 1\\
   	\end{pmatrix},\quad &&z_{(+)}\in X_1,\\
    &V^{X}_{ k_0}(\mathrm{r},z_{(+)})=\begin{pmatrix}
   		1 & 0 & 0\\
   		0 & 1 & 2\sigma {\mathrm{r}^*}z_{(+)}^{2\ri\nu(\mathrm{r})}\re^{-\frac{\ri z_{(+)}^2}{2}}\\
   		0 & 0 & 1\\
   	\end{pmatrix},\quad  &&z_{(+)}\in X_2,\\
   	&V^{X}_{ k_0}(\mathrm{r},z_{(+)})=\begin{pmatrix}
   		1 & 0 & 0\\
   		0 & 1 & 0\\
   		0 & -\mathrm{r}z_{(+)}^{-2\ri\nu(\mathrm{r})}\re^{\frac{\ri z_{(+)}^2}{2}}  & 1\\
   	\end{pmatrix},\quad  &&z_{(+)}\in X_3,\\
    &V^{X}_{ k_0}(\mathrm{r},z_{(+)})=\begin{pmatrix}
   		1 & 0 & 0\\
   		0 & 1 & -2\sigma \hat{\mathrm{r}}^*z_{(+)}^{2\ri\nu(\mathrm{r})}\re^{-\frac{\ri z_{(+)}^2}{2}}\\
   		0 & 0 & 1\\
   	\end{pmatrix},\quad  &&z_{(+)}\in X_4,
      \end{aligned}
  \end{equation}
  \begin{equation}\label{model_V_-k0}
  	\begin{aligned}
  		&V^{X}_{-k_0}(\mathrm{r},z_{(-)})=\begin{pmatrix}
  			1 & -\mathrm{r}z_{(-)}^{-2\ri\nu(\mathrm{r})}\re^{\frac{\ri z_{(-)}^2}{2}} & 0\\
  			0 & 1 & 0\\
  			0 & 0 & 1\\
  		\end{pmatrix},\quad &&z_{(-)}\in X_1,\\
    &V^{X}_{- k_0}(\mathrm{r},z_{(-)})=\begin{pmatrix}
  		1 & 0 & 0\\
  		-2\sigma \hat{\mathrm{r}}^*z_{(-)}^{2\ri\nu(\mathrm{r})}\re^{-\frac{\ri z_{(-)}^2}{2}} & 1 & 0\\
  		0 & 0 & 1\\
  		\end{pmatrix},\quad  &&z_{(-)}\in X_2,\\
  		&V^{X}_{-k_0}(\mathrm{r},z_{(-)})=\begin{pmatrix}
  			1 & \hat{\mathrm{r}}z_{(-)}^{-2\ri \tilde\nu(\mathrm{r})}\re^{\frac{\ri z_{(-)}^2}{2}} & 0\\
  			0 & 1 & 0\\
  			0 & 0 & 1\\
  		\end{pmatrix},\quad  &&z_{(-)}\in X_3,\\
    &V^{X}_{-k_0}(\mathrm{r},z_{(-)})=\begin{pmatrix}
  			1 & 0 & 0\\
  			2\sigma {\mathrm{r}^*}z_{(-)}^{2\ri\nu(\mathrm{r})}\re^{-\frac{\ri z_{(-)}^2}{2}} & 1 & 0\\
  			0 & 0 & 1\\
  		\end{pmatrix},\quad  &&z_{(-)}\in X_4,
  	\end{aligned}
  \end{equation}
  and the branch cuts running along the positive and negative real axes, respectively, i.e., $z_{(+)}^{2\ri\nu}=\re^{2\ri\nu\ln_0(z_{(+)})}$ and $z_{(-)}^{2\ri\nu}=\re^{2\ri\nu\ln_\pi(z_{(-)})}$.

   Through the following transformation, we can define functions $M^{(\pm k_0)}$ in $B_\epsilon(\pm k_0)$, respectively, to approximate the function  $M^{(2)}(x,t,k)$:
  \begin{align}\label{Mk_0}
      M^{(\pm k_0)}(\zeta,k)=Y_{(\pm)}(\zeta)M^{X}_{\pm k_0}(q(\zeta),z_{(\pm)}(k))Y^{-1}_{(\pm)}(\zeta), \quad k\in B_\epsilon(\pm k_0),
  \end{align}
  where $M^{X}_{\pm k_0}(q(\zeta),z_{(\pm)}(k))$ satisfy the following model problem.
  \begin{rhp}\label{model-rhp1}
  	Find a $3\times3$ matrix-valued function $M^{X}_{\pm k_0}(q,z_{(\pm)})$ with the following properties:
  	\begin{itemize}
  		\item The function $M^{X}_{\pm k_0}(q,z_{(\pm)})$ is holomorphic for $k\in\mathbb{C}\setminus X$.
  		\item The function $M^{X}_{\pm k_0}(q,z_{(\pm)})$ is analytic for $k\in\mathbb{C}\setminus X$, and satisfy the following relationship:
  		\begin{equation*}
  			\left(M^{X}_{\pm k_0}(q,z_{(\pm)}) \right)_+ =\left(M^{X}_{\pm k_0}(q,z_{(\pm)}) \right)_- V^{X}_{\pm k_0}(q,z_{(\pm)}), \quad k\in X,
  		\end{equation*}
  		where $V^{X}_{ k_0}$ and $V^{X}_{- k_0}$ are given by equations \eqref{model_V_k0} and \eqref{model_V_-k0}, respectively.
  		
  		\item  As $z_{(\pm)}\rightarrow\infty$, $M^{X}_{\pm k_0}(q,z_{(\pm)})=I+\mathcal{O}\left( z_{(\pm)}^{-1}\right) $.
  		\item  As $z_{(\pm)}\rightarrow0$, $M^{X}_{\pm k_0}(q,z_{(\pm)})=\mathcal{O}(1)$.
  	\end{itemize}
  \end{rhp}

   The solution $M^{X}_{\pm k_0}(q,z_{(\pm)})$ of the RH problem \ref{model-rhp1} admits the following expansion:
  \begin{equation}\label{Mpm_expand}
  	M^{X}_{\pm k_0}(q,z_{(\pm)})=I+\frac{\left( M^{X}_{\pm k_0}(q)\right) _1}{z_{(\pm)}}+\mathcal{O}\left(\frac{1}{z_{(\pm)}^2}\right), \quad z_{(\pm)}\to\infty,
  \end{equation}
  where
  \begin{equation*}
  	\left( M^{X}_{k_0}(q)\right) _1=\begin{pmatrix}
  		0 & 0 & 0\\
  		0 & 0 & \alpha_{23}\\
  		0 & \alpha_{32} & 0\\
  	\end{pmatrix},\quad 	\left( M^{X}_{-k_0}(q)\right) _1=\begin{pmatrix}
  	0 & \alpha_{12} & 0\\
  	\alpha_{21} & 0 & 0\\
  	0 & 0 & 0\\
  	\end{pmatrix},
  \end{equation*}
  and $\Gamma(\cdot)$ represents the Gamma function,
  \begin{equation*}
  	 \alpha_{12}=\frac{\sqrt{2\pi}\re^{-\frac{\pi\ri}{4}}\re^{-\frac{\pi\nu}{2}}}{2\sigma \mathrm{r}^*\Gamma(\ri\nu)},\quad \alpha_{21}=\frac{\sqrt{2\pi}\re^{\frac{\pi\ri}{4}}\re^{-\frac{\pi\nu}{2}}}{\mathrm{r}\Gamma(-\ri\nu)}, \quad \alpha_{23}=\frac{\sqrt{2\pi}\re^{-\frac{\pi\ri}{4}}\re^{-\frac{5\pi\nu}{2}}}{\mathrm{r}\Gamma(-\ri\nu)},\quad
  	\alpha_{32}=\frac{\sqrt{2\pi}\re^{\frac{\pi\ri}{4}}\re^{\frac{3\pi\nu}{2}}}{2\sigma \mathrm{r}^*\Gamma(\ri\nu)}.
  \end{equation*}

 \begin{lem}\label{lem_Error_v2-vk1}
	For each $\zeta\in\mathcal{I}_2$ and $t\geq 2$, the eigenfunctions $M^{( k_0)}$ and $M^{(-k_0)}$ defined in equation \eqref{Mk_0} are analytic and bounded for $k\in B_\epsilon( k_0)\setminus X^{ k_0}$ and $k\in B_\epsilon(-k_0)\setminus X^{-k_0}$, respectively. The jump relations $M^{( k_0)}_+(x,t,k)=M^{(k_0)}_-(x,t,k)V^{(k_0)}(x,t,k)$ for $k\in X^{k_0}$ and $M^{(- k_0)}_+(x,t,k)=M^{(-k_0)}_-(x,t,k)V^{(-k_0)}(x,t,k)$ for $k\in X^{-k_0}=\cup_{j=5}^8X_j^{-k_0}$ hold. Moreover, the jump matrices $V^{(\pm k_0)}$ satisfy the following estimates:
	\begin{align}
		&\left\|V^{(2)}(x,t,\cdot)-V^{(\pm k_0)}(x,t,\cdot)\right\|_{L^1(X^{\pm k_0})}\leq \frac{C \ln t}{t},\label{Error_v2-vk1}\\
		&\left\|V^{(2)}(x,t,\cdot)-V^{(\pm k_0)}(x,t,\cdot)\right\|_{L^\infty(X^{\pm k_0})}\leq \frac{C \ln t}{\sqrt{t}}.\label{Error_v2-v-k1}
	\end{align}
	In addition,
	\begin{equation}\label{Mk1-I}
		\left\|\left(M^{(\pm k_0)}(x,t,\cdot)\right)^{-1}-I\right\|_{L^\infty(\partial B_\epsilon(\pm k_0))}=\mathcal{O}\left(t^{-1/2}\right),
	\end{equation}
	\begin{equation}\label{intM-I}\frac{1}{2\pi\ri} \int_{\partial B_\epsilon(\pm k_0)} \left(\left(M^{(\pm k_0)}(x,t,k)\right)^{-1}-I\right)\rd k=-\frac{Y_{(\pm)}(\zeta)\left( M^{X}_{\pm k_0}(q)\right) _1Y_{(\pm)}^{-1}(\zeta)}{2\sqrt{t}}+\mathcal{O}\left(t^{-1}\right),
	\end{equation}
	uniformly for $\zeta\in\mathcal{I}_2$.
\end{lem}
\begin{proof}
	According to the definitions of $M^{(\pm k_0)}$ in equation \eqref{Mk_0}, we have
	\begin{align*}
		V^{(2)}(k)-V^{(\pm k_0)}(k)=Y_{(\pm)}(\zeta)\left(\tilde{V}^{(\pm k_0)}(k)-V^{X}_{\pm k_0}(z_{(\pm)}(k))\right)Y^{-1}_{(\pm)}(\zeta),\quad k\in X^{\pm k_0},
	\end{align*}
	Further, by utilizing the boundedness of the functions $Y_{(\pm)}(\zeta)$ and $Y_{(\pm)}^{-1}(\zeta)$ in Lemma \ref{lem_Y12_bound}, we prove that equations \eqref{Error_v2-vk1} and \eqref{Error_v2-v-k1} can be equivalently considered as $\tilde{V}^{(\pm k_0)}(k)-V^{X}_{\pm k_0}(z_{(\pm)}(k))$. 
	
	For $k\in X^{k_0}_1$, we can obtain
	\begin{equation*}
		\re^{\frac{t}{4}\left|\rre\Phi_{23}(\zeta,k)\right|}=\re^{\frac{t}{4}\left|\rre\Phi_{23}(\zeta,k)-\rre\Phi_{23}(\zeta,k_0)\right|}=\re^{\frac{1}{4}\left|\rre\left(\ri z_1^2/2\right)\right|}\leq \re^{\frac{|z_1|^2}{8}}.
	\end{equation*}
	By using Lemma \ref{lem_r_decomposition} and equation \eqref{Error_d1-1} in the Lemma \ref{lem_Y12_bound}, we can obtain for $k\in X^{k_0}_1$
	\begin{align*}
		\left|\left(\tilde{V}^{(k_0)}-V^{X}_{k_0}\right)_{32}\right|&=\left|\hat{r}_{1,a}(-k)\left(d^{(k_0)}_{1}\right)^{-1}\re ^{-2\ri\nu\ln_0(z_{(+)})}\re^{\frac{\ri z_{(+)}^2}{2}}-\hat{r}_{1,a}(-k_0)z_{(+)} ^{-2\ri\nu(q)}\re^{\frac{\ri z_{(+)}^2}{2}}\right|\\
		&=\left|\re ^{-2\ri\nu\ln_0(z_{(+)})}\right|\left|\left(\left(d^{(k_0)}_{1}\right)^{-1}-1\right)\hat{r}_{1,a}(-k)+\left(\hat{r}_{1,a}(-k)-\hat{r}_{1,a}(-k_0)\right)\right|\left|\re^{\frac{\ri z_{(+)}^2}{2}}\right|\\
		&\leq C \left(\left|\left(d^{(k_0)}_{1}\right)^{-1}-1\right|+\left|k-k_0\right|\right)\re^{-ct|k-k_0|^2}\\
		&\leq C|k-k_0|(1+|\ln|k-k_0||)\re^{-ct|k-k_0|^2}.
	\end{align*}
	Thus,
	\begin{equation*}
		\begin{aligned}
			&
			\left\|\left(\tilde{V}^{(k_0)}-V^{X}_{k_0}\right)_{32}\right\|_{L^1\left(X_1^{k_0}\right)}\leq C\int_{0}^{\infty} u (1+|\ln u|)\re^{-ctu^2}\rd u \leq \frac{C \ln t}{t},\\
			&\left\|\left(\tilde{V}^{(k_0)}-V^{X}_{k_0}\right)_{32}\right\|_{L^\infty\left(X_1^{k_0}\right)}\leq C\sup_{u\geq 0} u (1+|\ln u|)\re^{-ctu^2}\leq \frac{C \ln t}{\sqrt{t}}.
		\end{aligned}
	\end{equation*}
	Thus, the case for $k \in X^{k_0}_1$ in equations \eqref{Error_v2-vk1} and \eqref{Error_v2-v-k1} is proved, and the remaining cases can be proved similarly.
	
	 On the other hand, according to equation \eqref{zpm}, we have that $z_{(\pm)}\to\infty$ for $k\in \partial B_\epsilon(\pm k_0)$ as $t\to\infty$. Equation \eqref{Mpm_expand} can be rewritten as:
	\begin{align*}
		M^{X}_{\pm k_0}(q(\zeta),z_{(\pm)}(k))=I+\frac{\left( M^{X}_{\pm k_0}(q(\zeta))\right) _1}{2\sqrt{t}(k\mp k_0)}+\mathcal{O}\left(\frac{1}{t}\right), \quad k\in \partial
		B_\epsilon(\pm k_0),\,\,t\to\infty,
	\end{align*}
	uniformly for $\zeta\in\mathcal{I}_2$. Combining equation \eqref{Mk_0}, we obtain as $t\to\infty$,
	\begin{align*}
		\left(M^{(\pm k_0)}\right)^{-1}-I=-\frac{Y_{(\pm)}(\zeta)\left( M^{X}_{\pm k_0}(q(\zeta))\right) _1Y_{(\pm)}^{-1}(\zeta)}{2\sqrt{t}(k\mp k_0)}+\mathcal{O}\left(\frac{1}{t}\right), 
	\end{align*}
	uniformly for $k\in \partial B_\epsilon(\pm k_0)$. With Lemma \ref{lem_Y12_bound} and the Cauchy formula, the proof of equations \eqref{Mk1-I} and \eqref{intM-I} follows directly from the above equation.\qedhere
	
\end{proof}

   To obtain a small-norm RH problem, we construct using the functions $M^{\pm k_0}$ as follows: 
      \begin{equation}\label{hatM}
        \hat{M}=\left\{
        \begin{aligned}
            &M^{(2)}\left(M^{(\pm k_0)}\right)^{-1},\quad &&k\in B_\epsilon(\pm k_0),\\
            &M^{(2)},\quad &&\text{elsewhere}.
        \end{aligned}
        \right.
    \end{equation}
    The function $\hat M$ is analytic for $k \in \hat{\Sigma}$, where $\hat{\Sigma}=\Sigma^{(2)}\cup \partial B_{\epsilon}$, and $\partial B_{\epsilon}$ are oriented counterclockwise. The jump matrix $\hat V$ are
    \begin{equation}\label{hatv}
        \hat{V}=\left\{
        \begin{aligned}
            &M^{(\pm k_0)}_-V^{(2)}\left(M^{(\pm k_0)}_+\right)^{-1},\quad &&k\in \hat{\Sigma} \cap B_\epsilon(\pm k_0),\\
            & \left(M^{(\pm k_0)}\right)^{-1}, && k\in \partial B_\epsilon(\pm k_0),\\
            &V^{(2)},\quad &&k\in \hat{\Sigma} \setminus B_\epsilon.
        \end{aligned}
        \right.
    \end{equation}

    \begin{lem}\label{lem_hatw}
    Let $\hat{W}=\hat{V}-I$, the following estimates hold uniformly for $\zeta\in\mathcal{I}_2$, $t\geq2$:
     \begin{align*}
            &\left\|(1+|\cdot|) \hat{W}\right\|_{\left(L^1\cap L^\infty\right)(\mathbb{R})}\leq \frac{C}{t^{N}}, \quad \left\|(1+|\cdot|)  \hat{W}\right\|_{\left(L^1\cap L^\infty\right)(\Sigma^{(2)}\setminus(\mathbb{R}\cup B_\epsilon))}\leq C\re^{-ct},\label{Error_hatw_2}\\
            &\left\| \hat{W}\right\|_{\left(L^1\cap L^\infty\right)(\partial B_\epsilon(\pm k_0))}\leq \frac{C}{\sqrt{t}},\quad
            \left\| \hat{W}\right\|_{L^1(X^{\pm k_0})}\leq \frac{C\ln t}{t},\quad 
            \left\|  \hat{W}\right\|_{L^\infty(X^{\pm k_0})}\leq \frac{C\ln t}{\sqrt{t}}.
        \end{align*}
    \end{lem}
    \begin{proof}
       This follows directly from Lemmas \ref{lem_r_decomposition}, \ref{lem_delta}, and \ref{lem_v2_error}-\ref{lem_Error_v2-vk1}.
    \end{proof}

   Define the Cauchy transform of the following form:
    \begin{equation*}
        \left(\mathcal{C}f\right)(z)=\frac{1}{2\pi\ri} \int_{\hat{\Sigma}}\frac{f(z')}{z'-z}\rd z',\quad z\in \mathbb{C}\setminus\hat{\Sigma},
    \end{equation*}
   where the function $f$ is defined on $\hat{\Sigma}$. The Cauchy transform $\left(\mathcal{C}f\right)(z)$ is analytic from $\mathbb{C} \setminus \hat{\Sigma}$ to $\mathbb{C} $, 
  provided that $(1 + |z|)^{1/3} f(z) \in L^3(\hat{\Sigma})$ (i.e., $f\in \dot{L}^3(\hat{\Sigma})$). Moreover, $\mathcal{C}_{\pm} f $ exist almost everywhere for $z \in \hat{\Sigma} $, $\mathcal{C}_+-\mathcal{C}_-=I$ and $\mathcal{C}_{\pm} \in \mathcal{B}(\dot L^3(\hat{\Sigma})) $, where $\mathcal{B}(\dot L^3(\hat{\Sigma})) $ denotes the space of bounded linear operators on $\dot L^3(\hat{\Sigma})$.

  Lemma \ref{lem_hatw} implies that
  \begin{equation*}
            \left\|(1+|\cdot|)  \hat{W}\right\|_{L^1(\hat{\Sigma})}\leq \frac{C}{\sqrt{t}},\quad \left\|(1+|\cdot|)  \hat{W}\right\|_{L^\infty(\hat{\Sigma})}\leq \frac{C\ln t}{\sqrt{t}},
       \quad \zeta\in\mathcal{I}_2,\,\,t\geq2.
  \end{equation*}
  Next, we obtain by calculation:
  \begin{equation}\label{hatw_Lp}
      \left\|(1+|\cdot|)  \hat{W}\right\|_{L^p(\hat{\Sigma})}\leq \frac{C\ln t^{\frac{p-1}{p}}}{\sqrt{t}},\quad 1\leq p\leq \infty.
  \end{equation}
  This estimate shows that $\hat{W}\in (\dot{L}^3\cap L^\infty)(\hat{\Sigma})$, and we can define:
  \begin{equation*}
      \mathcal{C}_{\hat{W}(x,t,\cdot)}f=\mathcal{C}_-(f\hat
      W),\quad \mathcal{C}_{\hat{W}}: (\dot{L}^3\cap L^\infty)(\hat{\Sigma})\to\dot{L}^3(\hat{\Sigma}).
  \end{equation*}
  \begin{lem}\label{lem_I-C_inverse}
      For any sufficiently large $t$ and $\zeta\in\mathcal{I}_2$, $I-\mathcal{C}_{\hat{W}(x,t,\cdot)}\in \mathcal{B}(\dot L^3(\hat{\Sigma})) $ is invertible.
  \end{lem}
  \begin{proof}
      For each $f\in \dot L^3(\hat{\Sigma})$, we have
      \begin{equation*}
          \left\|\mathcal{C}_{\hat W}f\right\|_{\dot L^3(\hat{\Sigma})}\leq \left\|\mathcal{C}_-\right\|_{\mathcal{B}(\dot L^3(\hat{\Sigma}))}\left\|\hat{W}\right\|_{L^\infty(\hat{\Sigma})}\left\|f\right\|_{\dot L^3(\hat{\Sigma})}.
      \end{equation*}
      Thus $\left\|\mathcal{C}_{\hat W}\right\|_{\mathcal{B}(\dot L^3(\hat{\Sigma}))}\leq \left\|\mathcal{C}_-\right\|_{\mathcal{B}(\dot L^3(\hat{\Sigma}))}\left\|\hat{W}\right\|_{L^\infty(\hat{\Sigma})}\leq \frac{C\ln t}{\sqrt{t}}\left\|\mathcal{C}_-\right\|_{\mathcal{B}(\dot L^3(\hat{\Sigma}))}$. Equivalently, we can find a sufficiently large time $t$ such that  $\left\|\hat{W}\right\|_{L^\infty(\hat{\Sigma})}\leq \left(\left\|\mathcal{C}_-\right\|_{\mathcal{B}(\dot L^3(\hat{\Sigma}))}\right)^{-1}$, the operator $I-\mathcal{C}_{\hat{W}(x,t,\cdot)}\in \mathcal{B}(\dot L^3(\hat{\Sigma}))$  becomes invertible.
  \end{proof}
    
   Based on the conclusion of Lemma \ref{lem_I-C_inverse}, we can define a function  $\hat{\mu}(x,t,k)$ for  sufficiently large $t$:
  \begin{equation}\label{hatnu}
      \hat{\mu}=I+\left(I-\mathcal{C}_{\hat{W}}\right)^{-1}\mathcal{C}_{\hat{W}}I\in I+\dot L^3(\hat{\Sigma}).
  \end{equation}
  And define the space $\dot{E}^3(\mathbb{C} \setminus \hat{\Sigma})$ to include analytic functions $f : \mathbb{C} \setminus\hat{\Sigma} \rightarrow \mathbb{C}$ with the property that for each component $D$ of $\mathbb{C} \setminus \hat{\Sigma}$ there exist curves $\{C_n\}_1^\infty$ in $D$ such that the $C_n$ eventually surround each compact subset of $D$ and
  \begin{equation*}
  	\sup_{n \geq 1} \int_{C_n} (1 + |k|) |f(k)|^3 |dk| < \infty.
  \end{equation*}
  For specific details, refer to Reference \cite{Lenells_2018}.

  \begin{lem}\label{lem_hatM_solution}
      For sufficiently large time $t$, the RH problem for $\hat{M}$ in \eqref{hatM} has a unique solution $\hat{M}\in I+\dot{E}^3(\mathbb{C} \setminus \hat{\Sigma})$, which
      \begin{equation*}
          \hat{M}(x,t,k)=I+\mathcal{C}(\hat\mu \hat W)=I+\frac{1}{2\pi\ri}\int_{\hat{\Sigma}}\frac{\hat{\mu}(x,t,k')\hat{W}(x,t,k')}{k'-k}\rd k'.
      \end{equation*}
  \end{lem}
  \begin{lem}\label{Error_mu-I}
      For sufficiently large $t$ and $\zeta\in\mathcal{I}_2$, we have
      \begin{equation*}
          \left\|\hat{\mu}-I\right\|_{L^p(\hat{\Sigma})}\leq \frac{C\ln t^{\frac{p-1}{p}}}{\sqrt{t}}, \quad 1<p<\infty.
      \end{equation*}
  \end{lem}
  \begin{proof}
      Let $\|\mathcal{C}_-\|_p:=\|\mathcal{C}_-\|_{\mathcal{B}(L^p(\hat{\Sigma}))}<\infty$ and assume the time $t$ is sufficiently large such that $\|\hat{W}\|_{L^\infty(\hat{\Sigma})}<\|\mathcal{C}_-\|_p^{-1}$. Using \eqref{hatnu} and Neumann series:
      \begin{equation*}
          \|\hat{\mu}-I\|_{L^p(\hat{\Gamma})}\leq \sum_{j=0}^{\infty}\|\mathcal{C}_{\hat{w}}\|_{\mathcal{B}(L^p(\hat{\Gamma}))}^{j}\|\mathcal{C}_{\hat{w}}I\|_{L^p(\hat{\Gamma})}\leq \sum_{j=0}^{\infty}\|\mathcal{C}_-\|_p^{j+1}\|\hat{w}\|_{L^\infty(\hat{\Gamma})}^{j}\|\hat{w}\|_{L^p(\hat{\Gamma})}=\frac{\|\mathcal{C}_-\|_p\|\hat{w}\|_{L^p(\hat{\Gamma})}}{1-\|\mathcal{C}_-\|_p\|\hat{w}\|_{L^\infty(\hat{\Gamma})}}.
      \end{equation*}
      Thus, based on the estimates in equation \eqref{hatw_Lp}, we can complete the proof.
  \end{proof}

   Consider the following nontangential limit and decompose it using Lemmas \ref{lem_hatw} and \ref{Error_mu-I}, we obtain:
   \begin{align*}
       Q(x,t)&=\lim_{k\to\infty}k\left(\hat{M}(x,t,k)-I\right)=-\frac{1}{2\pi\ri}\int_{\hat{\Sigma}}\hat{\mu}(x,t,k)\hat{W}(x,t,k)\rd k\\
       &=-\frac{1}{2\pi\ri}\int_{\partial B_\epsilon}\hat{W}(x,t,k)\rd k-\frac{1}{2\pi\ri}\int_{\hat{\Sigma}\setminus\partial B_\epsilon}\hat{W}(x,t,k)\rd k-\frac{1}{2\pi\ri}\int_{\hat{\Sigma}}(\hat{\mu}(x,t,k)-I)\hat{W}(x,t,k)\rd k\\
       &=-\frac{1}{2\pi\ri}\int_{\partial B_\epsilon }\hat{W}(x,t,k)\rd k+\mathcal{O}\left(\frac{\ln t}{t}\right).
   \end{align*}
   The definition of $\hat{M}$ in \eqref{hatM} and jump matrix $\hat{V}$ in \eqref{hatv} lead to the following relation
   \begin{align*}
       Q(x,t)&=-\frac{1}{2\pi\ri}\int_{\partial B_\epsilon }\left(\left(M^{(\pm k_0)}\right)^{-1}-I\right)\rd k+\mathcal{O}\left(\frac{\ln t}{t}\right)\\
       &=\frac{Y_{(+)}(\zeta)\left( M^{X}_{ k_0}(q(\zeta))\right) _1Y_{(+)}^{-1}(\zeta)}{2\sqrt{t}}+\frac{Y_{(-)}(\zeta)\left( M^{X}_{- k_0}(q(\zeta))\right) _1Y_{(-)}^{-1}(\zeta)}{2\sqrt{t}}+\mathcal{O}\left(\frac{\ln t}{t}\right).
   \end{align*}
   
   Combining all the transformations made in this subsection together with Theorem \ref{theo_reconstruction formula}, Lemmas \ref{lem_delta} and \ref{lem_T_bound}, we obtain:
   \begin{align*}
        r(x,t)&=2\lim\limits_{k\to\infty}k\left(M-I \right)_{31}=2\lim\limits_{k\to\infty}k\left(\hat{M}T^{-1} \Delta^{-1}-I\right)_{31}\\
        &=\mathcal{O}\left(\frac{\ln t}{t}\right),\\
    q(x,t)&=\ri\lim\limits_{k\to\infty}k\left(M-I \right)_{21}=\ri \lim\limits_{k\to\infty}\left(\hat{M}T^{-1} \Delta^{-1}-I\right)_{21}\\
    &=\frac{\sqrt{2\pi}\left(2\sqrt{t}\right)^{-2\ri\nu}\text{exp}\left(\ri\nu \ln(2k_0)+\frac{3\pi \ri}{4}+2\ri tk_0^2-\frac{\pi\nu}{2}+s_1\right)}{2\sqrt{t}r_1(-k_0)\Gamma(-\ri\nu)}+\mathcal{O}\left(\frac{\ln t}{t}\right),
    \end{align*}
    where $s_1=\frac{1}{2\pi\ri}\int_{k_0}^{\infty}\ln\left({(s+k_0)}{(s-k_0)^2}\right){\rm{d}}\ln(1-2\sigma \left|r_1(-s)\right|^2 ).$

    \subsection{Long-time asymptotics in Region {\rm{I}}}
    In this region, $C\leq x/t\leq +\infty$ holds for sufficiently large $C$. Therefore, it is necessary to introduce the parameter $0\leq \tau:=t/x\leq C^{-1}$ for the analysis. Moreover, the dispersion functions $\Phi_{ij}$ satisfy:
    \begin{equation*}
    	t\Phi_{ij}(\zeta,k)=x\tilde\Phi_{ij}(\tau,k)=x[(u_i(k)-u_j(k))+(v_i(k)-v_j(k))\tau],\quad  1\leq i<j\leq3.
    \end{equation*} 
    Moreover, we can rewrite $k_0$ as $k_0=\frac{1-\tau}{2\tau}$ and
    \begin{equation*}
    	\begin{aligned}
    	&x\tilde\Phi_{12}(\tau,k)=x\tilde\Phi_{12}(\tau,-k_0)+\frac{\ri}{2}\left( 2\sqrt{x\tau}(k+k_0)\right) ^2,\quad x\tilde\Phi_{23}(\tau,k)=x\tilde\Phi_{12}(\tau,k_0)-\frac{\ri}{2}\left( 2\sqrt{x\tau}(k-k_0)\right) ^2.
    	\end{aligned}
    \end{equation*}
    
   The asymptotic analysis in this region is very similar to that in Region II. In fact, $\Phi_{ij}$ and $\tilde\Phi_{ij}$ satisfy the same signature of their real parts, with the difference that in Region I we consider the variable $\tau$ and the limit $x\to+\infty$. We will not repeat the same decomposition steps of the RH problem here, and directly present the asymptotic behavior of the solution to the initial value problem \eqref{Newell_initial} in this region under the new parameter:
    \begin{equation}\label{solution_I}
    	\begin{aligned}
    		r(x,t)&=\mathcal{O}\left(\frac{1}{x^{N}}+\frac{C_N(\tau)\ln x}{x}\right),\\
    		q(x,t)&=\frac{\sqrt{2\pi}\left(2\sqrt{x\tau}\right)^{-2\ri\nu}\text{exp}\left(\ri\nu \ln(2k_0)+\frac{3\pi \ri}{4}+2\ri tk_0^2-\frac{\pi\nu}{2}+s_1\right)}{2\sqrt{x\tau}r_1(-k_0)\Gamma(-\ri\nu)}+\mathcal{O}\left(\frac{1}{x^{N}}+\frac{C_N(\tau)\ln x}{x}\right),
    	\end{aligned}
    \end{equation}
    where $k_0=\frac{1-\tau}{2\tau}$, $\nu$ and $s_1$ are given in equations \eqref{nu_chi} and \eqref{s1}.

    \begin{remark}
    	Since $k_0\to+\infty$ as $\tau \to 0$, $r_1(k)$ vanishes to all orders at $k=\pm k_0$, it follows that  $\nu$ and $s_1$ vanish to all orders. Consequently, equation \eqref{solution_I} implies that as $x\to+\infty$:
    	\begin{equation*}
    			r(x,t)=\mathcal{O}\left(\frac{1}{x^{N}}+\frac{C_N(\tau)}{x}\right),\quad
    			q(x,t)=\mathcal{O}\left(\frac{1}{x^{N}}+\frac{C_N(\tau)}{x}\right),
    	\end{equation*}
    	uniformly for $\tau\in\mathcal{I}_1$. In particular, for any fixed $t\geq 0$, the above expression can be reduced to $\mathcal{O}\left({x^{-N}}\right) $ as $x \to +\infty$.
    \end{remark}

\subsection{Long-time asymptotics in Region {\rm{III}}}\label{subsec_regionIII}
	In Subsection \ref{subsec_regionII}, we performed transformations and decompositions on RH problem \ref{rhp_M}, and using the solvable $3 \times 3$ parabolic cylinder model RH problem \ref{model-rhp1}, we derived the asymptotic behavior of the Newell equation in Region II. This section will discuss the behavior in Region III, which can still lead to a solvable parabolic cylinder model. However, the following analysis remains necessary, as the two regions involve different transformations of the RH problem \ref{rhp_M}, and the dispersion relation exhibits properties in this region that differ from those in Region II. We first present the signatures of the real parts of the functions $\Phi_{ij}$, $1\leq i<j\leq 3$, for $(x,t)$ in Region III, i.e., $\zeta\in\mathcal{I}_3$, as shown in Figure \ref{fig_theta_sign2}.

      \begin{figure}[htbp]
    \centering
    \begin{subfigure}[t]{0.28\textwidth}
        \centering
     	\begin{tikzpicture} [scale=0.7]
        \definecolor{mycolor}{HTML}{f5d5d8}
        
         \fill[mycolor] (-3,0) -- (1.5,0) -- (1.5,3) -- (-3,3) -- cycle;
        \fill[mycolor] (1.5,0) -- (3,0) -- (3,-3) -- (1.5,-3) -- cycle;
	 		
	  \draw [very thick,black](-3,0) -- (3,0);
        \draw [very thick,black](1.5,-3) -- (1.5,3);

        \fill (-1.5,0) circle (1.5pt);
        \fill (0,0) circle (1.5pt);
	  \fill (1.5,0) circle (1.5pt);

        \node[below] at (1.8,0) {$-k_0$};
	  \node[below] at (0,0) {$\text{0}$};
	  \node[below] at (-1.5,0) {$k_0$};

       \node at (2.4,2) {$\tilde{D}_1^{(12)}$};
      \node at (-1,2) {$\tilde{D}_2^{(12)}$};
      \node at (-1,-2) {$\tilde{D}_3^{(12)}$};
      \node at (2.4,-2) {$\tilde{D}_4^{(12)}$};

        \node[right] at (3,0) {$\mathbb{R}$};
     	\end{tikzpicture}
     		\caption{The signature of $\Phi_{12}$.}
    \end{subfigure}
    \quad
    \begin{subfigure}[t]{0.28\textwidth}
        \centering
     	\begin{tikzpicture} [scale=0.7]
     	\definecolor{mycolor}{HTML}{f5d5d8}
        
         \fill[mycolor] (-3,0) -- (3,0) -- (3,3) -- (-3,3) -- cycle;
         \fill[white] (-3,0) -- (3,0) -- (3,-3) -- (-3,-3) -- cycle;
	 		
	  \draw [very thick,black](-3,0) -- (3,0);

        \fill (-1.5,0) circle (1.5pt);
        \fill (0,0) circle (1.5pt);
	  \fill (1.5,0) circle (1.5pt);

        \node[below] at (-1.5,0) {$k_0$};
	  \node[below] at (0,0) {$\text{0}$};
	  \node[below] at (1.5,0) {$-k_0$};

       \node at (0,2) {$\tilde{D}_1^{(13)}$};
      \node at (0,-2) {$\tilde{D}_2^{(13)}$};

        \node[right] at (3,0) {$\mathbb{R}$};
     	\end{tikzpicture}
     		\caption{The signature of $\Phi_{13}$.\label{subfig_theta13_2}}
    \end{subfigure}
    \quad
    \begin{subfigure}[t]{0.28\textwidth}
        \centering
     	\begin{tikzpicture} [scale=0.7]
     		\definecolor{mycolor}{HTML}{f5d5d8}
        
         \fill[mycolor] (3,0) -- (-1.5,0) -- (-1.5,3) -- (3,3) -- cycle;
        \fill[mycolor] (-1.5,0) -- (-3,0) -- (-3,-3) -- (-1.5,-3) -- cycle;
	 		
	  \draw [very thick,black](-3,0) -- (3,0);
        \draw [very thick,black](-1.5,-3) -- (-1.5,3);

        \fill (-1.5,0) circle (1.5pt);
        \fill (0,0) circle (1.5pt);
	  \fill (1.5,0) circle (1.5pt);

        \node[below] at (-1.2,0) {$k_0$};
	  \node[below] at (0,0) {$\text{0}$};
	  \node[below] at (1.5,0) {$-k_0$};

      \node at (-2.4,2) {$\tilde{D}_2^{(23)}$};
      \node at (1,2) {$\tilde{D}_1^{(23)}$};
      \node at (1,-2) {$\tilde{D}_4^{(23)}$};
      \node at (-2.4,-2) {$\tilde{D}_3^{(23)}$};

        \node[right] at (3,0) {$\mathbb{R}$};
     	\end{tikzpicture}
     		\caption{The signature of $\Phi_{23}$.}
             \end{subfigure}
             \caption{ Open sets in the complex $k$-plane for Region \rm{III}: $\rre \Phi_{ij}>0$ (shaded) and $\rre \Phi_{ij}<0$ (white).}
             \label{fig_theta_sign2}
\end{figure}

    For the convenience of notation in the subsequent analysis, we introduce the following notations:
\begin{align*}
    &\tilde{r}_2(k):=\frac{r_2(k)}{1-2\sigma |r_1(-k)|^2+|r_2(k)|^2},\qquad\tilde{\alpha}(k):=\frac{\alpha(k)}{1-2\sigma |r_1(-k)|^2+|r_2(k)|^2}.
\end{align*}
 and present the lemma on the analytic approximations of the reflection coefficients.
\begin{lem}\label{alpha_decomposition}
For each $\zeta\in\mathcal{I}_3$ and $t>0$, there are several decompositions
    \begin{align*}
        &\tilde{r}_2(k)=\tilde{r}_{2,a}(x,t,k)+\tilde{r}_{2,r}(x,t,k),\quad && k\in\mathbb{R}_+,\\
        &\hat{\alpha}(k)=\hat{\alpha}_a(x,t,k)+\hat{\alpha}_r(x,t,k), &&k\in \left(-\infty,k_0\right],\\
        &\tilde{\alpha}(k)=\tilde{\alpha}_a(x,t,k)+\tilde{\alpha}_r(x,t,k), &&k\in \left[k_0,\infty\right),
    \end{align*}
    where the above functions satisfy the following properties:
    \begin{itemize}
     \item $\tilde{r}_{2,a}(x,t,k)$ is defined and continuous for $k\in \{k\,|\,\rre k\geq0,\rim k\leq0\}$, and is analytic for $k\in \{k\,|\,\rre k>0,\rim k<0\}$. $\hat{\alpha}_a(k)$ is defined and continuous for $k\in \overline{\tilde{D}_2^{(23)}}$, and is analytic for $k\in \tilde{D}_2^{(23)}$.  $\tilde{\alpha}_a(k)$ is defined and continuous for $k\in \overline{\tilde{D}_4^{(23)}}$, and is analytic for $k\in \tilde{D}_4^{(23)}$.
    \item The functions $\tilde{r}_{2,a}(k)$, $\hat{\alpha}_a(k)$, and $\tilde{\alpha}_a(k)$ satisfy
        \begin{align*}
            &\left| \tilde{r}_{2,a}(x,t,k)-\sum_{j=0}^{N}\frac{\tilde r^{(j)}_{2,a}(0)}{j!}k^j \right|\leq C \left|k\right|^{N+1}\re^{\frac{t}{4}\left|\rre\Phi_{13}(\zeta,k)\right|},\quad&& k\in \{k\,|\,\rre k\geq0,\rim k\geq0\},\\
            &\left| \hat{\alpha}_a(x,t,k)-\sum_{j=0}^{N}\frac{\hat{\alpha}^{(j)}_a(k_0)}{j!}(k-k_0)^{j} \right|\leq C \left|k-k_0\right|^{N+1}\re^{\frac{t}{4}\left|\rre\Phi_{23}(\zeta,k)\right|},\quad&&  k\in \overline{\tilde{D}_2^{(23)}},\\
            &\left| \tilde{\alpha}_a(x,t,k)-\sum_{j=0}^{N}\frac{\tilde{\alpha}^{(j)}_a(k_0)}{j!}(k-k_0)^{j} \right|\leq C \left|k-k_0\right|^{N+1}\re^{\frac{t}{4}\left|\rre\Phi_{23}(\zeta,k)\right|},\quad && k\in \overline{\tilde{D}_4^{(23)}},
            \end{align*}
            and 
            \begin{align*}
            &\left| \tilde{r}_{2,a}(x,t,k) \right|\leq  \frac{C}{1+\left|k\right|}\re^{\frac{t}{4}\left|\rre\Phi_{13}(\zeta,k)\right|},\quad&& k\in \{k\,|\,\rre k\geq0,\rim k\geq0\},\\
            &\left| \tilde{\alpha}_a(x,t,k) \right|\leq  \frac{C}{1+\left|k\right|}\re^{\frac{t}{4}\left|\rre\Phi_{23}(\zeta,k)\right|},\quad && k\in \overline{\tilde{D}_2^{(23)}},\\
            &\left| \tilde{\alpha}_a(x,t,k) \right|\leq  \frac{C}{1+\left|k\right|}\re^{\frac{t}{4}\left|\rre\Phi_{23}(\zeta,k)\right|},\quad  &&k\in \overline{\tilde{D}_4^{(23)}},
            \end{align*}
        where the constant $C$ is independent of $\zeta,t,k$.
    \item For each $1\leq p\leq \infty$,
    \begin{align*}
        &\left\|(1+|\cdot|) \tilde{r}_{2,r}(x,t,\cdot)\right\|_{\mathbb{R}_+}=\mathcal{O}\left(t^{-N}\right),\\
        &\left\|(1+|\cdot|) \hat{\alpha}_r(x,t,\cdot)\right\|_{\left(-\infty,k_0\right)}=\left\|\frac{\hat{\alpha}_r(x,t,\cdot)}{\cdot-k_0}\right\|_{\left(-\infty,k_0\right)}=\mathcal{O}\left(t^{-N}\right),\\
        &\left\|(1+|\cdot|) \tilde{\alpha}_r(x,t,\cdot)\right\|_{\left(k_0,\infty\right)}=\left\|\frac{\tilde{\alpha}_r(x,t,\cdot)}{\cdot-k_0}\right\|_{\left(k_0,\infty\right)}=\mathcal{O}\left(t^{-N}\right),
    \end{align*}
    uniformly for $\zeta\in\mathcal{I}_3$ as $t\to\infty$.
        \end{itemize}
\end{lem}

   Note that the key difference between this region and Region II is that Figure \ref{subfig_theta13_2} and Figure \ref{subfig_theta13_1} are in opposite states. For the upcoming decomposition of the jump matrix on the real axis $\mathbb{R}$, we need to introduce the following three scalar RH problems on $\left(-\infty,k_0\right]$ and $[k_0,-k_0]$:
   \begin{equation}\label{vol_delta123}
   	\begin{aligned}
   		&\delta_1(\zeta,k)={\rm{exp}}\left\{\frac{1}{2 \pi\ri}\int_{-\infty}^{k_0}\frac{\ln \left(1-2\sigma|r_1(s)|^2\right)}{s-k}{\rm{d}}s\right\},\quad &&k\in \mathbb{C}\backslash \left(-\infty,k_0 \right],\\
   		&\delta_2(\zeta,k)={\rm{exp}}\left\{\frac{1}{2 \pi\ri}\int_{-\infty}^{k_0}\frac{\ln \left(1-2\sigma|r_1(s)|^2+|r_2(-s)|^2\right)}{s-k}{\rm{d}}s\right\},\quad &&k\in \mathbb{C}\backslash \left(-\infty,k_0 \right],\\
   		&\delta_3(\zeta,k)={\rm{exp}}\left\{\frac{1}{2 \pi\ri}\int_{k_0}^{-k_0}\frac{\ln \left(1-2\sigma|r_1(-s)|^2-|r_2(s)|^2\right)}{s-k}{\rm{d}}s\right\},\quad &&k\in \mathbb{C}\backslash \left[k_0,-k_0 \right].
   	\end{aligned}
    \end{equation}

 \begin{lem}\label{lem_delta1}
    	 The functions $\delta_j(\zeta,k)$, $j=1,2,3,$ satisfy the following properties for $\zeta\in\mathcal{I}_3$:
    	\begin{enumerate}
     		\item The functions 
     		$\delta_j^{\pm 1}(\zeta,k)$, $j=1,2,$ are analytic on $\mathbb{C}\setminus\left(-\infty,k_0 \right]$, and $\delta_3^{\pm 1}(\zeta,k)$ are analytic on $\mathbb{C}\setminus\left[k_0,-k_0 \right]$, and they can be written as
			\begin{equation*}
				\begin{aligned}
				&\delta_1(\zeta,k)=\re ^{\ri \nu_1\ln_\pi (k-k_0)}{\rm{e}}^{-\chi_1(\zeta,k)},\\
                &\delta_2(\zeta,k)=\re ^{\ri \nu_2\ln_\pi (k-k_0)}{\rm{e}}^{-\chi_2(\zeta,k)},\\
                &\delta_3(\zeta,k)=\re ^{\ri \nu_2\ln_\theta (k+k_0)-\ri\nu_3\ln_\theta (k-k_0)}{\rm{e}}^{-\chi_3(\zeta,k)},
				\end{aligned}
			\end{equation*}
			where $\theta=0$ when $k$ is near $k_0$, and $\theta=\pi$ when $k$ is near $-k_0$,
			\begin{equation*}
				\begin{aligned}
					&\nu_1=-\frac{1}{2\pi} \ln \left(1-2\sigma|r_1(k_0)|^2\right),\\
                    &\chi_1(\zeta,k)=\frac{1}{2\pi\ri}\int_{-\infty}^{k_0}\ln_\pi(k-s){\rm{d}}\ln \left(1-2\sigma|r_1(s)|^2\right),\\
					&\nu_2=-\frac{1}{2\pi} \ln \left(1-2\sigma|r_1(k_0)|^2+|r_2(-k_0)|^2\right),\\
                    &\chi_2(\zeta,k)=\frac{1}{2\pi\ri}\int_{-\infty}^{k_0}\ln_\pi(k-s){\rm{d}}\ln \left(1-2\sigma|r_1(s)|^2+|r_2(-s)|^2\right),\\
					&\nu_3=-\frac{1}{2\pi} \ln \left(1-2\sigma|r_1(-k_0)|^2+|r_2(k_0)|^2\right),\\
                    &\chi_3(\zeta,k)=\frac{1}{2\pi\ri}\int_{k_0}^{-k_0}\ln_\theta(k-s){\rm{d}}\ln \left(1-2\sigma|r_1(-s)|^2+|r_2(s)|^2\right).
				\end{aligned}
			\end{equation*}
	
    	\item For each $\zeta\in\mathcal{I}_3$, $\delta_j(k)$, $j=1,2$, are bounded in $\mathbb{C}\setminus\left(-\infty,k_0 \right]$, $\delta_3(k)$ is bounded in $\mathbb{C}\setminus\left[k_0,-k_0 \right]$ and  $
    		\delta_j^{-1}(k)= \overline{\delta_j(\overline{k})}.$
        \item As $k\to \pm k_0$ along the non-tangential direction of $\mathbb{R}$, we have the following equations hold, for $C$ independent of $\zeta$:
        \begin{align*}
            &\left|\chi_j(\zeta,k)-\chi_j({\zeta,k_0})\right|\leq C|k-k_0|\left(1+|\ln|k-k_0||\right),\quad j=1,2,\\
            &\left|\chi_3(\zeta,k)-\chi_3({\zeta,\pm k_0})\right|\leq C|k\mp k_0|\left(1+|\ln|k\mp k_0||\right).
             \end{align*}
    	\end{enumerate}
    \end{lem}

    \begin{proof}
        The lemma is derived from equations \eqref{vol_delta123}, and can be proved by direct estimation.
    \end{proof}

   With the functions $\delta_j(k)$, $j=1,2,3,$ constructed from \eqref{vol_delta123}, we define the following matrix transformation:
    \begin{equation}\label{trans_Delta1}
        m^{(1)}(x,t,k)=M(x,t,k)\Delta_1(k), \quad \Delta_1(k)=\begin{pmatrix}
        	\frac{\delta_1(-k)\delta_2(k)}{\delta_2(-k)\delta_3(-k)} & 0 & 0\\
        	0 & \frac{\delta_3(k)\delta_3(-k)}{{\delta_1(k)}{\delta_1(-k)}} & 0\\
        	0 & 0 & \frac{\delta_1(k)\delta_2(-k)}{\delta_2(k)\delta_3(k)}\\
        \end{pmatrix}.
    \end{equation}
   And the matrix $\Delta_1(k)$ satisfies the following two properties:
    \begin{align*}
    	&\Delta_1(\zeta,k)=I+\mathcal{O}\left(k^{-1}\right),\quad k\to\infty,\\
        &\Delta_1^\dagger(\zeta,\bar k)=\mathcal{A}\Delta_1^{-1}(\zeta,k)\mathcal{A}^{-1},\quad \Delta_1(\zeta,k)=\mathcal{B}\Delta_1(\zeta,-k)\mathcal{B}.
    \end{align*}

    The function $m^{(1)}$ obtained from transformation \eqref{trans_Delta1} is analytic on $k\in \mathbb{C}\setminus\Gamma^{(1)}$ ($\Gamma^{(1)}$ is essentially consistent with the form of $\Sigma^{(1)}$ in Figure \ref{fig_Gamma1}, except that the positions of $k_0$ and $-k_0$ are swapped) and satisfies the jump condition $m^{(1)}_+(x,t,k)=m^{(1)}_-(x,t,k)v_j^{(1)}(x,t,k)$ for $k\in\Gamma^{(1)}_j$, $j=1,2,3$, where $v^{(1)}_j(x,t,k):=v^{(1)}_{j,l}(x,t,k)v^{(1)}_{j,r}(x,t,k)$ are given by
    \begin{equation*}
     \begin{aligned}
            &v^{(1)}_{1,l}
            =\begin{pmatrix}
					1 &-\hat{r}_1(k)\tilde{\delta}^{-1}_{21-}(k)\re^{t\Phi_{12}} & -\tilde{r}_2^*(-k)\frac{\tilde{\delta}_{23-}(k)}{\tilde{\delta}_{21-}(k)}\re^{t\Phi_{13}}\\
					0 & 1 & 0\\
					0 & \hat \alpha^*(k)\tilde{\delta}_{23-}^{-1}(k)\re^{-t\Phi_{23}} & 1\\
				\end{pmatrix},\\
                & v^{(1)}_{1,r}=\begin{pmatrix}
					1 & 0 & 0\\
					2\sigma \hat{r}_1^*(k)\tilde{\delta}_{21+}(k)\re^{-t\Phi_{12}} & 1  & -2\sigma  \hat \alpha(k)\tilde{\delta}_{23+}(k)\re^{t\Phi_{23}}\\
					-\tilde{r}_2(-k) \frac{\tilde{\delta}_{21+}(k)}{\tilde{\delta}_{23+}(k)}\re^{-t\Phi_{13}} & 0 & 1\\
				\end{pmatrix}, \\
                &v^{(1)}_{2,l}
                =\begin{pmatrix}
					1 & -\tilde\alpha^*(-k)\tilde{\delta}^{-1}_{21-}(k)\re^{t\Phi_{12}}    & \tilde{r}_2(k)\frac{\tilde{\delta}_{23-}(k)}{\tilde{\delta}_{21-}(k)}\re^{t\Phi_{13}} \\
					0 & 1 & -2 \tilde\alpha(k)\tilde{\delta}_{23-}(k)\re^{t\Phi_{23}} \\
					0 & 0 & 1\\
				\end{pmatrix} ,\\
                & v^{(1)}_{2,r}=\begin{pmatrix}
					1 & 0  & 0\\
					 2 \sigma  \tilde\alpha(-k)\tilde{\delta}_{21+}(k)\re^{-t\Phi_{12}}   & 1  & 0\\
					\tilde{r}_2^*(k)\frac{\tilde{\delta}_{21+}(k)}{\tilde{\delta}_{23+}(k)}\re^{-t\Phi_{13} }   & \tilde \alpha^*(k)\tilde{\delta}_{23+}^{-1}(k)\re^{-t\Phi_{23} }   & 1\\
				\end{pmatrix},\\
               & v^{(1)}_{3,l}
            =\begin{pmatrix}
					1 & 0  & \tilde{r}_2(k)\frac{\tilde{\delta}_{23-}(k)}{\tilde{\delta}_{21-}(k)}\re^{t\Phi_{13} } \\
				2 \sigma \hat\alpha(-k)\tilde{\delta}_{21-}(k)\re^{-t\Phi_{12} }  & 1  & -2\sigma  \tilde\alpha(k)\tilde{\delta}_{23-}(k)\re^{t\Phi_{23} } \\
					0  & 0 & 1\\
				\end{pmatrix},\\
                & v^{(1)}_{3,r}=\begin{pmatrix}
					1 & -\hat\alpha^*(-k)\tilde{\delta}_{21+}^{-1}(k)\re^{t\Phi_{12} } & 0\\
					0  & 1  & 0\\
					\tilde{r}_2^*(k)\frac{\tilde{\delta}_{21+}(k)}{\tilde{\delta}_{23+}(k)}\re^{-t\Phi_{13} }   & \tilde \alpha^*(k)\tilde{\delta}_{23+}^{-1}(k)\re^{-t\Phi_{23} } & 1\\
				\end{pmatrix},
            \end{aligned}
            \end{equation*}    
              where $\tilde{\delta}_{21}(k)=\frac{\delta_1(k)\delta_1^2(-k)\delta_2(k)}{\delta_2(-k)\delta_3(k)\delta_3^2(-k)}$ and $\tilde{\delta}_{23}(k)=\tilde{\delta}_{21}(-k)=\frac{\delta_1^2(k)\delta_1(-k)\delta_2(-k)}{\delta_2(k)\delta_3^2(k)\delta_3(-k)}$.

     Similarly, we define the following matrix functions $R(\zeta,k)=R_j(\zeta,k)$ for $k\in \tilde D_j$, $j=1,2,\cdots,10,$ where the definitions of $\tilde D_j$ refer to Figure \ref{fig_Gamma2}, with the difference that the positions of $k_0$ and $-k_0$ in Figure \ref{fig_Gamma2} are swapped, while the jump contour $\Gamma^{(2)}$, as well as the labels of regions $\tilde D_j$, are still given according to their positions in Figure \ref{fig_Gamma2}:
     \begin{align}
     	&R_1(\zeta,k)=\begin{pmatrix}
     		1 & \hat\alpha_{a}^*(-k)\tilde{\delta}_{21}^{-1}(k)\re^{t\Phi_{12} } & 0\\
     		0  & 1  & 0\\
     		-\tilde{r}_{2,a}^*(k)\frac{\tilde{\delta}_{21}(k)}{\tilde{\delta}_{23}(k)}\re^{-t\Phi_{13} }   & -\hat{r}_{1,a}(-k) \tilde{\delta}_{23}^{-1}(k)\re^{-t\Phi_{23} } & 1\\
     	\end{pmatrix},\nonumber\\
     	&R_2(\zeta,k)=\begin{pmatrix}
     		1 & 0  & 0\\
     		-2\sigma  \tilde\alpha_{a}(-k)\tilde{\delta}_{21}(k)\re^{-t\Phi_{12} }   & 1 & 0\\
     		\tilde{r}_{2,a}(-k)\frac{\tilde{\delta}_{21}(k)}{\tilde{\delta}_{23}(k)}\re^{-t\Phi_{13} }  & -\tilde\alpha_{a}^*(k)\tilde{\delta}_{23}^{-1}(k)\re^{-t\Phi_{23} }  & 1\\
     	\end{pmatrix},\nonumber\\
     	& R_3(\zeta,k)=\begin{pmatrix}
    	1 & -\tilde\alpha_a^*(-k)\tilde{\delta}^{-1}_{21}(k)\re^{t\Phi_{12} }  & \tilde{r}_{2,a}(k)\frac{\tilde{\delta}_{23}(k)}{\tilde{\delta}_{21}(k)}\re^{t\Phi_{13} } \\
     	0  & 1  & -2\sigma \tilde\alpha_a(k)\tilde{\delta}_{23}(k) \re^{t\Phi_{23} } \\
     	0  & 0  & 1\\
     \end{pmatrix},\nonumber\\
     &R_4(\zeta,k)=\begin{pmatrix}
     	1 & 0  & \tilde{r}_{2,a}(k)\frac{\tilde{\delta}_{23}(k)}{\tilde{\delta}_{21}(k)}\re^{t\Phi_{13} } \\
     	2\sigma \hat\alpha_a(-k)\tilde{\delta}_{21}(k)\re^{-t\Phi_{12} }  & 1  & -2\sigma  \tilde\alpha_a(k)\tilde{\delta}_{23}(k)\re^{t\Phi_{23} } \\
     	0  & 0 & 1\\
     \end{pmatrix},\nonumber\\
     & R_5(\zeta,k)=\begin{pmatrix}
     1 & 0 & 0\\
     -2\sigma \tilde{\alpha}_a(-k)\tilde{\delta}_{21}(k)\re^{-t\Phi_{12}} & 1  & 2 \sigma \hat \alpha_a(k)\tilde{\delta}_{23}(k)\re^{t\Phi_{23}}\\
     \tilde{r}_{2,a}(-k) \frac{\tilde{\delta}_{21}(k)}{\tilde{\delta}_{23}(k)}\re^{-t\Phi_{13}} & 0 & 1\\
     \end{pmatrix},\nonumber\\ 
     &R_6(\zeta,k)=\begin{pmatrix}
     	1 &-\hat{r}_{1,a}(k)\tilde{\delta}^{-1}_{21}(k)\re^{t\Phi_{12}} & -\tilde{r}_{2,a}^*(-k)\frac{\tilde{\delta}_{23}(k)}{\tilde{\delta}_{21}(k)}\re^{t\Phi_{13}}\\
     	0 & 1 & 0\\
     	0 & \hat \alpha_a^*(k)\tilde{\delta}_{23}^{-1}(k)\re^{-t\Phi_{23}} & 1\\
     \end{pmatrix},\label{R1-8}\\
     & R_7(\zeta,k)=\begin{pmatrix}
     	1 & 0 & 0\\
     	0  & 1  & 0\\
     	-\tilde{r}_{2,a}^*(k)\frac{\tilde{\delta}_{21}(k)}{\tilde{\delta}_{23}(k)}\re^{-t\Phi_{13} }   & -\tilde{\alpha}_a^*(k)\tilde{\delta}_{23}^{-1}(k)\re^{-t\Phi_{23} } & 1\\
     \end{pmatrix},\nonumber\\
     &R_8(\zeta,k)=\begin{pmatrix}
     1 & 0  & 0\\
     -2\sigma  \tilde\alpha_{a}(-k)\tilde{\delta}_{21}(k)\re^{-t\Phi_{12} }   & 1 & 0\\
     \tilde{r}_{2,a}(-k)\frac{\tilde{\delta}_{21}(k)}{\tilde{\delta}_{23}(k)}\re^{-t\Phi_{13} }  & 0  & 1\\
     \end{pmatrix},\nonumber\\
     & R_9(\zeta,k)= \begin{pmatrix}
     	1 &-\tilde{\alpha}_a^*(-k)\tilde{\delta}^{-1}_{21}(k)\re^{t\Phi_{12}} & -\tilde{r}_{2,a}^*(-k)\frac{\tilde{\delta}_{23}(k)}{\tilde{\delta}_{21}(k)}\re^{t\Phi_{13}}\\
     	0 & 1 & 0\\
     	0 & 0 & 1\\
     \end{pmatrix},\nonumber\\
     &R_{10}(\zeta,k)=\begin{pmatrix}
     1 & 0  & \tilde{r}_{2,a}(k)\frac{\tilde{\delta}_{23}(k)}{\tilde{\delta}_{21}(k)}\re^{t\Phi_{13} } \\
     0 & 1  & -2\sigma  \tilde\alpha_a(k)\tilde{\delta}_{23}(k)\re^{t\Phi_{23} } \\
     0  & 0 & 1\\
     \end{pmatrix}.\nonumber
     \end{align}

     \begin{remark}\label{remark_relation_hatalpha}
     	The reflection coefficients satisfy the following relations:
     	\begin{align*}
     		&\tilde{r}_2(-k)\hat{\alpha}(k)+\tilde{\alpha}(-k)-\hat{r}_1^*(k)=0,\\
     		&2\sigma \tilde{\alpha}^*(k)\tilde{\alpha}(-k)-\tilde{r}_2^*(k)-\tilde{r}_2(-k)=0,
     	\end{align*} 
     	and it is precisely by using these relations that we simplified and reformulated the above matrix functions \eqref{R1-8}.
     \end{remark}

   \begin{lem}\label{lem_R_bound}
  	$R(\zeta,k)$ is uniformly bounded for $\zeta\in\mathcal{I}_3$, $k\in \mathbb{C}\setminus\Gamma^{(2)}$, and $t>0$. Moreover,
  	\begin{equation*}
  		R(k)=I+\mathcal{O}(k^{-1}),\quad k\to\infty.
  	\end{equation*}
  \end{lem}

    Next, we set $	m^{(2)}(x,t,k)=m^{(1)}(x,t,k)R(x,t,k),$ and require that $m_+^{(2)}(x,t,k)=m_-^{(2)}(x,t,k)v^{(2)}(x,t,k)$ for $k\in \Gamma^{(2)}$, where $v^{(2)}_j(x,t,k)$ for $k\in \Gamma^{(2)}_j$, $j=1,2,\cdots,8$, are as follows
  	\begin{align*}
  		&v^{(2)}_1(x,t,k)=\begin{pmatrix}
  			1 & -\hat \alpha_a^*(-k)\tilde\delta_{21}^{-1}(k)\re^{t\Phi_{12}(\zeta,k)}  & 0\\
  			0  & 1 & 0\\
  			0 & 0 & 1\\
  		\end{pmatrix},\quad &&v^{(2)}_2(x,t,k)=\begin{pmatrix}
  		1 & 0 & 0\\
  		-2\sigma \tilde \alpha_a(-k)\tilde\delta_{21}(k)\re^{-t\Phi_{12}(\zeta,k)} & 1 & 0\\
  		0 & 0 & 1\\
  		\end{pmatrix},\nonumber\\
  		&v^{(2)}_3(x,t,k)=\begin{pmatrix}
  			1 & \tilde \alpha_a^*(-k)\tilde\delta_{21}^{-1}(k)\re^{t\Phi_{12}(\zeta,k)} & 0\\
  			0 & 1 & 0\\
  			0 & 0 & 1\\
  		\end{pmatrix},\quad &&v^{(2)}_4(x,t,k)=\begin{pmatrix}
  		1 & 0  & 0\\
  		2\sigma  \hat \alpha_a(-k)\tilde\delta_{21}(k)\re^{-t\Phi_{12}(\zeta,k)}  & 1  & 0\\
  		0  & 0  & 1\\
  		\end{pmatrix},\\
  		&v^{(2)}_5(x,t,k)=\begin{pmatrix}
  			1 & 0 & 0\\
  			0 & 1 & 0\\
  			0 &  \tilde \alpha_a^*(k)\tilde\delta_{23}^{-1}(k)\re^{-t\Phi_{23}(\zeta,k)} & 1\\
  		\end{pmatrix},\quad &&v^{(2)}_6(x,t,k)=\begin{pmatrix}
  		1 & 0  & 0\\
  		0 & 1 &  2 \sigma  \hat \alpha_a(k)\tilde\delta_{23} (k)\re^{t\Phi_{23}(\zeta,k)}\\
  		0 & 0  & 1\\
  		\end{pmatrix},\nonumber\\
  		&v^{(2)}_7(x,t,k)=\begin{pmatrix}
  			1 & 0 & 0\\
  			0 & 1 & 0\\
  			0 & - \hat\alpha_a^*(k)\tilde\delta_{23}^{-1}(k)\re^{-t\Phi_{23}(\zeta,k)} & 1\\
  		\end{pmatrix},\quad &&v^{(2)}_8(x,t,k)=\begin{pmatrix}
  		1 & 0 & 0\\
  		0 & 1 & - 2 \sigma  \tilde \alpha_a(k)\tilde\delta_{23} (k)\re^{t\Phi_{23}(\zeta,k)}\\
  		0  & 0  & 1\\
  		\end{pmatrix}.\nonumber
  	\end{align*}
  
  When $k\in\Gamma^{(2)}_j$, $j=9,10,\cdots,14$, the corresponding jump matrices can be expressed as:
  \begin{equation*}
  	v^{(2)}_9(k)=R_6^{-1}(k)v^{(1)}_1(k)R_5(k),\quad
  	v^{(2)}_{10}(k)=R_3^{-1}(k)v^{(1)}_2(k)R_2(k),\quad v^{(2)}_{11}(k)=R_4^{-1}(k)v^{(1)}_3(k)R_1(k),
  \end{equation*}
  \begin{equation*}
  	v^{(2)}_{12}(k)=R^{-1}_7(k)R_8(k),\qquad v^{(2)}_{13}(k)=R^{-1}_9(k)R_{10}(k).
  \end{equation*}

  \begin{lem}\label{lem_V2_error}
  	The jump matrix $v^{(2)}(x,t,k)$ converges to the identity matrix $I$ uniformly for $\zeta\in\mathcal{I}_3$ and $k\in \Sigma^{(2)}$ except near the two critical points $\pm k_0$ as $t\to\infty$. Additionally, the following estimates are established:
  	\begin{align*}
  		&\left\|(1+|\cdot|)\left(v^{(2)}_j(x,t,\cdot)-I\right)\right\|_{(L^1\cap L^\infty)\left(\Gamma^{(2)}_j\right)}\leq Ct^{-N}, \quad  &&j=9,10,11,\\
  		&\left\|(1+|\cdot|)\left(v^{(2)}_j(x,t,\cdot)-I\right)\right\|_{(L^1\cap L^\infty)\left(\Gamma^{(2)}_j\right)}\leq C\re^{-ct}, \quad &&j=12,13.
  	\end{align*}
  \end{lem}
  \begin{proof}
  	The proof follows a similar argument as in Lemma \ref{lem_v2_error}.
  \end{proof}

 For $k\in B_\epsilon(k_0)\setminus\left(k_0-\epsilon,k_0+\epsilon \right) $,
 \begin{align*}
 	\tilde{\delta}_{23}(k)&=(2\sqrt{t})^{-2\ri\hat\nu} \re^{\ri(2\nu_1-\nu_2)\ln_\pi(z_{(+)})}\re^{\ri(2\nu_3-\nu_2)\ln_0(z_{(+)})}\re^{\ri(\nu_3-2\nu_2)\ln_0(2k_0)}\re^{-2\chi_1(k_0)+\chi_2(k_0)+2\chi_3(k_0)-\chi_3(-k_0)}\\
 	&\quad\re^{\ri(\nu_3-2\nu_2)\ln_0((k+k_0)/(2k_0))}\re^{-2(\chi_1(k)-\chi_1(k_0))+(\chi_2(k)-\chi_2(k_0))+2(\chi_3(k)-\chi_3(k_0))+(\chi_3(-k)-\chi_3(-k_0))}\delta_1(-k)\delta_2(-k)\\
 	&:=\re^{\ri(2\nu_1-\nu_2)\ln_\pi(z_{(+)})}\re^{\ri(2\nu_3-\nu_2)\ln_0(z_{(+)})}b^{(k_0)}_0(\zeta)b^{(k_0)}_1(\zeta,k),
 \end{align*}
 where $\hat{\nu}=\nu_1-\nu_2+\nu_3$, and
 \begin{align*}
 	&b^{(k_0)}_0(\zeta)=(2\sqrt{t})^{-2\ri\hat\nu}\re^{\ri(\nu_3-2\nu_2)\ln_0(2k_0)}\re^{-2\chi_1(k_0)+\chi_2(k_0)+2\chi_3(k_0)+\chi_3(-k_0)}\delta_1(-k_0)\delta_2(-k_0),\\
 	&b^{(k_0)}_1(\zeta,k)=\re^{\ri(\nu_3-2\nu_2)\ln_0((k+k_0)/(2k_0))}\re^{-2(\chi_1(k)-\chi_1(k_0))+(\chi_2(k)-\chi_2(k_0))+2(\chi_3(k)-\chi_3(k_0))+(\chi_3(-k)-\chi_3(-k_0))}\\
    &\qquad\qquad\quad\delta_1^{-1}(-k_0)\delta^{-1}_2(-k_0)\delta_1(-k)\delta_2(-k).
 \end{align*}
  And for $k\in B_\epsilon(-k_0)\setminus\left(-k_0-\epsilon,-k_0+\epsilon\right)$,
  \begin{align*}
  	\tilde{\delta}_{21}(k)&=(2\sqrt{t})^{-2\ri\hat\nu} \re^{\ri(2\nu_1-\nu_2)\ln_0(z_{(-)})}\re^{\ri(2\nu_3-\nu_2)\ln_\pi(z_{(-)})}\re^{\ri(\nu_3-2\nu_2)\ln_\pi(-2k_0)}\re^{-2\chi_1(k_0)+\chi_2(k_0)+2\chi_3(k_0)+\chi_3(-k_0)}\\
  	&\quad\re^{\ri(\nu_3-2\nu_2)\ln_\pi((k-k_0)/(-2k_0))}\re^{-2(\chi_1(-k)-\chi_1(k_0))+(\chi_2(-k)-\chi_2(k_0))+2(\chi_3(-k)-\chi_3(k_0))+(\chi_3(k)-\chi_3(-k_0))}\delta_1(k)\delta_2(k)\\
  	&:=\re^{\ri(2\nu_1-\nu_2)\ln_0(z_{(-)})}\re^{\ri(2\nu_3-\nu_2)\ln_\pi(z_{(-)})}b^{(-k_0)}_0(\zeta)b^{(-k_0)}_1(\zeta,k),
  \end{align*}
  where
  \begin{align*}
  	&b^{(-k_0)}_0(\zeta)=(2\sqrt{t})^{-2\ri\hat\nu} \re^{\ri(\nu_3-2\nu_2)\ln_\pi(-2k_0)}\re^{-2\chi_1(k_0)+\chi_2(k_0)+2\chi_3(k_0)+\chi_3(-k_0)}\delta_1(-k_0)\delta_2(-k_0),\\
  	&b^{(-k_0)}_1(\zeta,k)=\re^{\ri(\nu_3-2\nu_2)\ln_\pi((k-k_0)/(-2k_0))}\re^{-2(\chi_1(-k)-\chi_1(k_0))+(\chi_2(-k)-\chi_2(k_0))+2(\chi_3(-k)-\chi_3(k_0))+(\chi_3(k)-\chi_3(-k_0))}\\
    &\qquad\qquad\quad\delta_1^{-1}(-k_0)\delta^{-1}_2(-k_0)\delta_1(k)\delta_2(k).
  \end{align*}

 Based on the expansions above near $k_0$ and $-k_0$, construct the following two matrices for $\zeta\in\mathcal{I}_3$: 
 \begin{equation*}
     \tilde{m}^{(k_0)}(\zeta,k)=m^{(2)}(\zeta,k)X_{(+)}(\zeta),\qquad\tilde{m}^{(-k_0)}(\zeta,k)=m^{(2)}(\zeta,k)X_{(-)}(\zeta),
 \end{equation*}
 where
 \begin{align*}
 	&X_{(+)}(\zeta)=\begin{pmatrix}
 		1 & 0 & 0\\
 		0 & \left(b^{(k_0)}_{0}(\zeta)\right)^{1/2}\re^{\frac{t}{2}\Phi_{23}(\zeta,k_0)} & 0\\
 		0 & 0 & \left(b^{(k_0)}_{0}(\zeta)\right)^{-1/2}\re^{-\frac{t}{2}\Phi_{23}(\zeta,k_0)}
 	\end{pmatrix},\quad\qquad k\in B_\epsilon(k_0),\\
 		&X_{(-)}(\zeta)=\begin{pmatrix}
 		\left({b_{0}^{(-k_0)}(\zeta)}\right)^{-1/2}\re^{\frac{t}{2}\Phi_{12}(\zeta,-k_0)}  & 0 & 0\\
 		0 & \left({b_{0}^{(-k_0)}(\zeta)}\right)^{1/2}\re^{-\frac{t}{2}\Phi_{12}(\zeta,-k_0)} & 0\\
 		0 & 0 & 1
 	\end{pmatrix},\quad k\in B_\epsilon(-k_0) .
 \end{align*}
  Hence, the jump matrices $\tilde v^{(\pm k_0)}_j$ of the function $\tilde{m}^{(\pm k_0)}(\zeta,k)$ for $k\in X^{- k_0}_j$, $j=1,2,3,4$, and $k\in X^{ k_0}_j$, $j=5,6,7,8$, are as follows:
 \begin{align*}
 	&\tilde v^{(-k_0)}_1(\zeta,z_{(-)})=\begin{pmatrix}
 		1 & -\hat \alpha_a^*(-k)\left(b^{(-k_0)}_1\right)^{-1}z_{(-)}^{-2\ri \hat{\nu}}\re^{\frac{\ri z_{(-)}^2}{2}}  & 0\\
 		0  & 1 & 0\\
 		0 & 0 & 1\\
 	\end{pmatrix},\\
 	&\tilde v^{(-k_0)}_2(\zeta,z_{(-)})=\begin{pmatrix}
 		1 & 0 & 0\\
 		-2\sigma \tilde \alpha_a(-k)b^{(-k_0)}_1z_{(-)}^{2\ri \hat{\nu}}\re^{-\frac{\ri z_{(-)}^2}{2}} & 1 & 0\\
 		0 & 0 & 1\\
 	\end{pmatrix},\\
 	&\tilde v^{(-k_0)}_3(\zeta,z_{(-)})=\begin{pmatrix}
 		1 & \tilde \alpha_a^*(-k)\left(b^{(-k_0)}_1\right)^{-1}z_{(-)}^{-2\ri \hat{\nu}}\re^{\frac{\ri z_{(-)}^2}{2}} & 0\\
 		0 & 1 & 0\\
 		0 & 0 & 1\\
 	\end{pmatrix},\\
 	&\tilde v^{(-k_0)}_4(\zeta,z_{(-)})=\begin{pmatrix}
 		1 & 0  & 0\\
 		2\sigma \hat \alpha_a(-k)b^{(-k_0)}_1z_{(-)}^{2\ri \hat{\nu}} \re^{-\frac{\ri z_{(-)}^2}{2}}  & 1  & 0\\
 		0  & 0  & 1\\
 	\end{pmatrix},\\
 	&\tilde v^{(k_0)}_5(\zeta,z_{(+)})=\begin{pmatrix}
 		1 & 0 & 0\\
 		0 & 1 & 0\\
 		0 & \tilde \alpha_a^*(k)\left({b^{(k_0)}_1}\right)^{-1}z_{(+)}^{-2\ri \hat{\nu}}\re^{\frac{\ri z_{(+)}^2}{2}} & 1\\
 	\end{pmatrix},\\
 	&\tilde v^{(k_0)}_6(\zeta,z_{(+)})=\begin{pmatrix}
 		1 & 0  & 0\\
 		0 & 1 & 2\sigma \hat \alpha_a(k){b^{(k_0)}_1}z_{(+)}^{2\ri \hat{\nu}} \re^{-\frac{\ri z_{(+)}^2}{2}}\\
 		0 & 0  & 1\\
 	\end{pmatrix},\\
 	&\tilde v^{(k_0)}_7(\zeta,z_{(+)})=\begin{pmatrix}
 		1 & 0 & 0\\
 		0 & 1 & 0\\
 		0 & -\hat\alpha_a^*(k)\left({b^{(k_0)}_1}\right)^{-1}z_{(+)}^{-2\ri \hat{\nu}}\re^{\frac{\ri z_{(+)}^2}{2}} & 1\\
 	\end{pmatrix},\\
 	&\tilde v^{(k_0)}_8(\zeta,z_{(+)})=\begin{pmatrix}
 		1 & 0 & 0\\
 		0 & 1 & -2\sigma \tilde \alpha_a(k){b^{(k_0)}_1}z_{(+)}^{2\ri \hat{\nu}} \re^{-\frac{\ri z_{(+)}^2}{2}}\\
 		0  & 0  & 1\\
 	\end{pmatrix},
 \end{align*}
 where $z_{(+)}^{2\ri \hat{\nu}}=\re^{\ri(2\nu_1-\nu_2)\ln_\pi(z_{(+)})}\re^{\ri(2\nu_3-\nu_2)\ln_0(z_{(+)})}$ and $z_{(-)}^{2\ri \hat{\nu}}=\re^{\ri(2\nu_1-\nu_2)\ln_0(z_{(-)})}\re^{\ri(2\nu_3-\nu_2)\ln_\pi(z_{(-)})}$.
 
  \begin{lem}\label{lem_X12_bound}
  	For $\zeta\in\mathcal{I}_3$ and $t\geq2$, the functions $X_{(\pm)}(\zeta)$ are uniformly bounded, i.e.,
  	\begin{equation*}
  		\sup_{\zeta\in\mathcal{I}_3}\sup_{t\geq2}\left|X_{(\pm)}(\zeta)\right|<C,\quad 	\sup_{\zeta\in\mathcal{I}_3}\sup_{t\geq2}\left|X_{(\pm)}^{-1}(\zeta)\right|<C.
  	\end{equation*}
  	The functions $b_0^{(\pm k_0)}(\zeta)$ and $b_1^{(\pm k_0)}(\zeta,k)$ satisfy
  	\begin{align*}
  		&\left|b_0^{(k_0)}(\zeta)\right|=\re^{\pi(2\nu_1+\nu_2-\nu_3)},  \quad &&\left|b_1^{(k_0)}(\zeta,k)-1\right|\leq C|k-k_0|(1+|\ln|k-k_0||),\\
  		&\left|b_0^{(-k_0)}(\zeta)\right|=\re^{\pi(2\nu_1-\nu_2)}, \quad&&\left|b_1^{(-k_0)}(\zeta,k)-1\right|\leq C|k+k_0|(1+|\ln|k+k_0||).
  	\end{align*}
  	
  \end{lem}
  \begin{proof}
  	Direct calculation yields
  	\begin{align*}
  		&\rre\chi_1(\zeta,k_0)
  		=-\pi\nu_1,\quad \rre\chi_2(\zeta,k_0)
  		=-\pi\nu_2,\quad \rre\chi_3(k_0)=\rre\chi_3(-k_0)=0.
  	\end{align*}
  	The remaining estimates can be established analogously by referring to Lemmas \ref{lem_delta} and \ref{lem_delta1}.
  \end{proof}

 The following matrix functions are introduced:
  	\begin{align*}
  		&v^{X}_{- k_0}(\zeta,z_{(-)})=\begin{pmatrix}
  			1 & -\hat \alpha^*(k_0)z_{(-)}^{-2\ri\hat\nu}\re^{\frac{\ri z_{(-)}^2}{2}}  & 0\\
  			0  & 1 & 0\\
  			0 & 0 & 1\\
  		\end{pmatrix},\quad z_{(-)}\in X_1,\\ 
  		&v^{X}_{-k_0}(\zeta,z_{(-)})=\begin{pmatrix}
  			1 & 0 & 0\\
  			-2\sigma \tilde \alpha(k_0)z_{(-)}^{2\ri\hat\nu}\re^{-\frac{\ri z_{(-)}^2}{2}} & 1 & 0\\
  			0 & 0 & 1\\
  		\end{pmatrix},\quad  z_{(-)}\in X_2,\\
  		&v^{X}_{ -k_0}(\zeta,z_{(-)})=\begin{pmatrix}
  			1 & \tilde \alpha^*(k_0)z_{(-)}^{-2\ri\hat\nu}\re^{\frac{\ri z_{(-)}^2}{2}} & 0\\
  			0 & 1 & 0\\
  			0 & 0 & 1\\
  		\end{pmatrix},\quad\,\,\,\,  z_{(-)}\in X_3,\\
  		&v^{X}_{ -k_0}(\zeta,z_{(-)})=\begin{pmatrix}
  			1 & 0  & 0\\
  			2\sigma \hat \alpha(k_0)z_{(-)}^{2\ri\hat\nu}\re^{-\frac{\ri z_{(-)}^2}{2}}  & 1  & 0\\
  			0  & 0  & 1\\
  		\end{pmatrix},\quad  \,\,\,\, z_{(-)}\in X_4,
  	\end{align*}
  	\begin{align*}
  		&v^{X}_{k_0}(\zeta,z_{(+)})=\begin{pmatrix}
  			1 & 0 & 0\\
  			0 & 1 & 0\\
  			0 & \tilde \alpha^*(k_0)z_{(+)}^{-2\ri\hat\nu}\re^{\frac{\ri z_{(+)}^2}{2}} & 1\\
  		\end{pmatrix},\quad\,\,\,\, z_{(+)}\in X_1,\\  
  		&v^{X}_{ k_0}(\zeta,z_{(+)})=\begin{pmatrix}
  			1 & 0  & 0\\
  			0 & 1 & 2\sigma \hat \alpha(k_0)z_{(+)}^{2\ri\hat\nu}\re^{\frac{\ri z_{(+)}^2}{2}}\\
  			0 & 0  & 1\\
  		\end{pmatrix},\qquad\,\,\,  z_{(+)}\in X_2,\\
  		&v^{X}_{k_0}(\zeta,z_{(+)})=\begin{pmatrix}
  			1 & 0 & 0\\
  			0 & 1 & 0\\
  			0 & -\hat\alpha^*(k_0)z_{(+)}^{-2\ri\hat\nu}\re^{\frac{\ri z_{(+)}^2}{2}} & 1\\
  		\end{pmatrix},\quad  z_{(+)}\in X_3,\\
  		&v^{X}_{k_0}(\zeta,z_{(+)})=\begin{pmatrix}
  			1 & 0 & 0\\
  			0 & 1 & -2\sigma \tilde \alpha(k_0)z_{(+)}^{2\ri\hat\nu}\re^{-\frac{\ri z_{(+)}^2}{2}}\\
  			0  & 0  & 1\\
  		\end{pmatrix},\quad\,\,\,  z_{(+)}\in X_4.
  	\end{align*}
 Define functions $m^{(\pm k_0)}$ in $B_\epsilon(\pm k_0)$, respectively, to approximate the function  $m^{(2)}(x,t,k)$:
  \begin{align}\label{mk_0}
  	m^{(\pm k_0)}(\zeta,k)=X_{(\pm)}(\zeta)m^{X}_{\pm k_0}(q(\zeta),z_{(\pm)}(k))X^{-1}_{(\pm)}(\zeta), \quad k\in B_\epsilon(\pm k_0),
  \end{align}
  where $m^{X}_{\pm k_0}(q(\zeta),z_{(\pm)}(k))$ satisfy the RH problem \ref{model-rhp1}, except that the jump matrices $V^X_{\pm k_0}$ are replaced by $v^X_{\pm k_0}$.
 
  The solution $m^{X}_{\pm k_0}(q,z_{(\pm)})$ admits the following expansion:
  \begin{equation*}
  	m^{X}_{\pm k_0}(\zeta ,z_{(\pm)})=I+\frac{\left( m^{X}_{\pm k_0}(\zeta)\right) _1}{z_{(\pm)}}+\mathcal{O}\left(\frac{1}{z_{(\pm)}^2}\right), \quad z_{(\pm)}\to\infty,
  \end{equation*}
  where
  \begin{equation*}
  	\left( m^{X}_{k_0}(\zeta)\right) _1=\begin{pmatrix}
  		0 & 0 & 0\\
  		0 & 0 & \tilde\alpha_{23}\\
  		0 & \tilde\alpha_{32} & 0\\
  	\end{pmatrix},\quad 	\left( M^{X}_{-k_0}(\zeta)\right) _1=\begin{pmatrix}
  		0 & \tilde\alpha_{12} & 0\\
  		\tilde\alpha_{21} & 0 & 0\\
  		0 & 0 & 0\\
  	\end{pmatrix},
  \end{equation*}
  and 
  \begin{equation*}
  \begin{aligned}
      &\tilde\alpha_{12}=\frac{\sqrt{2\pi}\re^{-\frac{\pi\ri}{4}}\re^{-\frac{\pi\hat\nu}{2}}}{2\sigma  \alpha(k_0)\re^{2\pi(\nu_2-\nu_1)}\Gamma(\ri\hat\nu)},\quad 	&&\tilde\alpha_{21}=\frac{\sqrt{2\pi}\re^{\frac{\pi\ri}{4}}\re^{-\frac{\pi\hat\nu}{2}}}{  \alpha^*(k_0)\re^{2\pi\nu_1}\Gamma(-\ri\hat\nu)},\\
      &\tilde\alpha_{23}=\frac{\sqrt{2\pi}\re^{-\frac{\pi\ri}{4}}\re^{-\frac{5\pi\hat\nu}{2}}}{\alpha^*(k_0)\re^{2\pi\nu_1}\Gamma(-\ri\hat\nu)},\quad
  	&&\tilde\alpha_{32}=\frac{\sqrt{2\pi}\re^{\frac{\pi\ri}{4}}\re^{\frac{3\pi\hat\nu}{2}}}{2\sigma \alpha(k_0)\re^{2\pi(\nu_2-\nu_1)}\Gamma(\ri\hat\nu)}.
  \end{aligned}
  \end{equation*}

  \begin{lem}\label{lem_Error_V2-vk1}
  	For each $\zeta\in\mathcal{I}_3$ and $t\geq 2$, the eigenfunctions $m^{( \pm k_0)}$  defined in equation \eqref{mk_0} are analytic and bounded for $k\in B_\epsilon(\pm k_0)\setminus X^{\pm k_0}$, respectively. The jump relations $m^{(\pm k_0)}_+(x,t,k)=m^{(\pm k_0)}_-(x,t,k)v^{(\pm k_0)}(x,t,k)$ for $k\in X^{\pm k_0}$ hold. Moreover, the jump matrices $v^{(\pm k_0)}$ satisfy the following estimates:
  	\begin{align*}
  		&\left\|v^{(2)}(x,t,\cdot)-v^{(\pm k_0)}(x,t,\cdot)\right\|_{L^1(X^{\pm k_0})}\leq \frac{C \ln t}{t},\\
  		&\left\|v^{(2)}(x,t,\cdot)-v^{(\pm k_0)}(x,t,\cdot)\right\|_{L^\infty(X^{\pm k_0})}\leq \frac{C \ln t}{\sqrt{t}}.
  	\end{align*}
  	In addition,
  	\begin{equation*}
  		\left\|\left(m^{(\pm k_0)}(x,t,\cdot)\right)^{-1}-I\right\|_{L^\infty(\partial B_\epsilon(\pm k_0))}=\mathcal{O}\left(t^{-1/2}\right),
  	\end{equation*}
  	\begin{equation*}
  		\frac{1}{2\pi\ri} \int_{\partial B_\epsilon(\pm k_0)} \left(\left(m^{(\pm k_0)}(x,t,k)\right)^{-1}-I\right)\rd k=-\frac{X_{(\pm)}(\zeta)\left( m^{X}_{\pm k_0}(\zeta)\right) _1X_{(\pm)}^{-1}(\zeta)}{2\sqrt{t}}+\mathcal{O}\left(t^{-1}\right),
  	\end{equation*}
  	uniformly for $\zeta\in\mathcal{I}_3$.
  \end{lem}
  The proof follows the same reasoning as Lemma \ref{lem_Error_v2-vk1} and is omitted here.
  
  Define $\hat{\Gamma}=\Gamma^{(2)}\cup \partial B_{\epsilon}$, and $\partial B_{\epsilon}$ are oriented counterclockwise. Construct the following piecewise analytic function $\hat m(x,t,k)$ and its corresponding jump matrix $\hat{v}(x,t,k)$ for $k\in\hat{\Gamma}$:
  \begin{equation}\label{hatm}
  	\hat{m}=\left\{
  	\begin{aligned}
  		&m^{(2)}\left(m^{(\pm k_0)}\right)^{-1},\quad &&k\in B_\epsilon(\pm k_0),\\
  		&m^{(2)},\quad &&\text{elsewhere}.
  	\end{aligned}
  	\right.\qquad 	\hat{v}=\left\{
  	\begin{aligned}
  		&m^{(\pm k_0)}_-v^{(2)}\left(m^{(\pm k_0)}_+\right)^{-1},\quad &&k\in \hat{\Gamma} \cap B_\epsilon(\pm k_0),\\
  		& \left(m^{(\pm k_0)}\right)^{-1}, && k\in \partial B_\epsilon(\pm k_0),\\
  		&v^{(2)},\quad &&k\in \hat{\Gamma} \setminus B_\epsilon.
  	\end{aligned}
  	\right.
  \end{equation}

  \begin{lem}\label{lem_hatW}
  	Let $\hat{w}=\hat{v}-I$, the following estimates hold uniformly for $\zeta\in\mathcal{I}_3$, $t\geq2$:
  	\begin{align*}
  		&\left\|(1+|\cdot|) \hat{w}\right\|_{\left(L^1\cap L^\infty\right)(\mathbb{R})}\leq \frac{C}{t^{N}}, \quad \left\|(1+|\cdot|)  \hat{w}\right\|_{\left(L^1\cap L^\infty\right)(\Gamma^{(2)}\setminus(\mathbb{R}\cup B_\epsilon))}\leq C\re^{-ct},\label{Error_hatW_2}\\
  		&\left\| \hat{w}\right\|_{\left(L^1\cap L^\infty\right)(\partial B_\epsilon(\pm k_0))}\leq \frac{C}{\sqrt{t}},\quad
  		\left\| \hat{w}\right\|_{L^1(X^{\pm k_0})}\leq \frac{C\ln t}{t},\quad 
  		\left\|  \hat{w}\right\|_{L^\infty(X^{\pm k_0})}\leq \frac{C\ln t}{\sqrt{t}}.
  	\end{align*}
  \end{lem}
  \begin{proof}
  	This follows directly from Lemmas \ref{lem_r_decomposition}, \ref{lem_delta}, and \ref{lem_V2_error}-\ref{lem_Error_V2-vk1}.
  \end{proof}
  
  From the above lemma, we obtain
  \begin{equation*}
  	\left\|(1+|\cdot|)  \hat{w}\right\|_{L^p(\hat{\Gamma})}\leq \frac{C\ln t^{\frac{p-1}{p}}}{\sqrt{t}},\quad 1\leq p\leq \infty,
  \end{equation*}
  and $\hat{w}\in (\dot{L}^3\cap L^\infty)(\hat{\Gamma})$.
  \begin{lem}\label{lem_I-C_inverse1}
  	For any sufficiently large $t$ and $\zeta\in\mathcal{I}_3$, $I-\mathcal{C}_{\hat{w}(x,t,\cdot)}\in \mathcal{B}(\dot L^3(\hat{\Gamma})) $ is invertible.
  \end{lem}
 
   Define a function  $\tilde{\mu}(x,t,k)$ for  sufficiently large $t$:
  \begin{equation*}
  	\tilde{\mu}=I+\left(I-\mathcal{C}_{\hat{w}}\right)^{-1}\mathcal{C}_{\hat{w}}I\in I+\dot L^3(\hat{\Gamma}),
  \end{equation*}
 and obtain the following two lemmas.
  
  \begin{lem}\label{lem_hatm_solution}
  	For sufficiently large time $t$, the RH problem for $\hat{m}$ in \eqref{hatm} has a unique solution $\hat{m}\in I+\dot{E}^3(\mathbb{C} \setminus \hat{\Gamma})$, which
  	\begin{equation*}
  		\hat{m}(x,t,k)=I+\mathcal{C}(\tilde\mu \hat w)=I+\frac{1}{2\pi\ri}\int_{\hat{\Gamma}}\frac{\tilde{\mu}(x,t,k')\hat{w}(x,t,k')}{k'-k}\rd k'.
  	\end{equation*}
  \end{lem}
  \begin{lem}\label{Error_tildemu-I}
  	For sufficiently large $t$ and $\zeta\in\mathcal{I}_3$, we have
  	\begin{equation*}
  		\left\|\tilde{\mu}-I\right\|_{L^p(\hat{\Gamma})}\leq \frac{C\ln t^{\frac{p-1}{p}}}{\sqrt{t}}, \quad 1<p<\infty.
  	\end{equation*}
  \end{lem}

  Consider the following nontangential limit and decompose it using Lemmas \ref{lem_hatW} and \ref{Error_tildemu-I}, we obtain:
  \begin{align*}
  	\tilde Q(x,t)&=\lim_{k\to\infty}k\left(\hat{m}(x,t,k)-I\right)=-\frac{1}{2\pi\ri}\int_{\hat{\Gamma}}\tilde{\mu}(x,t,k)\hat{w}(x,t,k)\rd k\\
    &=-\frac{1}{2\pi\ri}\int_{\partial B_\epsilon }\left(\left(m^{(\pm k_0)}\right)^{-1}-I\right)\rd k+\mathcal{O}\left(\frac{\ln t}{t}\right)\\
  	&=\frac{X_{(+)}(\zeta)\left( m^{X}_{ k_0}(q(\zeta))\right) _1X_{(+)}^{-1}(\zeta)}{2\sqrt{t}}+\frac{X_{(-)}(\zeta)\left( m^{X}_{- k_0}(q(\zeta))\right) _1X_{(-)}^{-1}(\zeta)}{2\sqrt{t}}+\mathcal{O}\left(\frac{\ln t}{t}\right).
  \end{align*}
  
  Considering the transformations made in this subsection, together with Theorem \ref{theo_reconstruction formula} and Lemmas \ref{lem_delta}, \ref{lem_delta1} and \ref{lem_R_bound}, we obtain the long-time asymptotic behavior in \eqref{theo_region3} of the solution to the initial value problem \eqref{Newell_initial} for the Newell equation in Region III (i.e., $\zeta \in \mathcal{I}_3$) for the relation $\hat\nu=\nu$.

     \subsection{Long-time asymptotics in Region {\rm{IV}}}\label{subsec_region4}
    The difference between this region and Region I is that we consider the asymptotic behavior of the solution as $x \to -\infty$ under the parameter $\tau$. The specific procedure is similar to the discussion in Region III. The asymptotic result of the solution is given by equation \eqref{theo_region4}.
	
	\begin{remark}
		Since $k_0\to-\infty$ as $\tau \to 0$,  $r_1(k)$ and $r_2(k)$ vanish to all orders at $k=k_0$, it follows that the functions $\nu$, $\nu_1$, $ \tilde\nu$, $ \hat\nu$ and $s_2$ vanish to all orders as $\tau \to 0$. Consequently, equation \eqref{theo_region4} implies that as $x\to-\infty$:
		\begin{equation*}
			\begin{aligned}
				r(x,t)&=\mathcal{O}\left(\frac{1}{|x|^{N}}+\frac{C_N(\tau)}{|x|}\right),\qquad
				q(x,t)=\mathcal{O}\left(\frac{1}{|x|^{N}}+\frac{C_N(\tau)}{|x|}\right),
			\end{aligned}
		\end{equation*}
		uniformly for $\tau\in\mathcal{I}_4$. In particular, for any fixed $t\geq 0$, the above expression can be reduced to $\mathcal{O}\left({|x|^{-N}}\right) $ as $x \to -\infty$.
	\end{remark}

	\noindent{\bf Conflict of interest declaration.} 
    
    We declare we have no competing interests.
    
	\subsection*{Acknowledgements}
	Support is acknowledged from the National Natural Science Foundation of China, Grant No. 12371247 and No. 12431008 and Beijing Natural Science Foundation Grant No. 1262012.
	
	\bibliographystyle{amsplain}

\end{document}